\documentclass[nofootinbib, aps, prx, 11pt,eqsecnum, notitlepage,reprint]{revtex4-1}

\usepackage{amsmath,amsthm,amssymb,amsfonts,setspace}
\usepackage{bm}
\usepackage{float}
\usepackage{color}
\usepackage[usenames,dvipsnames]{xcolor}
\usepackage[colorlinks]{hyperref}
\usepackage{mathtools}
\usepackage{cleveref} % must be loaded after hyperref

\usepackage[normalem]{ulem} %temporary package for editing

\usepackage{graphicx}
\graphicspath{ {./figs/}} %changes the path for figures

\newcommand{\nc}{\newcommand}
\nc{\rnc}{\renewcommand}
\nc\mnb[1]{\medskip\noindent{\bf #1}}

\newcommand{\ket}[1]{\left| #1 \right>} % for Dirac bras
\newcommand{\bra}[1]{\left< #1 \right|} % for Dirac kets
\newcommand{\x}{\mathbf{x}} %element of Rd
\newcommand{\y}{\mathbf{y}} %element of Rd
\renewcommand{\k}{\mathbf{k}}  %element of Zd
\newcommand{\n}{\mathbf{n}}  %element of Zd
\newcommand{\bphi}{\text{\boldmath$\varphi$}} %element of Td
\newcommand{\N}{\mathcal{N}}

\newcommand{\R}{\mathbb{R}} %Real numbers
\newcommand{\T}{\mathbb{T}} %Torus
\newcommand{\Z}{\mathbb{Z}} %Integers
\newcommand{\ot}{\otimes}

\newcommand{\End}{\operatorname{End}}

\newcommand{\<}{\langle}

\renewcommand{\>}{\rangle}
\newcommand\be{\begin{equation}}
\newcommand\ee{\end{equation}}

\newcommand{\tdesign}{$t$-design}
\newcommand{\tdesigns}{$t$-designs}
\newcommand{\enet}{$\epsilon$-net}
\newcommand{\enets}{$\epsilon$-nets}

\newcommand{\Ttmu}{T_{\mu,t}}

\newcommand{\dproj}{\mathrm{D}}

\def\cred{\color{red}}

\def\blk{\color{black}}
\definecolor{cadetgrey}{rgb}{0.57, 0.64, 0.69}

\newcommand{\vanx}[1]{{\Delta(#1)^2}}

\def\gauss{f}

\renewcommand{\H}{\mathcal{H}} %Hilbert space
\newcommand{\dg}{\dagger}
\newcommand{\U}{\mathbf{U}} %Unitary channels
\newcommand{\UU}{\mathbb{U}} %Unitary operators
\newcommand{\C}{\mathbb{C}} %Complex numbers
\renewcommand{\R}{\mathbb{R}} %Real numbers
\renewcommand{\P}{\mathbb{P}} %Orthonormal projector
\DeclareMathOperator{\tr}{\mathrm{tr}} %trace
\newcommand{\ii}{\mathrm{i}} %imaginary unit
\renewcommand{\S}{\mathcal{S}} % Subset
\newcommand{\V}{\mathcal{V}} % Other subsets
\newcommand{\G}{\mathcal{G}} % gate-set
\newcommand{\K}{\mathcal{K}} % another representation
\newcommand{\E}{\mathcal{E}} % Statistical ensemble 
\newcommand{\I}{I} % Name for the integral
\newcommand{\F}{\mathcal{F}} % Name for the function on projective unitary group
\renewcommand{\L}{L}% square-integrable functions

    \newcommand{\h}[1]{\boldsymbol{#1}}

\newcommand{\ball}[2]{B( #1 , #2)} % for the ball
\newcommand{\vol}{\mathrm{Vol}} %volume of Haar measure

\usepackage{hyperref}
\hypersetup{colorlinks, citecolor=magenta, filecolor=blue, linkcolor=blue, urlcolor=green}
\newtheorem{theorem}{Theorem}

\newtheorem{corollary}[theorem]{Corollary}
\newtheorem{remark}[theorem]{Remark}

\theoremstyle{plain}
\newtheorem{Result}{Result}
\theoremstyle{plain}
\newtheorem{Theorem}{Theorem}

\theoremstyle{plain}
\newtheorem{Lemma}{Lemma}

\theoremstyle{plain}
\newtheorem{fact}{Fact}

\theoremstyle{plain}
\newtheorem{proposition}{Proposition}

\newcommand{\dt}{{\rm d}}
\newcommand{\ep}{\epsilon}
\newcommand{\hcal}{{\cal H}}

\begin{document} 	
	\title{Epsilon-nets, unitary designs and random quantum circuits}

\author{Micha\l\ Oszmaniec}
\affiliation{ 
Center for Theoretical Physics, Polish Academy of Sciences,\\ Al. Lotnik\'ow 32/46, 02-668
Warszawa, Poland}	
\affiliation{International Centre for Theory of Quantum Technologies, University of Gda\'nsk, Wita Stwosza 63, 80-308 Gda\'nsk, Poland}

\author{Adam Sawicki}
\affiliation{ 
Center for Theoretical Physics, Polish Academy of Sciences,\\ Al. Lotnik\'ow 32/46, 02-668
Warszawa, Poland}		 

\author{Micha\l\ Horodecki}
\affiliation{International Centre for Theory of Quantum Technologies, University of Gdansk, Wita Stwosza 63, 80-308 Gdansk, Poland}

\begin{abstract}

Epsilon-nets and approximate unitary $t$-designs are natural notions that capture properties of unitary operations relevant for numerous applications in quantum information and quantum computing. The former constitute subsets of unitary channels that are epsilon-close to any unitary channel in the diamond norm. The latter are ensembles of unitaries that (approximately) recover Haar averages of polynomials in entries of unitary channels up to order $t$.

In this work we systematically study quantitative connections between these two  notions.  Specifically, we prove that, for a fixed dimension $d$ of the Hilbert space, unitaries constituting $\delta$-approximate $t$-expanders form $\ep$-nets  for    $t\simeq\frac{d^{5/2}}{\ep}$ and $\delta\simeq
\left(\frac{\ep^{3/2}}{d}\right)^{d^2}$.  We also show that for arbitrary $t$, $\ep$-nets  can be used to construct $\delta$-approximate unitary $t$-designs for $\delta\simeq\ep t$, where the notion of approximation is based on the diamond norm. Finally, we  prove that the degree  of an exact unitary $t$ design  necessary to obtain an $\ep$-net must grow at least as fast as $\frac{1}{\ep}$ (for fixed dimension) and not slower than $d^2$ (for fixed $\ep$). This shows near optimality of our result connecting $t$-designs and $\ep$-nets.

We further apply our findings in conjunction with the recent results of \cite{Varju2013} in the context of quantum computing. First, we show that that approximate t-designs can be generated by shallow random circuits formed from a set of universal two-qudit gates in the parallel and sequential local architectures considered in \cite{BHH2016}. Importantly, our gate sets need not to be symmetric (i.e. contains gates together with their inverses) or consist of gates with algebraic entries. Second, we consider a problem of compilation of quantum gates and prove a non-constructive version of the Solovay-Kitaev theorem for general universal gate sets. Our main technical contribution is a new construction of efficient polynomial approximations to the Dirac delta in the space of quantum channels, which can be of independent interest.

\end{abstract}
	\maketitle	
\onecolumngrid

\section{Introduction}

Approximate $t$-designs and $\ep$-nets are natural proxies of the set of all unitary transformations of a finite-dimensional Hilbert space. They capture complementary aspects of unitary channels. We start by reviewing here relevance and contexts in which they appear in quantum information theory.

Unitary approximate $t$-designs \cite{TDesigns2009} are tailored to reproduce statistical moments of degree at most $t$ of the Haar measure on the unitary group. As such, they find numerous applications throughout quantum information, including randomized benchmarking \cite{Gambetta2014}, efficient estimation of properties of quantum states \cite{EffLearning2020}, decoupling \cite{Decoupl2013}, information transmission \cite{InfTransmission2009} and quantum state discrimination \cite{StateDiscrimination2005}. Pseudo-random unitaries are also used to model equilibration of quantum systems \cite{BHH2016,EquilibrationAcin}, quantum metrology with random bosonic states \cite{Oszmaniec2016} and in order to model scrambling inside black holes \cite{ChaosDesign2017,Yoshi2017,Nick2019} . Recently, approximate unitary designs got a lot of attention in the context of proposals for attaining the so-called quantum computational advantage \cite{SuprRev2017}, especially random circuit sampling \cite{Boixo2016} that was recently realized experimentally by Google \cite{Suprem2019}. The reason for this is the anticoncentration property \cite{Hangleiter2018,Jozsa2019}, which seems essential in the proofs of quantum speedup. 

Recently, there was a lot of interest in efficient implementations of pseudo-random  quantum unitaries. First, it is known that the multi-qubit Clifford group forms an exact unitary $3$-design but fails to be a unitary $4$-design \cite{CliffFour}. Second, in \cite{BHH2016} it was shown that random circuits built form Haar-random 2-qubit gates acting (according to the specified layout) on $N$-qubit systems of  the  depth polynomial in $N$ form approximate $t$-designs. This result holds also if the random two-qubit gate set is replaced by a universal gate set that is symmetric  (i.e. contains gates together with their inverses) and consists of gates with algebraic entries. Importantly, both of these requirements are crucial as the arguments of \cite{BHH2016} heavily rely on the work by Bourgain and Gamburd \cite{Bourgain2011}. These results were later improved in 2018 in \cite{Harrow2018}, where even faster convergence in $n$ was proved using specially design layouts in which random two-qubit gates were placed.  Additionally, recent work \cite{Rawad2019} (partially) lifted these stringent requirements. Moreover, the authors of \cite{Qhomeopathy2020} showed that random circuits constructed form Clifford gates and a small number of non-Clifford can be used to efficiently generate approximate designs.   Finally, there exist proposals for efficient generation of approximate $t$-designs using diagonal gates \cite{diagDes2017} and  via Hamiltonian \cite{Yoshi2017} and stochastic \cite{StochDes2017}  dynamics. 

Epsilon nets form (often discrete) subsets of the set of unitary channels that approximate every unitary operation up to some accuracy. They appear  naturally in the context of \emph{compilation} of quantum gates, i.e. the task is to approximate a target unitary gate via the sequence of elementary gates belonging to some "simple" gate-set $\G$. Traditionally, compilation of quantum gates is carried out using the celebrated Solovay-Kitaev algorithm \cite{Kitaev2002,Nielsen2005} which states that for any universal and \emph{symmetric} gate-set $\G$ and any target quantum gate $\h{U}$, there exist a sequence of gates from $\G$ that $\ep$-approximates $\h{U}$ and has length $l\sim\log(\frac{1}{\epsilon})^c$ for $c\approx 3.97$. Moreover, the aforementioned sequence can be found efficiently. Importantly, the Solovay-Kitaev algorithm  requires the gate-set to be symmetric as in the course of the compilation it is necessary to perform \emph{group commutators}. There have recently appeared works which partially lifted this restriction by assuming that the gate-set in question contains an irreducible representation of a group \cite{Harrow2016,BoulandOzols2017}. We also note that the relation between efficient gate approximations and spectral gaps (here we study spectral gaps on restricted spaces rather then on the full space of functions on the unitary group) have been previously used in \cite{Harrow2002}  to show that for specific gate sets the (optimal) scaling $l\sim\log(\frac{1}{\epsilon})$ is possible.

The notions of approximate $t$-designs and epsilon-nets seem to be intuitively related but, according to our best knowledge, the quantitative connection between them has not been systematically studied  before. We would like to remark however that analysis of the proof of Theorem 5 in \cite{HArrowHasstings2008} allows to infer that approximate $t$-expanders (i.e. a type of approximate $t$-design where quality of approximation is measured in operator norm) define $\ep$-nets for $t\simeq d^3/\ep^2$  and $\delta \simeq (\ep/\sqrt{d})^{2d^2}$.  Moreover, a related problem was recently studied in the context of harmonic analysis on Lie groups. Specifically, recent work in \cite{Varju2013} established quantitative relation between  spectral gaps on groups and epsilon nets on these manifolds 
which leads to the following scaling for unitary channels $t\simeq \ep^{-2d^2}$  and $\delta\simeq\epsilon^{(d^4+d^2)/2-1}$. We will comment on the relation of these findings to our Result \ref{res:desAREnets} in Section \ref{sec:approx-designs}.

\emph{Overview of the results and their significance---}  In our work we aim to provide quantitative relation between $\ep$-nets and approximate unitary designs.  We follow closely the approach that was put forward in \cite{Varju2013}, where for a semi-simple compact connected Lie group $G$ it was shown that  $\ep$-nets follow from spectral gaps of certain "transition operators" (defined via the gate-set of interest, and acting on the function spaces built from the irreps of $G$). We translate these to the quantum information language and observe that when a Lie group $G$ is a group of quantum channels $\U(d)$ (isomorphic to the projective unitary group), spectral gaps of the aforementioned transition operators are in one to one correspondence with the parameter $\delta$ in the definition of $\delta$-approximate $t$-designs \cite{RLOWPHD2010} (where the accuracy of approximation $\delta$ is measured by operator norms see Eq. \eqref{eq:DEFexpand}). Making use of this correspondence, we show that $\delta$-approximate $t$-designs can form $\epsilon$-nets. We modify the construction proposed in \cite{Varju2013} and attain better dependence of $t$  on $\ep$ and $d$, the dimension of the Hilbert space (see  Section \ref{sec:approx-designs} for a detailed discussion). Moreover, our arguments do not depend on the detailed knowledge of the representation theory and are instead based solely on the geometry of quantum channels (Result \ref{res:desAREnets}).  We also show that $\ep$-nets can be used to define approximate $t$-expanders. This allows us to prove the converse bounds, i.e. assessing lower bounding minimal $t$ necessary to obtain $\ep$-net (Result \ref{res:tight}).

These general results are then applied to different problems in quantum computation. First, we give a necessary and sufficient criterion for universality of \emph{any} collection of quantum gates. Second, Result \ref{res:SKinv} shows a  (nonconstructive) variant of Solovay-Kitaev theorem for gate-sets $\G$ that, in contrast do the existing results, does not require inverses (see the discussion on Solovay-Kitaev theorem above). Finally, we prove in Result \ref{res:shortDES} that short random quantum circuits generated from two-qubit universal gate-sets $\G$ placed in the parallel and sequential layouts considered in \cite{BHH2016} form approximate $t$-designs. Crucially, compared to previous approaches (see the discussion above) we do not require $\G$ to be symmetric or to have algebraic entries.

\emph{Structure of the paper---} 
In Section \ref{sec:preliminary} we introduce basic concepts and notation. In particular, we describe \tdesigns\ and approximate \tdesigns, and introduce a notion of  distance with respect to which we define epsilon nets.  This allows us to formulate our main results in Section \ref{sec:summary}. Then, in Section \ref{sec:open} we discuss open problems and possible further applications of our results. In Section \ref{sec:mixingOP} we introduce notion of mixing operator $T_{\mu}$ defined on functions acting on unitary channels, and its gap. We relate the operator to the moment operators $T_{\mu,t}$ introduced in Section \ref{sec:preliminary}. 

After these preliminary sections, we are in position to give formal statements and proofs of our findings.  The first group of results 
concerns arbitrary measures on unitary channels. And so, in Section \ref{sec:exact-designs} we prove that for $t$ large enough, an exact \tdesign\ forms an \enet, in Section \ref{sec:approx-designs} we show that 
approximate \tdesign\ also does. In that part we also prove the result in converse direction, namely that from a \enet\  with 
small enough $\ep$ one can construct a \tdesign.  In
Then we move to measures obtained from uniform distribution 
on sequences of gates from some gate set. In section \ref{sec:sequances} we apply the above results (employing some additional results from \cite{Varju2013}) to prove a nonconstructive version of Solovay-Kitaev theorem which does not require assumption that gate set contains inverses. 

Finally, we consider  much more structured measures - namely random circuits on $n$ qudits. In Section \ref{sec:randCIRC} we prove that random circuits (local and parallel) form approximate \tdesigns\ without assuming that the gate set contains inverses, and the gates have algebraic entries. We conclude the main part of the article with Section \ref{sec:constr} where we outline the construction of the polynomial approximation of the Dirac delta on the group of unitary channels. This polynomial function plays a crucial role in the proofs of the results from Sections \ref{sec:exact-designs} and \ref{sec:approx-designs}.

The Appendix is largely devoted to technical results needed in the construction of the aforementioned polynomial approximation of the Dirac delta. Some of the results presented there can be of independent interest because of the intriguing connection with the  random matrix theory (specifically, Tracy-Widom distribution \cite{SzarekBook} and distribution of operator norm of GUE matrices).

\section{Main concepts and notations}
\label{sec:preliminary}

Throughout this work we will be concerned with unitary channels acting on a $d$-unitary dimensional Hilbert space  $\H\simeq \C^d$.  A unitary channel is a CPTP map defined by $\h{U}[\rho]=U\rho U^\dg$, where $\rho$ is a quantum state and $U\in \UU(d)$ is a unitary operator on $\C^d$. In what follows we will denote by $\U(d)$ the set of all unitary quantum channels on $\C^d$.  Note that every unitary operator $U$ uniquely defines a quantum channel but the converse is not true: a quantum channel $\h{U}$ specifies a unitary $U$ up to a global phase. For this reason we can identify $\U(d)$ with the projective unitary group $\mathrm{P}\UU(d)=\UU(d)/\UU(1)$.  Therefore, $\U(d)$ is a compact connected semi-simple Lie group \cite{HallGroups} (we will use this observation in what follows). 

 In order to define the notion of $\ep$-net we need to first specify the distance in the set of unitary channels. We will consider the metric induced by the diamond norm $\dproj_\diamond\left(\h{U},\h{V}\right)\coloneqq \left\| \h{U}-\h{V} \right\|_\diamond$. This notion of distance has strong operational interpretation in terms of maximal statistical distinguishability of quantum channels \cite{NielsenBook}. In this work we use the following closely related notion of distance 
\begin{equation}\label{eq:opDIST}
\dproj\left(\h{U},\h{V}\right)= \min_{\varphi\in[0,2\pi)}\left\|U - \exp(\ii \varphi) V \right\|_\infty\ ,
\end{equation}
where $\|\cdot\|_\infty$ denotes the operator norm and $U,V$ are unitaries representing channels $\h{U}$ and $\h{V}$ respectively. It can be shown (see Proposition \ref{prop:equivOFnorms} in the Appendix) that distances $\dproj_\diamond$ and $\dproj$ are equivalent up to a constant independent on the dimension
\begin{equation}\label{eq:equivDISTANCE}
\dproj\left(\h{U},\h{V}\right) \leq  \dproj_\diamond \left(\h{U},\h{V}\right) \leq 2\dproj\left(\h{U},\h{V}\right)\ .
\end{equation}  

We say that a subset $\S\subset\U(d)$ is an $\ep$-net (with respect to the metric $\dproj$), if for every $\h{U}\in\U(d)$ there exist $\h{V}\in\S$ such that $\dproj\left(\h{U},\h{V}\right)\leq\epsilon$. A set of gates $\G\subset\U(d)$ is called \emph{universal} if  sequences $\h{V}_{n} \h{V}_{n-1} \ldots \h{V}_1$ of gates from $\G$ form $\ep$-nets in $\U(d)$ for arbitrary small $\epsilon$. 

The set of unitary channels $\U(d)$ inherits the unique invariant normalized measure from the unitary group $\UU(d)$ according to the following prescription. For $\S\subset\U(d)$ we set $\mu_P(\S)=\mu\left(\varphi^{-1}(\S)\right)$, where $\mu_P,\mu$ are Haar measures on $\U(d)$ and $\UU(d)$ respectively, and $\varphi^{-1}(\S)$ is the set of all unitary operators that define quantum channels belonging to $\S$.  Haar measure on $\U(d)$ can be also defined via the action on functions of unitaries that are invariant under the global phase (i.e. $F(\exp(\ii \alpha) U) =F(U)$, for arbitrary $U\in\UU(d)$ and $\alpha\in\R$), $\int_{\U(d)} d\mu_P(\h{U}) F(\h{U}) = \int_{\UU(d)} d\mu(U) F(U)$. In what follows we will not differentiate between unitary channels and unitary operators, as well as Haar measures defined on these sets, unless it leads to ambiguity. In particular, will denote by $\vol(\S)$ the Haar measure of a subset $\S$ of unitary channels $\U(d)$ or unitary group $\UU(d)$, depending on the context. We will also use the notation $\dt\mu(\h{U})$ and $\dt\mu(U)$ for "densities" of Haar measures on $\U(d)$ and $\UU(d)$ respectively.

An ensemble of unitaries $\E$ characterized by the probability measure $\nu$  is called a $t$-design \cite{TDesigns2009} iff
\begin{equation}\label{eq:defDESIGN}
\int_{\UU(d)} d\nu(U) G_t(U) = \int_{\UU(d)} d\mu(U) G_t(U)\ ,
\end{equation}
where $G_t$ is arbitrary \emph{balanced} polynomial in $\UU(d)$, i.e. a function of the form $G_t=\tr(A U^{\ot t} \ot \bar{U}^{\ot t})$, where $A$ is an operator on $(\C^d)^{\ot 2t}$. 
Note that balanced polynomials on $\UU(d)$ are well defined functions on $\U(d)$. In this work we will be predominantly interested in \emph{discrete} ensembles, i.e. the ones that take the form $\E=\lbrace{\nu_i, U_i\rbrace}$, for which $\int_{\UU(d)} d\nu(U) F(U) = \sum_i \nu_i F(U_i)$. Approximate unitary $t$-designs (see for example \cite{BHH2016,BHHPRL2016})  are ensembles $\nu$ of unitaries that satisfy \eqref{eq:defDESIGN} up to some desired accuracy. In this work we will focus on a version of  approximate $t$-designs called  \blk $\delta$-approximate $t$-expanders defined as ensembles $\nu$ satisfying 
\begin{equation}\label{eq:DEFexpand}
\left\|T_{\nu,t}-T_{\mu,t} \right\|_\infty \leq \delta\  ,
\end{equation}
where for any measure $\nu$ (in particular for the Haar measure $\mu$) we define a \emph{moment operator}
 \begin{equation}\label{eq:momentOP}
T_{\nu,t}\coloneqq\int_{\UU(d)} d\nu(U) U^{\ot t} \otimes \bar{U}^{\ot t}\ .
\end{equation}
The quantity  $\delta(t,\nu)\coloneqq \left\|T_{\nu,t}-T_{\mu,t} \right\|_\infty $ is sometimes called {\it expander norm} of $\nu$. 
We will also use another notion of approximate $t$-design which is based on the diamond norm distance between the $t$-particle quantum channel defined by $\nu$ and its counterpart defined by $\mu$.  Specifically, we will say that the ensemble of unitaries $\nu$ forms $\delta_\diamond$-approximate $t$-design if and only if 
\begin{equation}\label{eq:appDIAMOND}
\left\|\Delta_{\nu,t}-\Delta_{\mu,t}   \right\|_\diamond \leq \delta_\diamond\ ,
\end{equation}
where 
\begin{equation}
\Delta_{\nu,t}\ \coloneqq\int_{\UU(d)} d\nu(U) \h{U}^{\ot t}.
\end{equation}

There exist other related definitions of approximate designs that use different quantifiers to gauge how well $\nu$ approximates the properties of Haar measure $\mu$ (see for example \cite{RLOWPHD2010}).

%We conclude this section by introducing the neccesary assymptotic notation. Given two positive-valued functions $f$ and $g$, we write  $f=\o(g)$ if $\lim_{x\rightarrow\infty} f(x)/g(x) = 0$ and   $f=\O(g)$ if $ \lim_{x\rightarrow\infty} f(x)/g(x) < \infty$. Likewise, we say $g=\Omega(f)$ if $f=\O(g)$. Finally, $f\approx g$ will refer to the situation when $ \lim_{x\rightarrow\infty} f(x)/g(x)=1$. \urgent{Adjust and change when needed}

\section{Summary of main results}
\label{sec:summary}

 Here we present our main findings regarding the relation between approximate designs (expanders) and epsilon-nets.

\begin{Result}[Approximate $t$-expanders define $\ep$-nets]\label{res:desAREnets}
Consider an ensemble $\E=\lbrace{\nu_i, U_i\rbrace}$ of unitaries described by the discrete measure  $\nu$ on $\UU(d)$. Let $\ep\in[0,d)$ and assume that ensemble $\E$ is a $\delta$-approximate $t$-expander with   $t\simeq\frac{d^{5/2}}{\ep}$  (up to logarithmic factors in $d$ and $1/\ep$) and  $\delta\simeq\left(\frac{\ep^{3/2}}{d}\right)^{d^2}$. Then, the channels $\lbrace{\h{U}_i \rbrace}$ defined via the elements of $\E$ form an $\ep$-net in $\U(d)$ with respect to the distance $\dproj$ induced by the dimond norm.
\end{Result}

We note that setting $\delta =0$ gives the connection between \emph{exact} $t$-designs and $\ep$-nets. We give the technical formulation of the above result in Theorems \ref{th:Texact} and \ref{th:approx-design}. There we state the explicit dependence of $t$ and $\delta$ on the dimension of the Hilbert space $d$ and generalise the above statements to arbitrary probability measures (ensembles) on $\UU(d)$. Our proofs follow the method presented  in \cite{Varju2013}. Our  technical contributions are twofold. First, we simplify the original arguments making them largely independent of the machinery of group theory and thus more accessible for the broader audience. Second, in Theorem \ref{th:FUNC} we construct an efficient polynomial approximation of the Dirac delta on $\U(d)$ which allows us to attain better dependence of $t$ on the dimension $d$ and $\ep$. Our construction can be of independent interests and its details are provided in Section \ref{sec:constr}.

Result \ref{res:desAREnets} can be used to find out how many times one needs to iterate gates comprising the $\delta$-approximate $t$-design so that they form an $\ep$-net. Specifically in Proposition \ref{prop:prop-words}  we prove that that it is enough to iterate them $l\simeq\frac{d^2\log(\frac{1}{\ep})}{\log(\frac{1}{\delta})}$ times. This establishes intriguing connection between complexity of quantum gates and the property of being approximate $t$-design.

We also prove the connection in the opposite direction. Namely, we show that epsilon nets can be used to construct approximate designs (see  Theorem \ref{th:netsAREdesigns} for the formal statement).  
\begin{Result}[$\ep$-nets define approximate $t$-designs]
Let $\S$ be a gate-set forming an $\ep$-net in $\U(d)$ and let $t$ be arbitrary natural number. Then, there exists an ensemble of quantum gates from $\S$ which forms an $(2\ep t)$-approximate $t$-expander (see Theorem \ref{th:tightness} for the formal statement).
\end{Result}

This finding allows us to prove that the scaling of the degree $t$ necessary to ensure that $t$-design forms $\ep$-net is near optimal. 
\begin{Result}[Near-tightness of scaling of $t$ from Result \ref{res:desAREnets}]\label{res:tight} Let $t(d,\ep)$ be the minimal degree of an exact unitary $t$-design $\nu$ such that unitaries from $\nu$ form   $\ep$-net with respect to distance $\dproj$ in $\U(d)$ (see Eq. \eqref{eq:opDIST}. Then, $t(d,\ep)$ must scale faster than $\frac{1}{\ep}$ (for fixed $d$). Moreover, for fixed $\ep$ and increasing $d$ the degree  $t(d,\ep)$ scales at least like $d^2$.      
\end{Result}

We apply the results established above in the context of quantum computing. To this end we use additional ingredient which follows from Theorem 6 of \cite{Varju2013} (see  Section \ref{sec:mixingOP} and Theorem \ref{thm:Varju}  for the translation of representation-theoretic concepts to the formalism of tensor expanders). Specifically, the spectral gap of the moment operator $T_{\nu,t}$ associated to a measure $\nu$ supported on a universal gate-set $\G$ closes not faster than $\frac{A}{\log(t)^2}$. Note that the above relies solely on universality of $\G$ \ so that  the assumptions made e.g. in \cite{Bourgain2011}  on algebraic entries of gates and the property that $\G$ is symmetric (i.e. $\h{V}\in\G$ implies $\h{V}^{-1}\in\G$) are not relevant. Leveraging this and the recent results of \cite{SK2017,SK2017bis}, it is possible to prove that universality of of \emph{any} gate-set $\G$ is equivalent to being $\delta$-approximate $t_\ast$-expander, where $\delta<1$ and $t_\ast$ depends solely on $d$. This finding complements recent results \cite{Zimboras2015,OZ2017}  that classified semi-simple compact Lie subgroups of $\U(d)$ in terms of their second order commutants).

Finally, we use the above strong spectral gap results of Varju to show the following two results which are relevant to theoretical underpinnings of quantum computing.

\begin{Result}[Non-constructive inverse-free Solovay-Kitaev]\label{res:SKinv}
Let $\G\subset\U(d)$ be a universal gate-set in $\U(d)$ (not necessarily symmetric i.e. $\h{V}\in\G$ \emph{does not} imply $\h{V}^{-1}\in\G$).  Then, every unitary channel $\h{U}$ can be approximated by sequences of gates from $\G$ of length $l\approx \log(\frac{1}{\ep})^3$. 
\end{Result}
The formal proof can be found in Section \ref{sec:sequances}. We note that  Result  \ref{res:SKinv} does not give a constructive algorithm to find the approximating sequence of gates. Furthermore, we note that the scaling $l\approx \log(\frac{1}{\ep})^3$ is not optimal and for specific gate sets (with nonvanishing spectral gap)   Our last result shows that approximate $t$-designs can be generated efficiently by local random circuits without assuming inverses and algebraic entries. The formal proof is given Section \ref{sec:randCIRC}. 
\begin{Result}\label{res:shortDES}
Let $\G$ be a set of universal two-qudit gates. Consider two types of random circuits on line of $n$ qudits \cite{BHH2016}. 
\begin{itemize}
    \item \emph{Local random circuits}: we pick uniformly at random 
two neighboring qudits, and apply gate chosen from $\G$ according to 
uniform measure $\nu_\G$. We denote the resulting distribution by $\nu_{loc}(\G)$.
\item \emph{Parallel random circuits}: we apply with probability $1/2$
either $U_{12}\otimes U_{34}\otimes\ldots \otimes U_{n-1,n}$
or $U_{23}\otimes U_{45}\otimes\ldots \otimes U_{n-2,n-1}$, where each $U_{ij}$ is picked independently from $\G$ according to $\nu_\G$. We denote the resulting distribution by $\nu_{par}(\G)$.
\end{itemize}
Let $l_{loc,Haar}$ ($l_{par,Haar}$) be lengths of random local (parallel) circuits which are $\delta$-approximate t-expanders, where instead of $\nu_\G$ we take Haar measure over two-qudit gates. There exist a constant $C(\G)$ such  that if  
\begin{equation}\label{eq:comLEANGTHS}
    l_{loc}\geq n \log^2(t) C(\G) l_{loc,Haar},\quad l_{par}\geq 2 \log^2(t) C(\G) l_{par,Haar}\ .
\end{equation}
then, the corresponding random circuits ($\nu_{loc}(\G)^{\ast l_{loc}}$ and $\nu_{par}(\G)^{\ast l_{par}}$ ) are $\delta$-approximate $t$-expanders.

Note that in \cite{BHH2016} 
it was shown that  local (parallel) random quantum circuits with Haar distributed gates of lengths 
satisfying 
\begin{eqnarray}
    && l_{loc,Haar}\geq 
    42500 n\lceil \log_d (4t) \rceil^2 d^2 t^{5+3.1 \log d} (2nt\log d + \log(1/\ep))\nonumber \\
    && l_{par,Haar} \geq 
    523000\lceil \log_d (4t) \rceil^2 d^2 t^{5+3.1 \log d} (2nt\log d + \log(1/\ep))
\end{eqnarray}
are $\delta$-approximate t-expanders. It then follows 
that circuits constructed from $\G$ scale efficiently with $n$, too. 
\end{Result}
\begin{remark} It is straightforward to derive analogous bounds for other notions of approximate $t$-designs (based, for example, on the diamond norm). The conclusions are analogous. Let us stress, however, that our proof technique does not immediately apply to the scenarios considered in \cite{Harrow2018} and hence we cannot use it to get convergence faster than $n$ for $\sqrt{n}\times \sqrt{n}$- qubits square lattice. We however believe that this technical problem can be overcome with some effort.
\end{remark}
 {\bf Acknowledgements} We are grateful to Stanis\l aw Szarek for explaining to us the intricacies of computing volumes of balls in the unitary group and related manifolds.  AS acknowledges financial support from National Science Centre, Poland under the grant
SONATA BIS: 2015/18/E/ST1/00200. MH acknowledges support from the Foundation for Polish Science through IRAP project co-financed by EU within the Smart Growth Operational Programme (contract no.2018/MAB/5).  MO acknowledges
the financial support by TEAM-NET project (contract
no. POIR.04.04.00-00-17C1/18-00).

\section{Open problems}
\label{sec:open}
We conclude the introductory part of our wrok with a list of interesting problems which  we left for further research.
\begin{itemize}
    \item {\it Optimal scaling of $t(\ep,d)$ and $\delta(t,d)$:}     Can one improve scaling in the results connecting \enets\ with 
    \tdesigns? We conjecture that with some work it should be possible to obtain $t\simeq d^2$ (for fixed $\ep$).
    
    \item {\it Explicit constant in SK theorem:} Unlike in all other results,  our version of Solovay-Kitaev theorem  contains an unknown constant depending on the dimension and the gate set. To what extent we can determine it (at least to leading order in the dimension)?
 %   \item {\it Improve constants:} dependence on dimension in our estimates is not necessarily optimal, as we have applied quite crude estimate of volume of projective ball. There is also much room for improvement in other places - e.g. estimates such as in Lemma \ref{lem:upperbound} can be sharpened at least for small $d$. It is not however excluded, that dependence on dimension is, at least asymptotically, optimal. 
    \item{\it Termination of the universality checking algorithm:} The explicit value of constant in our version of Solovay-Kitaev theorem and the connection between approximate $t$-designs and $\epsilon$-nets can shed a new light on complexity of universality checking algorithms proposed in \cite{SK2017}.
    \item{\it Connection with black hole dynamics and complexity growth:} Recently, there were some interesting works connecting complexity of random circuits with black hole dynamics (cf. \cite{Roberts-Yosida2017,Suskind2018,Nick2019,ModelsComplexity2019}). It is conceivable that our findings may provide some useful tools, especially in the high complexity regime. In this context it is also natural to explore the possible generalizations of our results to \emph{approximate projective designs}  and $\ep$-nets in the set of pure quantum states.
\end{itemize}

\section{Mixing operators on unitary group, their gap and approximate designs}\label{sec:mixingOP}

In this section we establish the connection between spectral gaps of mixing operators on unitary channels and approximate unitary $t$-designs (expanders).  
Let  $\L^2(\UU(d))$ be the Hilbert space space of square-integrable functions on $\UU(d)$, i.e. functions satisfying $\int_{\UU(d)}d\mu(U)|F(U)|^2<\infty$, where $\mu$ denotes the Haar measure on $\UU(d)$. For every $V\in\UU(d)$ we introduce a unitary shift operator $T_V:\L^2(\UU(d))\rightarrow \L^2(\UU(d))$ defined {\it via} $\left(T_V(F)\right)(U)= F(V^{-1}U)$. It is easy to verify that the mapping $\phi: V\mapsto T_V $ satisfies $\phi(UV)=\phi(U)\phi(V)$, for every $U,V\in\UU(d)$. Therefore, so-defined map $\phi$ is a unitary representation of $\UU(d)$ in $\L^2(\UU(d))$, usually called the \emph{left regular representation} of $\UU(d)$. For more background in representation theory of Lie groups and Lie algebras see e.g. \cite{fulton1991representation}.  For every measure  $\nu$  on $\UU(d)$ we can consider an operator $T_\nu:\L^2(\UU(d))\rightarrow \L^2(\UU(d))$ which is defined as a convex combination of operators $T_V$ according to measure $\nu$, $T_\nu = \int_{\UU(d)} \dt \nu(V) T_V$. Its action on functions on $\UU(d)$ can be explicitly written as
\begin{equation}\label{eq:transitionOP}
\left(T_{\nu}F\right)(U)= \int_{\UU(d)} \dt \nu(V)  F(V^{-1}U)\ .
\end{equation}
The operator $T_\nu$ can be understood it as a transition operator of a random walk on $\UU(d)$ in which at every step a unitary is applied at random according to the measure $\nu$.  Specifically,  $\left(T_{\nu}F\right)(U)$ is the average value of the function $F$ over a realization of a single step of a random walk generated by $\mu$ an originating at $U$.

We shall also consider restriction $T_{\nu}|_{\hcal_t}$ of $T_{\nu}$ to the subspace $\hcal_t$ spanned by balanced  polynomials of degree up to $t$ in $U$ as well as in $\bar U$ i.e. subspace of functions on $\UU(d)$ of the form $G_t(U)=\tr\left(A U^{\ot t} \ot \bar{U}^{\ot t} \right)$. In particular, if we choose $\nu$ to be the Haar measure $\mu$ on $\UU(d)$ then the operators $T_{\mu}$ and $T_{\mu}|_{\hcal_t}$ are projectors - they project onto the space of constant functions on $\UU(d)$. Let us denote the space orthogonal to the constant functions on $\UU(d)$ by $\L^2_{0}(\UU(d))$. We define the gap of $T_\nu$ as:
\begin{equation}
g(T_\nu)\coloneqq 1 - \|T_\nu|_{L^2_{0}(\UU(d))}\|_\infty\ .
\end{equation}
We note that the so-defined function can be greater than zero iff the support of the measure includes a set of universal gates.  We are only interested in such situations, so we will the keep name gap for  $g$.   

We define a gap for $T_{\nu}|_{\hcal_t}$ analogously as for $T_\nu$ and denote it by $g(\nu,t)$. By straightforward calculations we get
\begin{equation} 
\|T_{\mu}|_{\hcal_t}-T_{\nu}|_{\hcal_t}\|_\infty=1-g(\nu,t)
\end{equation}
The following proposition establishes a very useful connection between $T_\nu|_{\hcal_t}$ and moment operator $T_{\nu,t}$ introduced in Eq.\eqref{eq:momentOP}.

\begin{proposition}\label{prop:gaps} For any measure $\nu$ on $\UU(d)$ we have
\begin{equation} 
\|T_{\mu}|_{\hcal_t}-T_{\nu}|_{\hcal_t}\|_\infty=\|T_{\mu,t}-T_{\nu,t}\|_\infty
\end{equation}
and consequently we have $\delta(\nu,t)=1-g(\nu,t)$, where $\delta(\nu,t)$ is the expander norm of $\nu$.
\end{proposition}
\begin{proof}
The action of $T_\nu$ and $T_\mu$ on $\mathcal{H}_t$ is determined by the left regular representation $\phi$ restricted to $\mathcal{H}_t$. The space $\mathcal{H}_t$ decomposes into irreducible components $\H_t=\bigoplus_\lambda \K^\lambda$ and we have 
\begin{gather}\label{eq:regDEC}
   T_\nu|_{\hcal_t}\approx\bigoplus_\lambda\int_{\UU(d)}d\nu(U)\Pi^\lambda(U)\ ,
\end{gather}
where $\Pi^\lambda(U)$ is the matrix corresponding to $U$ via the irreducible representation with the highest weight $\lambda$ (highest weight label different irreducible representation of $\UU(d)$ \cite{fulton1991representation}) and the symbol $\approx$ denotes unitary equivalence. On the other hand the representation $U\mapsto U^{\ot t}\ot  \bar{U}^{\ot t}$ is reducible and decomposes into 
\begin{gather}\label{eq:UUdecomposition}
 U^{\ot t}\ot  \bar{U}^{\ot t} \approx \bigoplus_{\lambda^\prime}\Pi^{\lambda^\prime}(U)\ .
\end{gather}
Thus the operator $T_{\nu,t}$ can be written as
\begin{gather}\label{eq:momentDECOMP}
    T_{\nu,t}\approx\bigoplus_{\lambda^\prime}\int_{\UU(d)}d\nu(V)\Pi^{\lambda^\prime}(U)\ .
\end{gather}
We notice, however, that the space $\mathcal{H}_t$ is spanned by the matrix elements of the representation $U^{\ot t}\ot  \bar{U}^{\ot t}$ and hence, by the decomposition \eqref{eq:UUdecomposition}, by matrix elements of irreducible representations $\Pi^{\lambda^\prime}$. Let $\mathcal{W}^{\lambda^\prime}$ be the linear span of functions   $F^{\lambda^\prime}_{ij}(U)=\bra{i}\Pi^{\lambda^\prime}(U)\ket{j}$. It can be verified by direct computation that for every $V\in\UU(d)$ we have
\begin{equation}
    T_V|_{\mathcal{W}^{\lambda^\prime}} \approx \Pi^{\lambda^\prime}(V)\ot \mathbb{I}_{m^{\lambda^\prime}}\ ,
\end{equation}
where $m^{\lambda^\prime}$ is the dimension of the multiplicity space equal to  $|\K^{\lambda^\prime}|$, the dimension of carrier space of representation $\Pi^{\lambda^\prime}$. Thus it follows that that collection of weights $\lbrace\lambda\rbrace$ and $\lbrace\lambda^\prime \rbrace$ agree, up to multiplicities. The theorem now follows from comparing decompositions \eqref{eq:regDEC} and \eqref{eq:momentDECOMP}.

\end{proof}
As an immediate consequence we get that for a \tdesign, the gap $g(\nu,t)$ is equal to $1$.   We conclude this part by noting that composition of operator $T_\nu$ is compatible with taking convolutions in the sense that for all $l$ we have  $T_{\nu^{\ast l}}=(T_\nu)^l$. This implies the following well-known result.
\begin{fact}
	\label{fact:delta-l}
If $\nu$ is a $\delta$-approximate $t$-expander, then $\nu^{\ast l}$ 
is a $\delta^l$-approximate $t$-expander. 
\end{fact}

\section{Exact t-designs and epsilon-nets}
\label{sec:exact-designs}

\begin{figure}[t]
    \centering
    \includegraphics[width=0.5\textwidth]{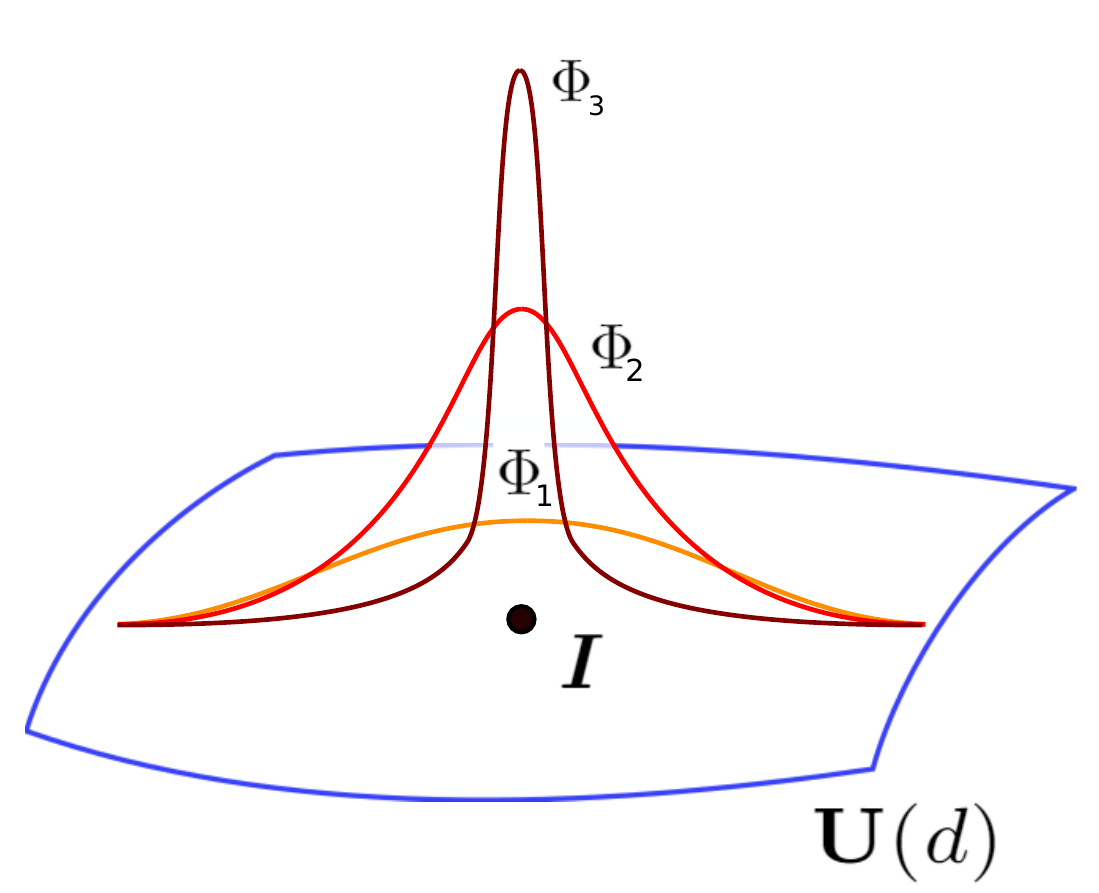}
    \caption{A graphical presentation of a sequence of polynomial approximations $\Phi_k(\h{U})$ of Dirac delta at $\h{I}$ in $\U(d)$. As the degree $k$ increases the functions $\Phi_k$ are more and more peaked in the vicinity of $\h{I}$, while retaining the normalisation $\int_{\U(d)} \dt \mu(\h{U}) \Phi_k(\h{U})~=~1$.  }
    \label{fig:inputs}
\end{figure}

In this part we will show that elements of exact $t$-designs form $\ep$-nets with respect to distance $\dproj$ provided $t\simeq \frac{d^{5/2}}{\ep}$ (up to logarithmic factors in $d$ and $1/\ep$). We follow the ideas from \cite{Varju2013} with two important differences. First, we significantly reduce the usage of representation theory. Second, we construct a new polynomial approximation of the Dirac delta on the group of quantum channels (see Theorem \ref{th:FUNC} and Section \ref{sec:constr} where we provide details of the construction). This allows us to obtain improved dependence of $t$ on $d$ and $\ep$  in Theorem \ref{th:Texact}.

We start with giving the intuition beyond the proof of our result. We consider a family of real-valued balanced polynomials $\Phi_k \in\H_k$ (i.e. polynomials of degree at most $k$)  that has the following properties:
\begin{itemize}
\item Normalisation: $\int_{\U(d)} \dt \mu(\h{U}) \Phi_k(\h{U}) =1$, for all $k$.
\item Vanishing integrals on balls sufficiently far from identity $\h{I}$: for every $\ep\in[0,2]$ and for every $\h{V}_0$ such that $\dproj\left(\h{V}_0,\h{I}\right)\geq \ep $ we have
\begin{equation}\label{eq:qualDELTA}
\int_{B(\h{V}_0,\ep/2)} \dt \mu (\h{U}) |\Phi_k (\h{U})| \rightarrow 0 \ \text{as } k\rightarrow\infty \ ,
\end{equation}
where  $\ball{\h{V}_0}{\ep}=\lbrace{\h{U}\in\U(d) |\  \dproj(\h{U},\h{V}_0)\leq\ep \rbrace}$.
\end{itemize}
Functions $\Phi_k$ can be regarded as polynomial approximation of the Dirac delta localized at $\h{I}$, the identity channel  (see Fig. \ref{fig:inputs}).

We then consider the following integral,
\begin{equation}\label{eq:KeyInt}
    \I(\nu,\ep,k,\h{V}_0) \coloneqq \int_{\ball{\h{V}_0}{\ep/2}} d\mu(\h{U}) (T_\nu \Phi_k)(\h{U})\  ,
\end{equation}
where $\h{V}_0 \in\U(d)$,  and for any measure $\nu=\lbrace \nu_i, V_i\rbrace$ and a function $\F$ on $\U(d)$ we define (analogously as before for functions on $\UU(d)$).  Next, under the assumption that $\nu$ is an exact $k$-design we show in Lemma \ref{prop:volupperbound} that  $\I(\nu,\ep,k,\h{V}_0)$ equals the Haar measure of $\ball{\h{V}_0}{\ep/2}$. On the other hand if channels from the support of $\nu$ do not form an $\epsilon$-net in $\U(d)$  we can use Eq.\eqref{eq:qualDELTA}  to prove that $\I(\nu,\ep,k,\h{V}_0)$ vanishes as $k\rightarrow \infty$   (see Lemma \ref{lem:upperbound} and Theorem \ref{th:FUNC}). We finally look for $k$ such that $\I(\nu,\ep,k,\h{V}_0)$  is smaller than $\vol(\ball{\h{V}_0}{\ep/2})$, value of which is controlled by Fact \ref{fact:Szarek}. This number gives a degree of exact $t-$design that is ensured to form $\ep$-net.  The graphical presentation of this general reasoning is given in Fig. \ref{fig:actionOnPhiK} while technical details are given below. The main result  connecting exact $t$ designs with $\ep$ nets is Theorem \ref{th:Texact}.

\begin{figure}[h]
    \centering
    \includegraphics[width=1\textwidth]{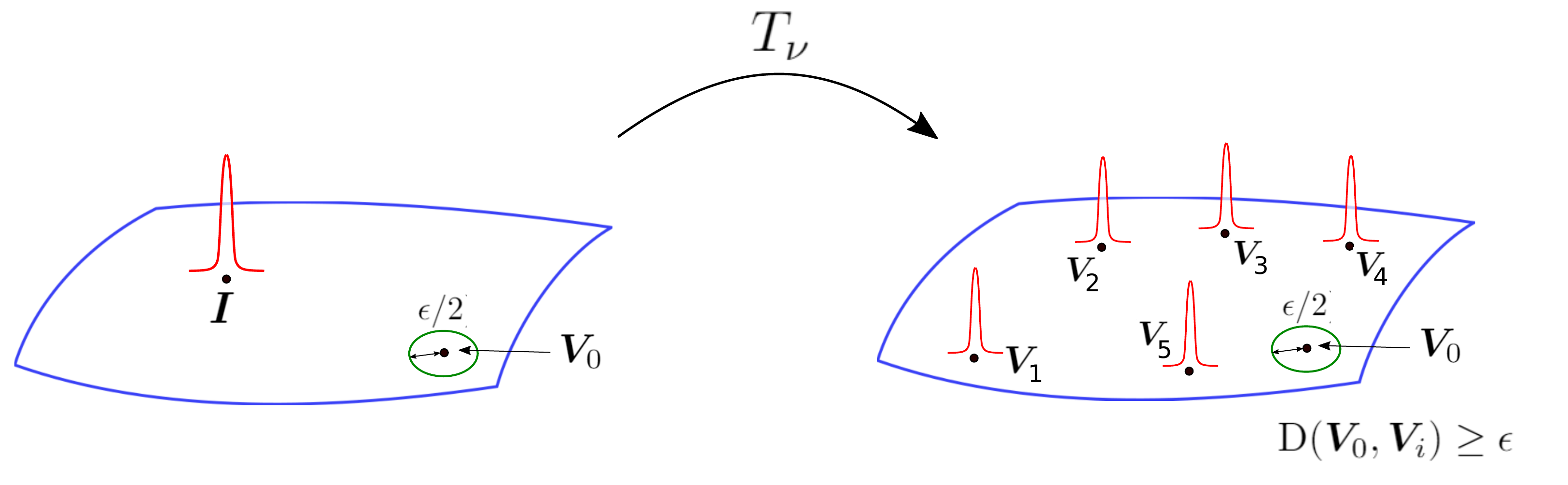}
    \caption{A visualisation of the general argument that allows to connect $t$-designs with $\ep$-nets. From the $t$- design property we know that for $k\leq t$ the integral  $\int_{\ball{\h{V}_0}{\ep/2}} d\mu(\h{U}) (T_\nu \Phi_k)(\h{U})$ equals $\vol(B(\h{V}_0,\ep/2)$, the volume of the Ball of the radius $\ep/2$ centered around $\h{V}_0$ (interior of the green cycle).  On the other hand action of the transition operator $T_\nu$ transforms the function $\Phi_k$, initially localized around $\h{I}$ (red peak in the left part of the figure) into a convex combination of functions $\Phi^i_k(\h{U})=\Phi_k(\h{V}^{-1}_i\h{U})$ localized around points $\h{V}_i\in\mathrm{supp}(\nu)$ (smaller red peaks in the right part of the figure). Assuming that $\mathrm{supp}(\nu)$ does not form an $\ep$-net we know that there exist $\h{V}_0$ that satisfies $\dproj(\h{V}_0,\h{V}_i)\geq\ep$. By increasing $k$ and keeping $\dproj(\h{V}_0,\h{V}_i)\geq\ep$ we get $\int_{\ball{\h{V}_0}{\ep/2}} d\mu(\h{U}) (T_\nu \Phi_k)(\h{U})\rightarrow 0$ since the integral is over the ball $\ball{\h{V}_0}{\ep/2}$  whose points are far away form unitaries $\h{V}_i$ and functions $\Phi^i_k$ approximate the Dirac delta localized at $\h{V}_i$ as $k$ increases.   Therefore, there must exist $t$ such that elements of a $t$-design form an $\ep$-net. }
    \label{fig:actionOnPhiK}
\end{figure}

\begin{Lemma}
	\label{prop:volupperbound}
	Let $\nu$ be a measure on $\UU(d)$ which is an exact unitary $t-$design. Then for arbitrary function $\Phi\in \hcal_t$  (i.e. a balanced polynomial of degree at most $t$ in $U$ and in $\bar{U}$) satisfying 
\begin{equation}
	\int_{\U(d)} \dt \mu(\h{U}) \Phi(\h{U}) = 1,
	\label{eq:phi-normalization1}
\end{equation}
and for any $ V\in\UU(d)$, we have 
\begin{equation}\label{eq:desAVER}
\int_{\ball{\h{V}}{\ep}} \dt \mu(\h{U}) \left(T_\nu \Phi \right) (\h{U})  =  \vol(\ball{\h{V}}{\ep}), 
\end{equation}
where $\ep\in[0.2]$.
\end{Lemma}

\begin{proof}
	Since $\nu$ is an exact \tdesign, and $\Phi\in\H_t$, the moment operator $T_\nu$ projects function $\Phi$ onto a constant function. Therefore, due to Eq.\eqref{eq:phi-normalization1}
	$T_{\nu} \Phi=1$ (a constant function equal to 1). As a result we get Eq.\eqref{eq:desAVER}.
\end{proof}

As explained above, our goal is to upper bound  the integral defined in Eq.\eqref{eq:KeyInt} in terms of $k$. To this aim we will use the following technical Lemma. 

\begin{Lemma}
	\label{lem:upperbound}
 Let $\nu$ be a measure on $\UU(d)$. Suppose that the support of  $\nu$ is not an $\epsilon$-net   in $\U(d)$. Then, there exists $\h{V}_0$ such that for any function $\Phi$ on $\U(d)$, and any $\kappa$ satisfying $0\leq \kappa \leq \epsilon$  we have 
	\begin{equation}
	\int_{\ball{\h{V}_0}{\kappa}} \dt \mu(\h{U}) (T_{\nu} \Phi) (U)  \leq \max_{\h{V}:\dproj(\h{V},\h{I})\geq \ep} \int_{\ball{\h{V}}{\kappa}} \dt \mu(U) \Phi(U) \ .
	\end{equation}
\end{Lemma}

\begin{proof}
 For simplicity we assume that the measure $\nu$ is discrete i.e. $\nu=\lbrace\nu_i,V_i\rbrace$. The proof is analogous in the general case. Let $\h{V}_0$ be a unitary channel that cannot be $\ep$-approximated by elements form the support of $\nu$: $\dproj(\h{V}_0,\h{V}_i)\geq\ep$. From the definition of the moment operator (see Eq. \eqref{eq:transitionOP}) we have
\begin{equation}
\int_{\ball{\h{V}_0}{\kappa}} \dt \mu(\h{U}) (T_{\nu} \Phi) (\h{U}) = \sum_i \nu_i \int_{\ball{\h{V}_0}{\kappa}} \dt \mu(\h{U}) \Phi (\h{V}_{i}^{-1}\h{U})\ .   
\end{equation}
By changing the variables in each summand  $\h{U}'=\h{V}_{i}^{-1}\h{U}$ and denoting $\h{V}'_i=\h{V_i}^{-1} \h{V}_0$ we get
\begin{equation}\label{eq:firstSTEP}
\int_{\ball{\h{V}_0}{\kappa}} \dt \mu(\h{U}) (T_{\nu} \Phi) (\h{U}) =  \sum_i \nu_i \int_{\ball{\h{V}'_i}{\kappa}} \dt \mu(\h{U}') \Phi (\h{U}')\ .
\end{equation}
Finally, using the defining property of $\h{V}_0$ and employing the unitary invariance of $\dproj$ we obtain 
\begin{equation} 
\dproj(\h{V}'_i,\h{I})=\dproj(\h{V}^{-1}_i \h{V}_0,\h{I})=\dproj(\h{V}_0,\h{V}_i) \geq \ep \ .
\end{equation}
We conclude the proof by using the above inequality in each summand of Eq.\eqref{eq:firstSTEP}.
\end{proof}

The following statement about the volume of the Ball in the space of unitary channels is known as folklore in quantum information community. Here we adapt a rigorous result of \cite{Szarek98}. 

\begin{fact}[Estimates for the volume of Ball in the manifold of quantum channels \cite{Szarek98}] \label{fact:Szarek}
Let  $\ball{\h{V}}{\ep}=\lbrace{\h{U}\in\U(d) |\  \dproj(\h{U},\h{V})\leq\ep \rbrace} $ be a ball centered around $\h{V}\in\U(d)$, where  $\dproj$ is the distance from Eq.\eqref{eq:opDIST}. There exist absolute constants $c,C>0$ such that for all $\ep\in[0,2]$ 
\begin{equation}\label{eq:boundONvolume}
\left(\frac{\ep}{C}\right)^{d^2-1}\leq \vol(\ball{\h{V}}{\ep}) \leq \left(\frac{\ep}{c}\right)^{d^2-1}\,,
\end{equation}
where  $C=9\pi$ and $c=1/87$.
\end{fact}
\begin{remark}
In the original work of Szarek \cite{Szarek98} considered general homogeneous spaces of $\UU(d)$ equipped with the metric induced from the operator norm. By the virtue of the variational characterization of the distance $\dproj$ given in Eq.\eqref{eq:opDIST} the results presented there apply directly to our scenario. 
\end{remark}

The last necessary element in our proof strategy is the existence of efficient polynomial approximation of the Dirac $\delta$ 
 in the space of unitary channels. Here we present only the final result, while details of the construction and the necessary technical details are presented in Section \ref{sec:constr} and the appendix. 

\begin{Theorem}[Efficient polynomial approximation of the Dirac $\delta$ on unitary channels]\label{th:FUNC}
Consider a set of Unitary channels $\U(d)$ on $d$-dimensional quantum system equipped with a metric $\dproj$  (see  Eq.\eqref{eq:opDIST}). 
Let $\ep\in(0,2/3]$, $\sigma\leq \frac{\ep}{6 \sqrt{d}}$.  There exists a function $\F^{\sigma}_k:\U(d)\rightarrow \R$ with the following properties 
\begin{enumerate}
\item Normalisation: $\int_{\U(d)} \dt \mu (\h{U}) \F^{\sigma}_k(\h{U})  =1$.
\item Vanishing integral of modulus outside of the ball $B(\h{I},\ep)$:
for 
\begin{equation}\label{eq:degreeTAIL}
k\geq 5 \frac{d^{\frac{3}{2}}}{\sigma}  \sqrt{\frac{1}{8}\frac{\ep^2}{d^2\sigma^2} + \log(\frac{1}{\sigma})}
\end{equation}
we have 
\begin{equation}
\int_{ B(\h{I},\ep)^c } \dt \mu (\h{U}) |\F^\sigma_k (\h{U})| \leq 9 \exp\left(-\frac{ \epsilon^2 }{4\sigma^2}\right) \left(\frac{\pi}{2}\right)^{d(d-1)} \ .
\end{equation}
\item Low degree polynomial: $\F^{\sigma}_k (\h{U})$ can be represented as a balanced polynomial in $U$ and $\bar{U}$  of degree $k$.
\item Bounded $L^2$-norm:  for $k\geq d/\sigma$ we have 
\begin{equation}\label{eq:L2norm}
\left\|\F^\sigma_k\right\|_{2}=\sqrt{\int_{\U(d)} \dt  \mu (\h{U})  \left|\F^{\sigma}_k(\h{U}) \right|^2} \leq 
8 \times 2^{d^2} \, \sigma^{-d(d-\frac12)}
\ .
\end{equation}

\item $L^1$-norm  close to $1$: for  $k$ 
satisfying \eqref{eq:degreeTAIL} 
we have 
\begin{align}
    1\leq \|\F_k^\sigma\|_1\leq  1 + 
6 \exp\left(-\frac{ \ep^2 }{4\sigma^2}\right) \left(\frac{\pi}{2}\right)^{d(d-1)} \ ,
\end{align}
where $\|\F_k^\sigma\|_1 = \int_{\U(d)} \dt  \mu (\h{U})  \left|\F^{\sigma}_k(\h{U}) \right|$.
\end{enumerate}

\end{Theorem}

\begin{corollary} \label{cor:far-away-ball}
Let $\kappa,\epsilon,\sigma$ be  positive numbers satisfying $\ep\in[0,2/3]$,  $\kappa\leq\ep$, $\sigma\leq \frac{\ep-\kappa}{6\sqrt{d}}$. Moreover, let $k$ by natural number satisfying \eqref{eq:degreeTAIL}.
Then for every $\h{V}$ such that $\dproj\left(\h{V},\h{I}\right)\geq \ep $ we have 
\begin{equation}
\int_{B(\h{V},\kappa)} \dt \mu (\h{U}) \F^\sigma_k (\h{U}) \leq 9 \exp\left(-\frac{(\epsilon-\kappa)^2}{4\sigma^2}\right) \left(\frac{\pi}{2}\right)^{d(d-1)}\ .
\end{equation}.
\end{corollary}
\begin{proof}
We note  that for $\dproj(\h{V},\h{I})\geq\ep$ and $\kappa\leq\ep$ we have $B(\h{V},\kappa)\subset B(\h{I},\ep-\kappa)^c$ and consequently 
\begin{equation}\label{eq:firstINEQ}
\int_{B(\h{V},\kappa)} \dt \mu (\h{U}) %\tilde{\F}^\sigma 
\F^\sigma_k
(\h{U})\leq
    \int_{B(\h{V},\kappa)} \dt \mu (\h{U}) %|\tilde{\F}^\sigma
    |\F^\sigma_k
    (\h{U})| \leq \int_{B(\h{I},\ep-\kappa)^c} \dt \mu (\h{U}) %|\tilde{\F}^\sigma|
    |\F^\sigma_k(\h{U})|\ .
\end{equation}
Then the assertion follows from property 2 of the above theorem applied with  $ \ep-\kappa$ in place with $\ep$.
\end{proof}

We are now ready to prove the main result of this section. In the course of the proof we will make use of properties 1, 2 and 3 listed above. Property 4, which bounds the second norm of $\F_\sigma$ will be used in the subsequent section while discussing connection between approximate designs and $\ep$-nets.

\begin{Theorem}[Exact $t$-expanders define $\ep$-nets for sufficiently large $t$]\label{th:Texact}\label{thm:exact-designs}Let $\ep\leq1$ and let $\nu$ be a measure on $\UU(d)$ which is an exact $t$-design with 
	\begin{equation}\label{eq:t-eps-d}
t\geq  5\frac{d^{5/2}}{\ep} \tau(\ep,d)\ , 
	\end{equation}
for $\tau(\ep,d)=\log\left(6C/\ep\right)^{\frac{1}{2}}\sqrt{\frac{1}{32}\log\left(6C/\ep\right)^{\frac{1}{2}}+\log\left(\frac{d}{\ep} \log\left(6C/\ep\right)^{\frac{1}{2}} \right)}$, where $C=9\pi$ is the constant appearing in Fact \ref{fact:Szarek}. 

Then, the set of unitary channels from the support of $\nu$, $\lbrace{\h{V} \rbrace}_{V\in\mathrm{supp}(\nu)}$ forms an $\ep$-net in $\U(d)$ with respect to the distance $\dproj$ defined in  Eq.\eqref{eq:opDIST}.
\end{Theorem}
Recall that for a discrete measure $\nu$ we have simply $V\in\mathrm{supp}(\nu)$ iff $\nu(V)>0$. 

\begin{proof} Assume that $\nu$ is an exact $t$-design and that the set of unitary channels from the support of $\nu$, $\lbrace{\h{V} \rbrace}_{V\in\mathrm{supp}(\nu)}$ is not an $\ep$-net. Let $\h{V}_0$ be a unitary channel that cannot be $\ep$-approximated by elements form the support of $\nu$: $\dproj(\h{V}_0,\h{V}_i)\geq\ep$. Let $\F^{\sigma}_k$ be the function satisfying conditions described in Theorem \ref{th:FUNC}. By  Fact \ref{fact:Szarek} and Lemma \ref{prop:volupperbound} we have 
\begin{equation}\label{eq:lowerboundBALL}
  \left(\frac{\epsilon}{2C}\right)^{d^2-1} \leq \vol(\ball{\h{V}_0}{\epsilon/2}) = \int_{\ball{\h{V}_0}{\epsilon/2}} \dt \mu(\h{U}) \left(T_\nu \F^{\sigma}_k \right) (\h{U})\ ,
\end{equation}
where $k\leq t$. On the other hand by Lemma \ref{lem:upperbound} we have
\begin{equation}
	\int_{\ball{\h{V}_0}{\epsilon/2}} \dt \mu(\h{U}) (T_{\nu} \F^{\sigma}_k) (U)  \leq \max_{\h{V}:\dproj(\h{V},\h{I})\geq \ep} \int_{\ball{\h{V}}{\epsilon/2}} \dt \mu(U) \F^{\sigma}_k(U) \ .
\end{equation}
Using  Corollary \ref{cor:far-away-ball} for
\begin{equation}\label{eq:settingK}
   k=5 \frac{d^{\frac{3}{2}}}{\sigma} \sqrt{\frac{1}{32}\frac{\ep^2}{d^2\sigma^2} + \log(\frac{1}{\sigma})}\ ,
\end{equation}
we get
\begin{equation}\label{eq:upbound}
\max_{\h{V}:\dproj(\h{V},\h{I})\geq \ep} \int_{B(\h{V},\epsilon/2)} \dt \mu (\h{U}) \F^\sigma_k (\h{U}) \leq 9  \exp\left(-\frac{\epsilon^2}{16\sigma^2}\right) \left(\frac{\pi}{2}\right)^{d(d-1)} \ .
\end{equation}
As $\sigma$ decreases (and the degree $k$ increases according to \eqref{eq:settingK}) eventually the right-hand side of \eqref{eq:upbound} becomes smaller than the lower bound $\left(\frac{\epsilon}{2C}\right)^{d^2-1}$ from \eqref{eq:lowerboundBALL}. In particular, by inserting $\sigma$ which satisfies 
\begin{equation}\label{eq:fullinequality}
\frac{1}{2}\left(\frac{\epsilon}{10\pi}\right)^{d^2-1}\geq 9 \exp\left(-\frac{\epsilon^2}{16\sigma^2}\right) \left(\frac{\pi}{2}\right)^{d(d-1)}  \ .
\end{equation}
 to Eq.\eqref{eq:settingK} we get (by contradiction) the degree $k_\ast$ such that that unitaries from the support of an  exact $k\geq k_\ast$ design form an $\ep$ net in $\U(d)$. It is easy to see that taking
\begin{equation}\label{eq:OptSigma}
    \sigma\leq \sigma_\ast(d,\ep)=\frac{\ep}{d}\frac{1}{\log(6C/\ep)^{\frac12}}\ ,
\end{equation}
suffices to satisfy \eqref{eq:fullinequality}. Inserting $\sigma_\ast$ to \eqref{eq:settingK} gives $k_\ast$ equal to the right-hand side of inequality \eqref{eq:t-eps-d}. Finally, we remark that dropping the factor of $1/2$ in inequality \eqref{eq:fullinequality} yields essentially identical scaling.

\end{proof}

\section{Approximate t-designs and epsilon-nets}
\label{sec:approx-designs}

In this section, we will establish even a closer connection between approximate $t$- designs and $\ep$-nets. Specifically, we prove that under suitable conditions approximate $t$-expanders define $\ep$-nets and vice versa.

We first extend the reasoning established in the preceding section. Namely,  for {\it approximate} 
\tdesigns\ the 
integral $\I(\nu,\ep,k,\h{V}_0)$ from Eq. \eqref{eq:KeyInt} is not anymore equal $\vol(\ball{\h{V}_0}{\ep})$. Therefore, we need to argue  that the integral is not too small with respect to the volume.  This is expected, as $\nu$ is almost a $k$-design. The following Lemma expresses this intuition quantitatively.

\begin{Lemma} \label{lem:int-lowerbound}Let $\nu$ be an arbitrary measure on $\UU(d)$ which is a $\delta$-approximate unitary $k$-expander. Then for arbitrary $\h{V}\in\U(d)$, $\ep\in[0,2]$ and a function $\Phi\in \hcal_k$  (i.e. a balanced polynomial of degree at most $k$ in $U$ and in $\bar{U}$) satisfying 
\begin{equation}\label{eq:phi-normalization}
	\int_{\U(d)} \dt \mu(\h{U}) \Phi(\h{U}) = 1\  ,
\end{equation}
we have the following inequality 
\begin{equation}\label{eq:desAVERineq}
\left|\int_{\ball{\h{V}}{\ep}} \dt \mu(\h{U}) \left(T_\nu \Phi \right) (\h{U}) - \vol(\ball{\h{V}}{\ep})  \right| \leq   \delta  \sqrt{\vol(\ball{\h{V}}{\ep})} \left\| \Phi\right\|_2\ .
\end{equation}
\end{Lemma}

\begin{proof} 
	We start with the identity
	\begin{equation}\label{eq:identityINNER}
	\int_{B(\h{V},{\ep})}\dt\mu(\h{U})  -\int_{B(\h{V},\ep)}\dt\mu(\h{U}) T_{\nu} (\Phi)(\h{U}))  = \langle 1-T_\nu \Phi ,I_{\ball{\h{V}}{\ep}}\rangle\  ,
	\end{equation}
	where $\<\cdot,\cdot\>$ is the inner product in $L^2(\U(d))$, and $I_A$ is the indicator function of a set $A\subset\U(d)$. 	Using  Cauchy-Schwartz inequality we obtain  
	\begin{equation}\label{eq:CSinnerprod}
	\left|\langle 1-T_\nu \Phi ,I_{\ball{\h{V}}{\ep}}\rangle \right|\leq \|1-T_\nu \Phi\|_2 ||I_{\ball{\h{V}}{\ep}}||_2 = \|1-T_\nu \Phi\|_2 \sqrt{\vol\left(\ball{\h{V}}{\ep}\right)}\ .
	\end{equation}
	Furthermore, condition \eqref{eq:phi-normalization}, the assumption $\Phi\in\H_k$ and definitions of $T_\mu$ and the infinity norm allows us to write an estimate 
	\begin{equation}\label{eq:2normEstmate}
\|1-T_\nu \Phi\|_2 =  \|(T_\mu-T_\nu) \Phi\|_2 \leq \|(T_\mu-T_\nu)_{\H_k}\|_\infty \|\Phi\|_2 = \delta \|\Phi\|_2\ ,
	\end{equation}
	where in the last equality we used Proposition \ref{prop:gaps}. By combining bounds \eqref{eq:2normEstmate} and \eqref{eq:CSinnerprod} with \eqref{eq:identityINNER} we obtain the desired result, i.e. we get \eqref{eq:desAVERineq}. 
\end{proof}

With the help of the above Lemma and due to properties of a carefully chosen polynomial approximation to the Dirac delta given in Theorem \ref{th:FUNC} we are in the position to prove the main result of this section. 

\begin{Theorem}[$\delta$-approximate $t$-expanders define $\ep$-nets]
\label{th:approx-design}
Suppose that a measure $\nu$ on $\UU(d)$ is a $\delta$-approximate unitary $t$-expander with
	\begin{equation}\label{eq:t-app-eps-d}
t\geq  5\frac{d^{5/2}}{\ep} \tau(\ep,d)\ ,\ 
  \delta\leq   \frac{1}{32} \left(
\frac{\ep^\frac32}{4C\log\left(\frac{6C}{\ep}  \right)^\frac12 d}\right)^{d^2-1} 
\end{equation}
and $\tau(\ep,d)=\log\left(6C/\ep\right)^{\frac{1}{2}}\sqrt{\frac{1}{32}\log\left(6C/\ep\right)^{\frac{1}{2}}+\log\left(\frac{d}{\ep} \log\left(6C/\ep\right)^{\frac{1}{2}} \right)}$, where $C=9\pi$ is the constant appearing in Fact \ref{fact:Szarek}. Then, the set of unitary channels $\lbrace{\h{V} \rbrace}_{V\in\mathrm{supp}(\nu)}$ forms an $\ep$-net in $\U(d)$ with respect to the distance $\dproj$ defined in  \eqref{eq:opDIST}.
\end{Theorem}

 \begin{remark} A similar result follows from arguments given in the proof of Theorem 5 in \cite{HArrowHasstings2008}. From careful analysis of the arguments presented there it can be shown that $\delta$-approximate $t$-expanders  with $t\simeq d^3/\ep^2$  and $\delta \simeq (\ep/\sqrt{d})^{2d^2}$  define $\ep$-nets with respect to the distance between unitary channels induced from the Hilbert-Schmidt norm 
 \begin{equation}
 \tilde{\mathrm{D}}\left(\h{U},\h{V}\right)\coloneqq \min_{\varphi\in[0,2\pi)} \| U -\exp(\ii \varphi) V \|_{\mathrm{HS}}\ .
 \end{equation}
 In order to attain the scaling claimed above one can tight bounds on volumes of Hilbert-Schmidt balls in $\UU(d)$ (cf. \cite{SzarekBook}  Theorem 5.11).   Our result gives a more favorable behaviour of $t$ and $\delta$ in $d$ and $\epsilon$.  This is because the inequality  $\mathrm{D}\left(\h{U},\h{V}\right)\leq \tilde{\mathrm{D}}\left(\h{U},\h{V}\right)$  implies that $\ep$-net with respect to distance $\tilde{\mathrm{D}}$ is automatically $\ep$-net with respect to  distance $\mathrm{D}$.
 
 \end{remark}

\begin{proof} We proceed analogously as in the proof of Theorem \ref{th:Texact}. Assume that $\nu$ is a $\delta$-approximate $t$-design and that the set of unitary channels from the support of $\nu$, $\lbrace{\h{V} \rbrace}_{V\in\mathrm{supp}(\nu)}$, is not an $\ep$-net. We choose $\h{V}_0$ to be a unitary channel that cannot be $\ep$-approximated by elements form the support of $\nu$: $\dproj(\h{V}_0,\h{V}_i)\geq\ep$. Moreover, we take $\F^{\sigma}_k$ to be the polynomial function described in Theorem \ref{th:FUNC} for $\sigma=\sigma_\ast(d,\ep)$ (c.f Eq.\eqref{eq:OptSigma}) and $k=5\frac{d^{5/2}}{\ep} \tau(\ep,d)$. 

From Lemma \ref{lem:int-lowerbound} and Eq.\eqref{eq:L2norm} it follows that 
\begin{equation}\label{eq:downdelta}
\vol(\ball{\h{V}_0}{\ep/2})-\|\F_k^\sigma \|_2\,\delta\, \vol(\ball{\h{V}_0}{\ep/2})^{\frac{1}{2}}\leq \int_{\ball{\h{V}_0}{\epsilon/2}} \dt \mu(\h{U}) \left(T_\nu \F^{\sigma}_k \right) (\h{U})\ ,
\end{equation}
for any $k\leq t$. On the other hand, by repeating the same arguments as in the proof of Theorem \ref{th:Texact} we have 
\begin{equation} \label{updelta}
	\int_{\ball{\h{V}_0}{\epsilon/2}} \dt \mu(\h{U}) (T_{\nu} \F^{\sigma}_k) (\h{U})   \leq \exp\left(-\frac{\epsilon^2}{16\sigma^2}\right) \left(\frac{\pi}{2}\right)^{d(d-1)}\ .
\end{equation}\label{eq:suffLHS}
It is now clear that if $\delta$ is such that 
\begin{equation} \label{eq:intermediateINEQ}
\|\F^\sigma_k \|_2\,\delta\, \vol(\ball{\h{V}_0}{\ep/2})^{\frac{1}{2}} \leq \frac{1}{2} \vol(\ball{\h{V}_0}{\ep/2})\ ,
\end{equation}
then we obtain inequality \eqref{eq:fullinequality}. However, already form the proof of Theorem \ref{th:Texact} we know that this inequality cannot be satisfied for $\sigma=\sigma_\ast(d,\ep)$ and $k=5\frac{d^{5/2}}{\ep} \tau(\ep,d)$ and hence unitaries from the support of $\nu$ must form an $\ep$-net. We conclude the proof observing that
\begin{align}
\delta\leq \frac{1}{32} \left(
\frac{\ep}{4C}\right)^{\frac{d^2-1}{2}} \sigma_*(d,\ep)^{-d^2+1}   
\end{align}
and hence also 
\begin{align}
    \delta\leq 
    \frac{1}{32} \left(
\frac{\ep^{\frac{3}{2}}}{4C\log\left(\frac{6C}{\ep}\right)^\frac12 d}\right)^{d^2-1} 
\end{align}
is a sufficient condition for validity \eqref{eq:intermediateINEQ}. This can be verified easily using Fact \ref{fact:Szarek}.

\end{proof}

We now prove the statement in the opposite direction to the one given in Theorem \ref{th:approx-design}.

\begin{Theorem}[$\ep$-nets in the set of unitary channels can be used to define $(2\ep t)$-approximate unitary $t$-expanders]\label{th:netsAREdesigns}
Consider a subset of unitary channels $\S\subset\U(d)$ which forms an $\ep$-net in $\U(d)$ with respect to the distance $\dproj$ defined in \eqref{eq:opDIST}. Let $t$ be arbitrary natural number. Then, there exists an ensemble $\E=\lbrace{\nu_i ,\h{V}_i\rbrace}$, with $\h{V}_i\in\S$, which forms:  (a)  $\ep t$-approximate $t$-design (see Eq. \eqref{eq:appDIAMOND}); (b)    $2\ep t$-approximate $t$-expander.
\end{Theorem}

\begin{proof}
We present an explicit (although possibly computationally inefficient) construction of an ensemble of gates from $\S$ which will form: (a) $2\ep t$-approximate $t$-design and (b) $2\ep t$-approximate $t$-expander. First, we note that, by definition of $\ep$-net, elements for $\S$ define a cover of $\U(d)$ via balls of radius \emph{at most} $\ep$:
\begin{equation}
\U(d)= \bigcup_{\h{V}\in\S} \ball{\h{V}}{\ep}\ .
\end{equation}
Since $\U(d)$ is a compact space, we can take a finite collection of gates $\lbrace{ V_i\rbrace}_{1=1}^K\subset\S$ such that 
\begin{equation}\label{eq:finiteCOVER}
\U(d)= \bigcup_{i=1}^K \ball{\h{V}_i}{\ep}\ .
\end{equation}
We note that $K$ in the above equation is some, in general unknown, but \emph{finite} number. We now use Eq. \eqref{eq:finiteCOVER} to define a disjoint collection of subsets $\V_i$ that cover $\U(d)$. We set
\begin{equation}
\V_1 = \ball{\h{V}_1}{\ep}\ ,\ \V_{k+1}=\ball{\h{V}_{k+1}}{\ep}\setminus \bigcup_{i=1}^k \ball{\h{V}_{i}}{\ep}\ , k=1,\ldots,K-1\ .
\end{equation}
By the construction we have $\V_i \cap \V_j=\emptyset$ whenever $i\neq j$. Moreover, $\bigcup_{i=1}^K \V_i =\U(d)$, while it might also happen that $\V_i = \emptyset$ for $i>K'$, where $K'=\min\lbrace{k|\ \bigcup_{i=1}^k \ball{\h{V}_{i}}{\ep} =\U(d)\rbrace}$. From the definition we have $\V_i \subset \ball{\h{V}_{i}}{\ep}$ and hence for $i\leq K'$ 
\begin{equation}\label{eq:distCOND}
\h{U}\in\V_i \Longrightarrow  \dproj \left(\h{U},\h{V}_i\right)\leq\ep\ .
\end{equation}

We are now ready to define a discrete ensemble that is a desired approximate $t$-design. We set $V_i$ to be \emph{any} unitary operator which is compatible with channel $\h{V}_i$ and define
\begin{equation}
\E=\lbrace{\nu_i, V_i \rbrace}_{i=1}^{K'}\ ,\ \text{where}\ \nu_i =\mu_P(\V_i)\ .
\end{equation}
The normalization of the probability $\nu_i$ follows from the construction.  Let us finally upper bound $\left\|\Delta_{\nu,t}-\Delta_{\mu,t} \right\|_\diamond$ and $\left\|T_{\nu,t}-T_{\mu,t} \right\|_\infty$. We start with the former. Inserting definitions of channels $\Delta_{\nu,t} , \Delta{\mu,t}$ we get
\begin{equation}
\Delta_{\nu,t}-\Delta_{\mu,t}= \sum_{i=1}^{K'} \mu_P(\V_i) \h{V}_{i}^{\otimes t} -\int_{\UU(d)}d\mu(\h{U})\h{U}^{\otimes t}=\sum_{i=1}^{K'} \int_{\V_i}d\mu(\h{U})\left[ \h{V}_i^{\ot t} -\h{U}^{\ot t} \right]\ .
\end{equation}
From the above we obtain
\begin{equation} \label{eq:intermedCHANNEL}
\left\|\Delta_{\nu,t}-\Delta_{\mu,t} \right\|_\diamond = 
\left\| \sum_{i=1}^{K'} \int_{\V_i}d\mu(\h{U})\left[ \h{V}_i^{\ot t} -\h{U}^{\ot t} \right] \right\|_\diamond  \leq
\sum_{i=1}^{K'} \int_{\V_i}d\mu(\h{U}) \left\|\h{V}_i^{\ot t} -\h{U}^{\ot t} \right\|_\diamond \ .
\end{equation}
For every summand in the last expression we use Eq. \eqref{eq:distCOND} and the telescopic property of the diamond norm and \eqref{eq:equivDISTANCE} 
\begin{equation}
\|\h{V}_i^{\ot t} -\h{U}^{\ot t} \left\|_\diamond \leq
 t\|\h{V}_i -\h{U} \right\|_\diamond \leq t\dproj(\h{V}_i,\h{U})\leq 2 t\ep\ .
\end{equation}
Inserting the above bound to \eqref{eq:intermedCHANNEL} proves $\left\|\Delta_{\nu,t}-\Delta_{\mu,t} \right\|_\diamond \leq 2t \ep$. 
The proof for the case of operator norm $\left\|T_{\nu,t}-T_{\mu,t} \right\|_\infty$ is analogous: 

\begin{equation}
T_{\nu,t}-T_{\mu,t}= \sum_{i=1}^{K'} \mu_P(\V_i) V_{i}^{\ot t} \ot \bar{V}_{i}^{\ot t}-\int_{\UU(d)}d\mu(U)U^{\ot t} \ot \bar{U}^{\ot t}=\sum_{i=1}^{K'} \int_{\varphi^{-1}(\V_i)}d\mu(U)\left[ V_{i}^{\ot t} \ot \bar{V}_{i}^{\ot t}-U^{\ot t}\ot \bar{U}^{\ot t} \right]\ .
\end{equation}
We have the following chain of inequalities
\begin{equation}\label{eq:chainFIN}
\left\|T_{\nu,t}-T_{\mu,t} \right\|_\infty \leq 
\sum_{i=1}^{K'} \int_{\varphi^{-1}(\V_i)}d\mu(U) \left\| V_{i}^{\ot t} \ot \bar{V}_{i}^{\ot t}-U^{\ot t}\ot \bar{U}^{\ot t} \right\| \leq 
\sum_{i=1}^{K'} \int_{\varphi^{-1}(\V_i)}d\mu(U)  \left\|V_i\otimes\bar{V}_i - U\otimes\bar{U}  \right\|_\infty t \leq 2\ep t\ .
\end{equation}
The second inequality follows from  the well-known telescopic bound (see for example page 27 in \cite{BHH2016})
\begin{equation}
\left\|A^{\ot t}-B^{\otimes t}  \right\|_\infty \leq  \left\| A - B \right\|_\infty t\ ,
\end{equation}
applied for $A=V_i \ot \bar{V}_i$ and $B=U\ot\bar{U}$. The third inequality in Eq.\eqref{eq:chainFIN} follows form Eq.\eqref{eq:opDIST} and the fact that numbers $\lbrace{\mu_P(\V_i\rbrace}_{i=1}^{K'}$ sum up to 1. To see this we first choose the relative phase between $U$ and $V_i$ in such a way that these operators saturate Eq.\eqref{eq:opDIST} and arrive at the bound $\left\|V_i\otimes\bar{V}_i - U\otimes\bar{U}  \right\|_\infty \leq 2\dproj(\h{U},\h{V})$.  Second, we use Eq.\eqref{eq:distCOND} which ensures  $\dproj\left(\h{V}_i, \h{U}\right)\leq \ep$.
\end{proof}

Let $t(\ep,d)$ be the minimal degree of a $t$-design $\nu$  that ensures that unitaries from the support of $\nu$ form an $\ep$-net with respect to diamond norm. We conclude this section with by showing that the dependence of $t(\ep,d)$  on $\ep$ appearing in Theorems \ref{th:Texact} and \ref{th:approx-design} is essentially optimal (for fixed dimension $d$). We also show that the dependence of $t(\ep,d)$ on $d$ in the said theorems is close to being optimal. Specifically, we proved  $t(\ep,d) \propto d^{\frac{5}{2}}$ while the bound given below requires $t(\ep,d) \geq d^2$ (for fixed $\ep$).

\begin{theorem}[Lower bounds on the degree of exact $t$-designs that define $\ep$-nets]\label{th:tightness} Let $\ep\in(0,2]$ and let $t(d,\ep)$ be the minimal degree of an exact unitary $t$-design $\nu$ such that support of $\nu$ forms   $\ep$-net with respect to distance $\dproj$ in $\U(d)$.Then we have the following inequalities
\begin{align}
  t(d,\ep)\geq  \frac{(d^2-1)}{2}\log(\frac{c}{2\ep}) -d ,\  \text{where }&  c=1/87\  , \label{eq:comparNETS} \\ 
  t(d,\ep)\geq  \frac{A}{\ep} -2\ ,\ \text{where }& A=(1-\frac{1}{d}) \geq 0.4\ \label{eq:comparTDESIGN} .
\end{align}
\end{theorem}
\begin{proof}
In order to prove both inequalities we use the fact (proven in Theorem 11 of \cite{scott}) that in any dimension $d$ there exist an exact unitary $t$ designs $\nu_\ast$ whose support has cardinality  upper bound by
\begin{equation}\label{eq:compar1}
\left|\mathrm{supp}(\nu_\ast) \right| \leq \binom{d+t-1}{t}^2\ .
\end{equation}
We now assume that $t$ is the minimal number such that $\mathrm{supp}(\nu_\ast)$ forms an $\ep$-net. As a consequence we have $t\leq t(d,\ep)$.  In order to prove Eq. \eqref{eq:comparNETS} we note under this assumption we have   
\begin{equation}\label{eq:compar2}
|\mathrm{supp}(\nu_\ast)|\geq N_\mathrm{cov} (\ep)\ ,
\end{equation}
where $N_\mathrm{cov}(\ep)$  i.e. is the covering number of $\U(d)$ i.e the minimal cardinality of $\ep$-net in $\U(d)$. Using the standard volume-comparison reasoning (see for example introduction of \cite{Szarek98}) we get
\begin{equation}\label{eq:compar3}
N_\mathrm{cov} (\ep) \geq 1/\vol(\ball{\h{V}}{\ep}) \geq \left(\frac{c}{\ep}\right)^{d^2-1}\ ,
\end{equation}
where in the last inequality we used \eqref{eq:boundONvolume}. Combing above estimates gives
\begin{equation}
\left(\frac{c}{2\ep}\right)^{d^2-1} \leq \binom{d+t-1}{t}^2\ .
\end{equation}
Inserting to the above elementary inequality $\binom{d+t-1}{t}\leq e^{d+t}$, using $t\leq t(d,\ep)$ and taking logarithm from both sides gives Eq. \eqref{eq:comparNETS}.

In order to prove \eqref{eq:comparTDESIGN}  we use the construction from the proof of Theorem \ref{th:netsAREdesigns}. There, we showed that any  $\ep$-net $\S$ can be used to form $\ep K$-approximate $K$-design $\tilde{\nu}$ (in the sense of definition from Eq. \eqref{eq:appDIAMOND}), where $K$ is any natural number. Importantly,  from the arguments given in the proof it follows that
\begin{equation}\label{eq:estimNETdesign}
|\mathrm{supp}(\tilde{\nu})| \leq  |\S| .
\end{equation}
Now, we assume as before that $t$ is such that  $\S=\mathrm{supp}(\nu_\ast)$ is an $\ep$-net, where $\nu_\ast$ is an exact $t$-design that we introduced above. From \eqref{eq:estimNETdesign} and \eqref{eq:compar1} we get
\begin{equation}\label{eq:step4}
|\mathrm{supp}(\tilde{\nu})| \leq \binom{d+t-1}{d}^2\ .
\end{equation}
Now, we use the following lower bound for the cardinality of the support of any  $\delta_\diamond$-approximate unitary $K$-design $\nu$ (see Lemma 26 in \cite{BHH2016}):
\begin{equation}\label{eq:step5}
(1-\delta_\diamond)\binom{d+K-1}{K}^2\leq |\mathrm{supp}({\nu})|\ .
\end{equation}
The measure $\tilde{\nu}$ constructed above is $\ep K$-approximate $K$-design and therefore by combining \eqref{eq:step4} and \eqref{eq:step5} we obtain $(1-\ep K)\binom{d+K-1}{K}^2 \leq \binom{d+t-1}{t}^2$. Assuming $\ep K\leq 1$ we get en equivalent inequality 
\begin{equation}\label{eq:tightINEQ}
(1-2\ep K)^{\frac{1}{2}}\binom{d+K-1}{K} \leq \binom{d+t-1}{t}\ .
\end{equation}
We denote $x=(1-2\ep K)^{\frac{1}{2d}}$ and decompose both sides of the above inequality 
\begin{align}
(1-2\ep K)^{\frac{1}{2}}\binom{d+K-1}{K} = & 
\frac{1}{(d-1)!}\prod_{l=2}^d x \left(  l+K-1\right) \ , \\
\binom{d+t-1}{t}=& \frac{1}{(d-1)!} \prod_{l=2}^d  \left(  l+t-1\right) \ .
\end{align}
By comparing individual terms in the above products we get that if 
\begin{equation}\label{eq:termINEQ}
x (l+K-1) > l+t-1 \ \text{ for all } l=2,\ldots,d\ ,
\end{equation}
then inequality \eqref{eq:tightINEQ} cannot hold. Therefore \eqref{eq:tightINEQ} implies that inequalities \eqref{eq:termINEQ} cannot be simultaneously satisfied. Now, from $x\leq 1$ it follows that the violation of \eqref{eq:termINEQ} implies 
\begin{equation}\label{eq:almostFINALineq}
x K-(d-1)(1-x) \leq t\ .
\end{equation}
We conclude the proof by setting $K$ to be the largest natural number such that $x\geq 1-\frac{1}{d}$. This condition condition is equivalent to 
\begin{equation}\label{eq:almostFINeK}
(1-2\ep K)^{\frac{1}{2}} \geq \ \left(1-\frac{1}{d}\right)^{d-1}\ .
\end{equation}
Note that for every natural number $n$,  $e^{-1}> \left( 1-\frac{1}{n}\right)^n$, where $e$ is the Euler number. Therefore to satisfy $x\geq 1-\frac{1}{d}$ it suffices to set an integer $K$ satisfying
\begin{equation}
K\leq \frac{e^2-1}{2 e^2}\frac{1}{\ep} \leq \frac{2}{5 \ep}\ . 
\end{equation}
Inserting this to Eq. \eqref{eq:almostFINALineq} and utilizing the fact that $K$ must be an integer finally gives
\begin{equation}
 \frac{2}{5\ep} -2 \leq t \ ,
\end{equation}
which gives \eqref{eq:comparTDESIGN} since $t$ was chosen such that  $t\leq t(d,\ep)$. 

\end{proof}

\section{Sequences of gates and epsilon-nets}\label{sec:sequances}
Let $\G\subset \UU(d)$ be the support of a probability measure $\nu_\G$ on $\UU(d)$. Then the words of length $l$ composed of gates from the set $\G$, we denote them by $\G_l$, constitute the support of $\nu_\G^{*l}$. In this section we explore properties of sets $\G_l$ and formulate an inverse-free version of the Solovay-Kitaev theorem.
Our first result is the following
\begin{proposition} 	\label{prop:prop-words}Let $\nu$ be an arbitrary measure on $\UU(d)$ which is a $\delta$-approximate unitary $t$-expander with 
	\begin{equation}\label{t-sequence}
t\geq  5\frac{d^{5/2}}{\ep} \tau(\ep,d)\ ,
	\end{equation}
and $\tau(\ep,d)=\log\left(6C/\ep\right)^{\frac{1}{2}}\sqrt{\frac{1}{32}\log\left(6C/\ep\right)^{\frac{1}{2}}+\log\left(\frac{d}{\ep} \log\left(6C/\ep\right)^{\frac{1}{2}} \right)}$, where $C=9\pi$ is the constant appearing in Fact \ref{fact:Szarek}. Then for
\begin{gather}\label{lcondition}
    l\geq\frac{\log(32)+(d^2-1)\left(2\log(\frac{1}{\epsilon}) +\log(4C^{\frac32} d )   \right) }{1-\delta}\ ,
\end{gather}
the set of unitary channels in the support of $\nu^{*l}$, $\lbrace{\h{V} \rbrace}_{V\in\mathrm{supp}(\nu^{*l})}$, forms an $\ep$-net in $\U(d)$ with respect to the distance $\dproj$ defined in  \eqref{eq:opDIST}. 
\end{proposition}
\begin{proof}
	By Fact \ref{fact:delta-l} it follows that the support of $\nu^{*l}$ defines $\delta^l$-approximate $t$-expander. Next, if $l$ satisfies \eqref{lcondition} one easily checks that  for $l$ satisfying \eqref{lcondition} we have
	 \begin{equation}
\delta^l \leq   \frac{1}{32} \left(
\frac{\ep^{\frac{3}{2}}}{4C\log\left(\frac{6C}{\ep}\right)^\frac12 d}\right)^{d^2-1} \ .
\end{equation}
Therefore by Theorem \ref{th:approx-design} the support of $\nu^{*l}$ is an $\epsilon$-net.
\end{proof}

We now reformulate a strong result by Peter Varju (Theorem 6 in  \cite{Varju2013}) in the language of approximate $t$-expanders (see part \ref{app:spectralGAPdacay} of Appendix the for details)
 
\begin{Theorem}[Slow decay of the spectral gap]\label{thm:Varju}
Let $\nu$ be arbitrary probability measure on $\U(d)$. Then, there exist a natural number $t_0$ and a constant $D>0$ (depending only on $d$) such that for all natural $t>t_0$ we have
\begin{equation}
\left\|T_{\nu,t}-T_{\mu,t} \right\|_\infty  \leq 1- \frac{1-\delta(\nu,t_0)}{D \log(t)^2}\ ,
\end{equation}
where $\delta(\nu,t_0) =\left\|T_{\nu,t_0}-T_{\mu,t_0} \right\|_\infty$. In other words the gap of the random walk generated by $\nu$ cannot decreases faster than $\log^{-2}(t)$  for large $t$. 
\end{Theorem}

The above Theorem in conjunction with Proposition \ref{prop:prop-words} allows us to state an inverse-free version of the Solovay-Kitaev theorem. We note that an equivalent result has already appeared in \cite{Varju2013}. It was, however,  obscured by the mathematical character of that work. 

\begin{Theorem}[Non-constructive inverse-free Solovay-Kitaev]\label{th:InvFree}
Let $\G\subset\U(d)$ be a universal gate-set in $\U(d)$ (not necessarily symmetric i.e $\h{V}\in\G$ \emph{does not} imply $\h{V}^{-1}\in\G$). Let $\nu_\G$ be a uniform measure on $\G$.  Then, there exist absolute constants $A, B>0$ (depending on $d$), such that for 
\begin{equation}\label{SKi}
 l\geq A \frac{\log^3\left(\frac{1}{\ep}\right)+B}{1-\delta(\nu_\G,t_0)}
\end{equation}
the set $\G_l$ forms an $\ep$-net in $\U(d)$. 
\end{Theorem}

\begin{proof}
Proposition \ref{prop:prop-words} tells us that if $\G$ is a $\delta(\nu_\G,t)$-approximate $t$-expander, with $t$ given by \eqref{t-sequence}, then for 
\begin{equation}\label{lt1}
		l\geq\frac{\log(32)+(d^2-1)\left(2\log(\frac{1}{\epsilon}) +\log(4C^{\frac32} d\right) ) }{1-\delta(\nu_\G,t)}\ ,
\end{equation}
$\G_l$ is an \enet.
On the other hand,  Theorem  \ref{thm:Varju} allows us to bound $1-\delta(\nu_\G,t)$ as follows
\begin{equation}\label{var2}
1-\delta(\nu_\G,t) \geq \frac{1-\delta(\nu_\G,t_0)}{D \log^2(t)}\ .
\end{equation}
Combining \eqref{var2} with \eqref{lt1} and making use of \eqref{t-sequence} we obtain \eqref{SKi}.
\end{proof}
\begin{remark}
Another application of our results that relate $\epsilon$-nets and $\delta$-approximate $t$-designs is connected to universality of gate-sets $\G$. In \cite{SK2017,SK2017bis} it was shown that the necessary condition for universality of a gate-set $\G\subset\U(d)$ is $\mathrm{dim}\left(\mathrm{Comm}(U\otimes\bar{U}|U\in\G)\right)=2$. Moreover, sets $\G$ that satisfy the necessary condition are either universal or they generate finite subgroups of $\U(d)$. Furthermore, in order to verify universality of a set $\G$ that satisfies the necessary condition one has to check that there is $l$ such that $\G_l$ forms an $\epsilon$-net with $\epsilon\leq\frac{1}{2\sqrt{2}}$. Thus using Proposition \ref{prop:prop-words} a gate-set $\G$ satisfying the necessary condition is universal iff it is a $\delta$-approximate $t$-expanders with $\delta<1$ and $t$ given by \eqref{t-sequence} with $\epsilon=\frac{1}{2\sqrt{2}}$. Otherwise $\mathcal{G}$ generates a finite group (this follows from Lemma 4.8 of \cite{SK2017}). Therefore checking universality of $\G$ can be reduced to two steps 1) checking if $\mathrm{dim}\left(\mathrm{Comm}(U\otimes\bar{U}|U\in\G)\right)=2$ and  2) checking if $\delta(\nu_{\G},t)<1$ for $t$ given by \eqref{t-sequence} with $\epsilon=\frac{1}{2\sqrt{2}}$.  
\end{remark}

\section{Random circuits and approximate designs}\label{sec:randCIRC}
\def\gap{{\rm g}}
\def\mun{\mu^{(n)}}
In this section we shall prove that random circuits composed of universal gates are approximate $t$-designs \emph{without} any assumptions on the set of gates (i.e. unlike in \cite{BHH2016} we shall not assume that the set contains inverses or that the unitaries have algebraic entries). Importantly, we are not using  result due to Bourgain and Gomburd \cite{Bourgain2011} who proved that universal set of gates has a gap that does not diverge with growing $t$ under the assumption of algebraic entries. However, from the results of Varju \cite{Varju2013} is possible to prove a lower bound on the the gap, which vanishes very slowly with with $t$, yet without any assumptions (see Theorem \ref{thm:Varju}). 

Let $\nu_\G$ be uniform measure on set of gates $\G$, and let $\G^\dagger$ be the set of inverses of gates from $\G$, and $\nu_{\G^\dagger}$ uniform measure on $\G^\dagger$. Note that $\nu_{\G 
\G^\dagger}=\nu_\G*\nu_{\G^\dagger}$. We shall also employ the following (well-known in the mathematics community,  see e.g. \cite{Varju2013}) Lemma in order to remove the assumption that the set of gates contains inverses.

\begin{Lemma}[Bounds on the spectral gap without assuming a symmetric gate-set]
\label{lem:mu-tilde-mu}
Let $\G$ be arbitrary finite gate set in $\UU(d)$. Let $\nu_\G$ and $\nu_{\G^\dagger}$ be two measures uniformly supported on $\G$ and $\G^\dagger$ respectively. We have the following inequalities 
\begin{equation}
\label{eq:delta-mult}
\delta(\nu_\G,t)^2=\delta(\nu_{\G\G^\dagger},t)
\end{equation}
\begin{equation}
\label{eq:gap-mult}
\frac12 \gap(\nu_\G*\nu_{\G^\dagger},t)\leq \gap(\nu_\G,t) \leq  \gap(\nu_\G*\nu_{\G^\dagger},t)
\end{equation}
\end{Lemma}

\begin{proof}
Recall that $\delta(\nu_\G,t)=\|T_{\nu_\G,t}-T_{\mu,t}\|_\infty$. 
From definition of $T_{\nu,t}$ 
we have for any measures $\nu,\nu'$
\begin{eqnarray}
T_{\nu*\nu',t}=T_{\nu,t} T_{\nu',t},\quad
T_{\nu,t}=\Ttmu\oplus  T_{\nu,t}^\perp,\quad  \Ttmu=\Ttmu^2=\Ttmu^\dagger,
\end{eqnarray}
so that $\delta(\nu,t)=\|T_{\nu_\G,t}^\perp\|$. It also immediately follows that $T_{\nu*\nu',t}^\perp=T_{\nu,t}^\perp T_{\nu',t}^\perp$. 
Further, from definition of moment operators $T_{\nu_\G,t}$ we have 
\begin{equation}
    T_{\nu_{\G^\dagger},t}=T^\dagger_{\nu_\G,t}
\end{equation}
which gives
\begin{equation}
    T_{\nu_{\G^\dagger},t}^\perp =
    (T_{\nu_{\G},t}^\perp)^\dagger.
\end{equation}
We then also get 
\begin{equation}
    T_{\nu_{\G \G^\dagger},t}^\perp= 
    T_{\nu_{\G}* \nu_{\G^\dagger},t}^\perp=
    T_{\nu_{\G},t}^\perp T_{\nu_{\G^\dagger},t}^\perp=
     T_{\nu_{\G},t}^\perp (T_{\nu_{\G},t}^\perp)^\dagger
\end{equation}
and hence
\begin{equation}
    \delta(\nu_{\G\G^\dagger},t)=\|T_{\nu_{\G\G^\dagger},t}^\perp\|= 
    \|T_{\nu_{\G},t}^\perp (T_{\nu_{\G},t}^\perp)^\dagger\| =
    \|T_{\nu_{\G},t}^\perp\|^2=\delta(\nu_\G,t)^2.
\end{equation}
We have thus proved the formula \eqref{eq:delta-mult}. Now, the first inequality of \eqref{eq:gap-mult} follows by  definition of $\gap$: $\gap=1-\delta$ and the use of  $\sqrt{x}\leq \frac12+\frac12x$ for $x\geq 0$, while the second one follows from  $x^2\leq x$ for $x\leq 1$, 
\end{proof}

We shall now consider two layouts for random circuits acting on $n$ qudits, composed of two qudit gates form set $G$: 
(i) local random circuits and (ii) parallel random circuits. 
Local random circuits are the following. We pick uniformly at random 
two neighboring qudits, and apply gate chosen from $\G$ according to 
uniform measure. The resulting measure we shall denote by 
$\nu_{loc}^{n}(\G)$.    Let us also denote by $\nu_{loc}^{n}(\mu)$ similarly defined measure, but with $\nu_\G$ replaced with Haar measure on two qudits $\mu$.
Regarding parallel random circuits, we apply with probability $1/2$
either unitary $U_{12}\otimes U_{34}\otimes\ldots \otimes U_{n-1,n}$
or $U_{23}\otimes U_{45}\otimes\ldots \otimes U_{n-2,n-1}$ where each $U_{ij}$
is picked independently from $\G$ according to $\nu_\G$. 
The resulting measure we shall denote by $\nu_{par}^{n}(\G)$
and if $\nu_\G$ is replaced by Haar, by $\nu_{par}^{n}(\mu)$.
We shall now relate the gaps of two steps of such circuits
to the gap of one step of circuit with measure $\nu_{\G\G^\dagger}$. 
In this way we shall reduce the problem to gate sets with inverses so that 
we then can invoke results on such circuits from \cite{BHH2016}.

\def\nulocg{\nu_{loc}^{(n)}(\G)}
\def\nulocgtil{\nu_{loc}^{(n)}(\tilde\G)}
\def\nulocgd{\nu_{loc}^{(n)}(\G^\dagger)}
\def\nulocgg{\nu_{loc}^{(n)}(\G\G^\dagger)}
\def\nuloch{\nu_{loc}^{(n)}(\mu)}
\def\nuparg{\nu_{par}^{(n)}(\G)}
\def\nupargtil{\nu_{par}^{(n)}(\tilde\G)}
\def\nupargd{\nu_{par}^{(n)}(\G^\dagger)}
\def\nupargg{\nu_{par}^{(n)}(\G\G^\dagger)}
\def\nuparh{\nu_{par}^{(n)}(\mu)}

\begin{Lemma}[Bound on the gap of random local quantum circuits for non symmetric gate-set]
Let $\nulocg$ be a measure describing random local quantum circuits generated two qudit gate-set $\G$ (not necessarily symmetric).  We have the following lower bound
\label{lem:nonHerm-Herm-loc}
\begin{equation}
\label{eq:nonHerm-Herm-loc}
    \gap\left(\nulocg*\nulocgd\right) \geq \frac1{n-1} \gap(\nulocgg)
\end{equation}
\end{Lemma}
\begin{proof}
By definition of $\nulocg$ we have 
\begin{equation}\label{eq:WALKgen}
    T_{\nulocg,t}=\frac{1}{n-1}\sum_{i=1}^{n-1} A_{ii+1}\ .
\end{equation}
where $A_{ii+1}=T_{\nu_\G,t}$ with $\G$ acting on  qudits $i$ and $i+1$. Note that $A_{i i+1}$ is not necessarily Hermitian. Denoting for clarity 
by $P_{Haar}^\perp$ the complement of $\Ttmu$ we write
\begin{eqnarray}
&&1- \gap\left(\nulocg*\nulocgd)\right)=\|P^\perp_{Haar}  T_{\nulocg,t} T_{\nulocg,t}^\dagger P^\perp_{Haar}\|
=\nonumber \\
&&\|P^\perp_{Haar} \left(\frac{1}{(n-1)^2}\sum_{i} A_{ii+1} A_{ii+1}^\dagger  + \frac{1}{(n-1)^2}\sum_{i\not=j} A_{ii+1} A_{jj+1}^\dagger\right) P^\perp_{Haar}\|\leq 
\nonumber \\
&&\leq \|P^\perp_{Haar} \left(\frac{1}{(n-1)^2}\sum_{i} A_{ii+1} A_{ii+1}^\dagger\right) P^\perp_{Haar}\| + \frac{(n-1)^2-(n-1)}{(n-1)^2} 
=\frac{1}{n-1}(1-\gap(\nulocgg) + \frac{(n-1)^2-(n-1)}{(n-1)^2}\nonumber \ .
\end{eqnarray}
The first equality is definition of gap. The second uses Eq.\eqref{eq:WALKgen}. The inequality comes from triangle inequality and $||A_{ii+1}||\leq 1$ (since $A$'a are moment operators). The last equality follows from $T_{\nu*\nu',t}=T_{\nu,t} T_{\nu',t}$,  applied to operators $A$. 
Hence we obtain the claimed result
\begin{equation}
   \gap\left(\nulocg*\nulocgd\right) \geq   \frac{1}{n-1}\gap(\nulocgg)\ .
\end{equation}

\end{proof}

In exactly analogous way one proves 
\begin{Lemma}\label{lem:nonHerm-Herm-par}[Bound on the gap of random parallel quantum circuits for non symmetric gate-set]
Let $\nuparg$ be a measure describing random parallel quantum circuits generated two qudit gate-set $\G$ (not necessarily symmetric).  We have the following lower bound
\begin{equation}
\label{eq:nonHerm-Herm-par}
    \gap\left(\nuparg*\nupargd\right) \geq \frac12\gap(\nupargg)\ .
\end{equation}
\end{Lemma}

Next we need Lemma proved in \cite{BHH2016} (it is not formulated as a separate Lemma, but it is a contents of  the proof of Corollary 7 of \cite{BHH2016}).
\begin{Lemma}[local circuits]
\label{lem:GversusHaar}
For a set $\tilde \G$ of gates containing inverses we have 
\begin{equation}
    \gap(\nulocgtil,t)\geq \gap(\nu_{\tilde \G},t) \gap(\nuloch,t),\quad \gap(\nupargtil)\geq \gap(\nu_{\tilde G},t) \gap(\nuparh,t)
\end{equation}
    I.e. we have the same relation for local as well as parallel circuits.
\end{Lemma}

We will also make use of theorem by Peter Varju, which we reformulated in Theorem \ref{thm:Varju}). It sattes that
\begin{equation}
\label{eq:Varjuthm2}
    \gap(\nu,t) \geq \frac{\gap(\nu,t_0)}{B \log^2(t)}
\end{equation}
for any measure $\nu$ on $\UU(d)$,
where $B$ depends only on $d$. 
Finally we shall need an estimate for the length of circuits that is needed to
produce an $\delta$-approximate \tdesign, 
provided that single step has gap $\gap(\nu,t)$: 
\begin{proposition}
\label{lem:length}
For any measure $\nu$ on $\UU(d)$, the measure $\nu^{*l}$
is $\delta$-approximate \tdesign, if 
\begin{equation}
    l\geq \frac{1}{\gap(\nu,t)}\log\frac1\delta\ .
\end{equation}
\end{proposition}
\begin{proof}
In order $\nu^{*l}$ to be \tdesign\ we need $(1-\gap(\nu,t))^l\leq \delta$. Taking logarithm of both sides, and 
using $1-\gap\leq e^{-\gap}$ proves the estimate. 
\end{proof}

We are now ready to prove the main result of this section
\begin{Theorem}
A random local circuit composed of two-qudit local gates from universal set of G given by $\nulocg$ is an $\delta$-approximate $t$-design provided its length $l_{loc}$ satisfies
\begin{equation}
    l_{loc}\geq n \log^2(t) C(\G) l_{loc,Haar}
\end{equation}
where $C(\G)>0$ is a constant depending only on set of gates and dimension $d$, while 
$l_{loc,Haar}$ is the length of the local Haar circuit that is $\delta$-approximate \tdesign.
For random parallel circuit, with analogous notation  we have 
\begin{equation}
    l_{par}\geq 2 \log^2(t) C(\G) l_{par,Haar}
\end{equation}
\end{Theorem}

\begin{proof}
Let us prove the result for local circuits first.
We have the following chain of inequalities:
\begin{eqnarray}
    &&\gap(\nulocg,t)\geq \frac12 \gap(\nulocg*\nulocgd,t)\geq \frac12\frac{1}{n-1} \gap(\nulocgg,t)\geq \nonumber\\
    &&\frac12\frac{1}{n-1} \gap(\nulocgg,t)
    \gap(\nuloch,t)\geq 
    \frac12\frac{1}{n}\frac{\gap(\nu_{\G\G^\dagger},t_0)}{B(d) \log^2 t} \gap(\nuloch,t) \geq  \frac12\frac{1}{n}\frac{\gap(\nu_{\G},t_0)}{B(d) \log^2 t} \gap(\nuloch,t)= \nonumber \\
    &&=
    \frac1n \frac{C(\G)}{\log^2 t} \gap(\nuloch,t)\nonumber\\
\end{eqnarray}
Here the first inequality comes from Lemma \ref{lem:mu-tilde-mu}, the second from Lemma \ref{lem:nonHerm-Herm-loc}, the third from Lemma \ref{lem:GversusHaar}, the fourth from Eq. \eqref{eq:Varjuthm2} and the fifth from Lemma \ref{lem:mu-tilde-mu}. Now, since 
$\G$ is universal  $\gap(\nu_\G,t_0)$ and hence $C(\G)$ is nonzero. 
Then applying Lemma \ref{lem:length} and inserting $l_{Haar}$ in place of 
\begin{equation}
    \frac{\ln\frac1\delta}{\gap(\nuloch)}
\end{equation}
ends the proof.
In the case of parallel circuits the proof is exactly the same 
with just one  difference: instead of $1/(n-1)$ there is factor $1/2$  after 
second inequality,  since instead of Lemma 
\ref{lem:nonHerm-Herm-loc} concerning local circuits we apply Lemma
\ref{lem:nonHerm-Herm-par} concerning parallel 
circuits. 
\end{proof}

\section{Polynomial approximation of Dirac delta on Unitarty channels}\label{sec:constr}

In this section we present the main idea behind the construction of the polynomial approximation of the Dirac delta in the manifold of unitary channels. The features of this particular approximation were stated without a proof in Theorem \ref{th:FUNC}. In our exposition we will follow a `bottom up' approach. The polynomial function on $\U(d)$ will be constructed from the Fourier series truncation (denoted by $f^\sigma_{p,k}$) of a suitable symmetric function $f^\sigma_p$ on $d$ dimensional torus $\T^d=\lbrace{ \bphi = \left(\varphi_1,\ldots,\varphi_d\right)  |\ \varphi_i\in[-\pi,\pi] \rbrace}$. This Fourier truncation will be then used to define a function  $F^\sigma_{p,k}$  which is a ''class extension'' of $f^\sigma_{p,k}$, i.e. is a function on $\UU(d)$ defined via $F^\sigma_{p,k}(U)=f^{\sigma}_{p,k}(\mathrm{Eig}(U))$, where $\mathrm{Eig}(U)$ denotes a diagonal matrix of eigenvalues of a unitary operator $U$. Finally, the function $F^\sigma_{p,k}$ will be averaged over the global phase, resulting in a well-defined polynomial function $\tilde{\F}^\sigma_{p,k}$ on $\U(d)$. This function will define (up to a normalisation constant) a polynomial approximation of Dirac $\delta$, denoted by  $\F^\sigma_k$, whose existence is  claimed in Theorem \ref{th:FUNC}.

We will adopt the following convention when referring to elements from the sets relevant in our considerations:  $\x,\y\in\R^d$, $\n,\k\in\Z^d$. Moreover,  we denote by $\x\cdot\y$ the standard inner product in $\R^d$ (note that we can apply it to    to elements of $\Z^d \subset \R^d$). Finally, we will denote by $|\x|$ and $|\x|_1$ respectively  euclidean and 1-norm of $\x\in\R^d$. 

We begin by introducing a number of useful functions on $\R^d$ and $\T^d$.  The standard Gaussian distribution on $\R^d$ is defined by
\begin{equation}\label{eq:gaussianDistribution}
    f^\sigma :\R^d \rightarrow \R^d\ ,\ \ f^\sigma(\x) \coloneqq \frac{1}{(\sqrt{2\pi} \sigma)^d} \exp\left(-\frac{\x^2}{2\sigma^2} \right) \ . 
\end{equation}
It will be convenient for us to introduce a periodized version of this function 
\begin{equation}
    f^\sigma_p: \T^d \rightarrow \T^d\ ,\ \  f^\sigma_p(\bphi) \coloneqq \sum_{\k\in\Z^d} f^\sigma(\bphi +2\pi\k) \ .
\end{equation}
By the virtue of the Poisson summation formula  \cite{Poisson} we know that for any function $f\in L^1(\mathbb{R}^d)$ that satisfies:
\begin{equation}
    |f(\x)|\leq\frac{C}{(1+|\x|)^{d+\alpha}}\ ,
\end{equation}
for some positive constants $C$ and $\alpha$, we have
\begin{equation}
    f_p (\bphi) = \frac{1}{(2\pi)^d} \sum_{\n\in\Z^d} \hat{f}(\n) \exp(\ii \n\cdot \bphi)\ , \ \text{where}\ \hat{f}(\n)=\int_{\R^d}d\y f(\y) \exp(-\ii \n\cdot \y ) \ ,
 \end{equation}
is the standard Fourier transform of $f$ computed at point $\n\in\Z^d$. Using the fact that $\hat{f}^\sigma(\n)= e^{-\frac{1}{2}\sigma^2 \n^2}$  we obtain
\begin{equation}\label{eq:PeriodGAUSS}
 f^\sigma_p (\bphi)= \frac{1}{(2\pi)^d} \sum_{\n\in\Z^d} e^{-\frac{1}{2}\sigma^2 \n^2} \exp(\ii \n\cdot \bphi) \ .
 \end{equation}
Analogously we define a truncated version of $f^\sigma_p (\bphi)$,
\begin{equation}\label{eq:TruncatedGAUSS}
    f^{\sigma}_{p,k}(\bphi) =\frac{1}{(2\pi)^d} \sum_{\n\in S_k} e^{-\frac{1}{2}\sigma^2 \n^2} \exp(\ii \n\cdot \bphi)  \ , 
\end{equation}
where $S_k=\{\n\ |\ |\n|_1 \leq k \}$.

 Finally we define a `phase averaged' versions of functions  $f^{\sigma}_{p}$ and   $f^{\sigma}_{p,k}$
\begin{equation}
     f^{\sigma,a}_{p} (\bphi) \coloneqq \frac{1}{2\pi} \int_{0}^{2\pi} d\phi  f^{\sigma}_{p}(\bphi+(\phi,\ldots,\phi))\  ,\   f^{\sigma,a}_{p,k} (\bphi) \coloneqq \frac{1}{2\pi} \int_{0}^{2\pi} d\phi  f^{\sigma}_{p,k}(\bphi+(\phi,\ldots,\phi))\ .
\end{equation}We use the fact that functions $f^\sigma_p$ and $f^\sigma_{p,k}$ are functions on $\T^d$ that are invariant under the permutation of angles. Therefore, we can define class functions $F^\sigma$ and $F^\sigma_k$ on $\UU(d)$ that recover $f^\sigma_p$ and $f^\sigma_{p,k}$ when restricted to $\T^d$.  In other words
\begin{equation}\label{eq:classFUNCTIONdef}
F^\sigma (U) \coloneqq  f^\sigma_p(\mathrm{Eig}(U))\ , \ F^\sigma_k (U) \coloneqq  f^\sigma_{p,k}(\mathrm{Eig}(U))\ ,
\end{equation}
where $\mathrm{Eig}(U)=\mathrm{diag}(\exp(\ii \phi_1) ,\ldots, \exp(\ii \phi_d))$ is a diagonal matrix formed by eigenvalues of $U$. When we average the above functions over the global phase we get well-defined functions on the group of unitary channels $\U(d)$.
\begin{equation}\label{eq:unnormalizedPROJECTIVE}
\tilde{\F}^\sigma =\P_{\mathrm{phase}} F^\sigma\ ,\ \tilde{\F}^\sigma_k =\P_{\mathrm{phase}} F^\sigma_k \ ,
\end{equation}
where linear operator $\P_{\mathrm{phase}}:\L^2 (\UU(d))\rightarrow \L^2 (\UU(d))$ is defined by
\begin{equation}\label{eq:phaseProjector}
 (\P_{\mathrm{phase}}F)(U)=  \frac{1}{2\pi} \int_{0}^{2\pi} d\phi F(\exp(\ii \phi) U )\ .
 \end{equation}
We note that $\P_{\mathrm{phase}}$ is an orthonormal projector in $\L^2 (\UU(d))$ that projects onto functions in $\L^2(\UU(d))$ that are invariant under a global phase transformation. As explained in Section \ref{sec:preliminary} we can interpret such functions as functions defined on $\U(d)$. The normalised version of $\tilde{\F}^\sigma_k$, 
\begin{equation}\label{eq:MAINfuncDEF}
\F^\sigma_k \coloneqq  \tilde{\F}^\sigma_k/\N^{\sigma}_{k} \ ,\ \N^{\sigma}_{k}\coloneqq \int_{\U(d)} \dt \mu(\h{U}) \tilde{\F}^\sigma_k (\h{U})
\end{equation}
is our candidate for a `low degree' approximation of the Dirac $\delta$ at $\h{I}$. We shall prove Theorem \ref{th:FUNC} via a sequence of technical Lemmas that will eventually cover all the properties stated in Theorem \ref{th:FUNC}. It will be also convenient to introduce auxiliary  function of $\U(d)$, 
\begin{equation}\label{eq:AUXfuncDEF}
\F^\sigma \coloneqq  \tilde{\F}^\sigma /\N^{\sigma} \  ,\  \N^{\sigma}\coloneqq \int_{\U(d)} \dt \mu(\h{U}) \tilde{\F}^\sigma (\h{U}) 
\end{equation}
that will serve as a reference function that $\F^\sigma_k$ approximates as $k\rightarrow\infty$. We begin with the following Lemma.

\begin{Lemma}\label{lem:PolynDegree} The function $\F^\sigma_k$ defined in the preceding paragraphs satisfies $\F^\sigma_k\in\H_k$ i.e. is a balanced polynomial of degree $k$ in $U$ and $\bar{U}$. 
\end{Lemma}
The proof of this result, which seems intuitive at the first sight, turns out to surprisingly complex. We present it in Part \ref{app:ClassFUNC} of the Appendix.
n
We proceed with giving a number of properties of function $\tilde{\F}^\sigma$. The relevant properties of $\tilde{\F}^\sigma_k$ will be derived latter by controlling the error resulting form the truncation. In order to facilitate the computations involved, our proof strategy effectively shifts the considerations from $\U(d)$ to $\UU(d)$. In particular, for class functions defined on the unitary group $\UU(d)$ we often make use of the Weyl integration formula \cite{WalachBook}, which ensures that for any class function $F$ on $\UU(d)$ we have
\begin{equation}\label{eq:WaylIntegration}
\int_{\UU(d)}\dt\mu(U)F(U)=\int_{\T^d} \dt \mu(\bphi) F(\mathrm{diag}(\exp(\ii \varphi_1,\ldots,\exp((\ii \varphi_d) )\ ,
\end{equation}
where the  measure on $\T^d$ is the \emph{push-forward} of a Haar measure on $\UU(d)$ and is given by
\begin{equation}\label{eq:torusMeasure}
     \dt \mu(\bphi)=
    \frac{1}{(2\pi)^d d!}
    \prod_{1\leq i<j\leq d}|e^{i \varphi_i}-e^{i \varphi_j}|^2     \dt \varphi_1 \ldots \dt \varphi_d
\end{equation}
Although the corresponding formula is guaranteed to exist in principle also for class functions on $\U(d)$, we are not aware of any explicit expressions analogous to Eq.\eqref{eq:WaylIntegration}.

\begin{Lemma}[Lower bound on the normalization constant $\N^\sigma$]\label{lem:loverBOUNDnormCONST} Let $\N_{\sigma}$ be defined as in Eq.\eqref{eq:AUXfuncDEF} and let $\sigma\leq\frac{\pi}{4\sqrt{d}}$.  We have the following inequality 
\begin{equation}\label{eq:LOWERboundCPNST}
    \N^{\sigma}\geq  \frac{1}{2} C_d\ \sigma^{d(d-1)} \left(\frac{2}{\pi}\right)^{d(d-1)}\ ,
\end{equation}
where $C_d=\frac{\prod_{k=1}^d k! }{(2\pi)^d d!}$.
\end{Lemma}

\begin{proof}[Sketch of the proof]
Observe first that due to the definition of the Haar measure on $\U(d)$ (see Section \ref{sec:preliminary})  the normalisation constant $\N^\sigma$ can be expressed via the integral from function $F^\sigma$ (defined in Eq.\eqref{eq:classFUNCTIONdef})  
\begin{equation}\label{eq:NormPhase}
    \N^\sigma   =   \int_{\U(d)} \dt \mu(\h{U}) \tilde{\F}^\sigma(\h{U}) = \int_{\UU(d)} \dt \mu(U) \left(\frac{1}{2\pi}\int_0^{2\pi} \dt \varphi  F^\sigma(\exp(\ii \varphi) U)\right)=\int_{\UU(d)} \dt \mu(U) F^\sigma(U)\ ,
\end{equation}
where in the last equality we used invariance of the Haar measure on $\UU(d)$ under the translations by unitary operations (in this case $\exp{(\ii\varphi)I}$). 
Importantly, by the virtue of Weyl integration formula (cf. Eq.\eqref{eq:WaylIntegration}) the integral appearing in the right-hand side of Eq.\eqref{eq:NormPhase} can be expressed via the integral of the periodized Gaussian $f^\sigma_p$ defined on $\T^d$. This allows us to write  
\begin{equation}\label{eq:nonnegative}
\int_{\UU(d)} \dt \mu(U) F^\sigma(U) =  \int_{\T^d}\dt \mu(\bphi) \gauss^\sigma_p(\bphi) \geq \int_{\T^d}\dt \mu(\bphi) \gauss^\sigma(\bphi)\ , 
\end{equation}
where the inequality follows from $\gauss^\sigma_p(\bphi) \geq \gauss^\sigma (\bphi)$. The function $\gauss^\sigma$ turns out to be closely related to the GUE ensemble of random Hermitian  matrices \cite{SzarekBook} which ultimately allows us to to establish the following bound 
\begin{equation}\label{eq:UnitGroupNormBound}
\int_{\T^d}\dt \mu(\bphi) \gauss^\sigma(\bphi)  \geq  \frac{1}{2} C_d\,\sigma^{d(d-1)}\,\left(\frac{2}{\pi}\right)^{d(d-1)}\ ,
\end{equation}
where the dimension-dependant constant $C_d=\frac{\prod_{k=1}^d k! }{(2\pi)^d d!}$ appears because of the usage of the \emph{Mehta integral} \cite{Mehta-integral}. Combining the  above inequality with  Eq.\eqref{eq:NormPhase} and Eq.\eqref{eq:nonnegative}  concludes proofs of Lemma \ref{lem:loverBOUNDnormCONST}. The detailed reasoning justifying Eq. \eqref{eq:UnitGroupNormBound} is given in Lemma \ref{lem:denominator} in Appendix \ref{app:torusEST}. 
\end{proof}

The following result allows us to upper bound the rate of decay of integrals of the form $\int_{B(\h{V},\kappa)} \dt \mu (\h{U}) \tilde{\F}^\sigma (\h{U})$, where $\dproj\left(\h{V},\h{I}\right)\geq\ep$ in terms of the integrals on the unitary group $\UU(d)$. The latter turn out to be simpler to analyze.

\begin{Lemma}\label{lem:diamondVSop} Let $F^\sigma$   and $\tilde{\F}^\sigma$ be functions on $\UU(d)$ and  $\U(d)$ defined in Eq.\eqref{eq:classFUNCTIONdef}  and Eq.\eqref{eq:unnormalizedPROJECTIVE} respectively. Let $r \geq 0$. Then
we have the following inequality
% \begin{equation}\label{eq:unitaryUPPERbound}
% \int_{B(\h{V},\kappa)} \dt \mu (\h{U}) \tilde{\F}^\sigma (\h{U}) \leq \int_{B({I},\ep-\kappa)^c} \dt \mu (U) F^\sigma (U)\ ,
% \end{equation}
\begin{equation}
\label{eq:unitaryUPPERbound}
\int_{ B(\h{I},r)^c} \dt \mu (\h{U}) \tilde{\F}^\sigma (\h{U}) \leq  \int_{ B({I},r)^c} \dt \mu (U) F^\sigma (U)\ ,
\end{equation}
where  $B({I},r)^c=\left\lbrace U\in\UU(d)\ |\ \|U-I\|> r \right\rbrace$ is  is the complement of the ball with respect to the operator norm in $\UU(d)$  and $B(\h{I},r)=\left\lbrace \h{U}\in\U(d)\ |\  \dproj\left(\h{U},\h{I}\right)> r \right\rbrace$.

\end{Lemma}
\begin{proof}
% We begin by noting that for $\dproj(\h{V},\h{I})\geq\ep$ and $\kappa\leq\ep$ we have $B(\h{V},\kappa)\subset B(\h{I},\ep-\kappa)^c$ and consequently 
% \begin{equation}\label{eq:firstINEQ}
%     \int_{B(\h{V},\kappa)} \dt \mu (\h{U}) \tilde{\F}^\sigma (\h{U}) \leq \int_{B(\h{I},\ep-\kappa)^c} \dt \mu (\h{U}) \tilde{\F}^\sigma (\h{U})\ .
% \end{equation}
From  the characterization of $\dproj$ given Eq.\eqref{eq:opDIST} we get that for all $r>0$
\begin{equation}
\label{eq:ballIdentity}
    \{ U\in\UU(d)\ |\ \dproj(\h{U},\I)\leq r \} = \bigcup_\phi \{U\in\UU(d)\ |\  \|e^{i\phi}\I- U \| \leq r\}\ .
\end{equation}
The connection between the Haar measures on $\UU(d)$ and $\U(d)$ and the definition of $\tilde{\F}^\sigma$ gives
\begin{equation}
    \int_{B(\h{I},r)} \dt \mu (\h{U}) \tilde{\F}^\sigma (\h{U}) = \int_{\bigcup_\phi B({\exp(\ii\phi)I},r)} \dt \mu(U) \mathbb{P}_{\mathrm{phase}} F^\sigma (U) = \int_{\bigcup_\phi B({\exp(\ii\phi)I},r)} \dt \mu(U)  F^\sigma (U)\ ,
\end{equation}
where in the last equality we used the invariance of the set $\bigcup_\phi B({\exp(\ii\phi)I},r)$ with respect to the multiplication by the global phase. Next, since  $F^\sigma (U) \geq0$ we get
\begin{equation}
     \int_{B(\h{I},r)} \dt \mu (\h{U}) \tilde{\F}^\sigma (\h{U}) \geq \int_{B(I,r)} \dt \mu(U)  F^\sigma (U)\ .
\end{equation}
Combining this with

\begin{equation}
     \int_{\U(d)} \dt \mu (\h{U}) \tilde{\F}^\sigma (\h{U}) = \int_{\UU(d)} \dt \mu(U)  F^\sigma (U)
\end{equation}
(shown in  Eq. \eqref{eq:NormPhase})
concludes the proof.

\end{proof}

We now want to control the rate of decay of the integral appearing in the right-hand side of Eq.\eqref{eq:unitaryUPPERbound}. To this end we use Weyl integration formula
(cf. Eq.\eqref{eq:WaylIntegration}) which gives 
\begin{equation}\label{eq:tailTorus}
    \int_{B(I,r)^c} \dt \mu(U)  F^\sigma (U) = \int_{B_\infty(0,r)^c} \dt \mu(\bphi) f^\sigma_{p} (\bphi)\ , 
\end{equation}
where $B_\infty(0,r)=\left\{\bphi\in\T^d\ |\ |\varphi_i|\leq r \right\}$ and $B_\infty(0,r)^c$ is its complement in $\T^d$. Next, in Lemma \ref{lem:numerator} given in part \ref{app:torusEST} of the Appendix  we establish upper bounds on the right-hand side of \eqref{eq:tailTorus}. This result is proven by (i) establishing appropriate upper bounds on the norm $\|f^\sigma -f^\sigma_p\|_1$ where $\|\cdot\|_1$ denotes $L^1$ norm of the space of integrable functions on  $\T^d$ equipped with the measure $\dt\mu(\bphi)$, and (ii) connecting $\int_{B_\infty(0,r)^c} \dt \mu(\bphi) f^\sigma (\bphi)$ to the tail behaviour of the operator norm of GUE matrices \cite{SzarekBook}.  In this way we obtain the following lemma:

\begin{Lemma}\label{lem:TailUnitary} Let $F^\sigma$  be the function on $\UU(d)$ defined in Eq.\eqref{eq:classFUNCTIONdef}. Let $\sigma\leq \frac{r}{4\sqrt{d}}$ and $r\leq 2/3$. We have the following inequality
\begin{equation}
    \int_{B({I},r)^c} \dt \mu (U) F^\sigma (U) \leq \frac{3}{2} C_d
      \,\sigma^{d(d-1)}\, 
      e^{-\frac{1}{4} \frac{r^2}{\sigma^2}}
\end{equation}
where  $B({I},r)^c=\left\lbrace U\in\UU(d)\ |\ \|U-I\|> r \right\rbrace$ and $C_d= \frac{\prod_{k=1}^d k! }{(2\pi)^d d!}$.
\end{Lemma}
We conclude our characterisation of the functions $\tilde{\F}^\sigma$ by the following lemma:
\begin{Lemma}
\label{lem:uppBOUNDsecondNORMunitary1}
Let $\tilde{\F}^{\sigma}$ be a function on $\U(d)$  defined in Eq.~\eqref{eq:unnormalizedPROJECTIVE}.
We then have for $\sigma\leq 1/4$
\begin{equation}
   \left\|\tilde{\F}^{\sigma}  \right\|_2\leq  \frac{1}{d! (2 \pi)^d}
   2^{\frac{d(d+1)}{2}}\sqrt{d!} \, \sigma^{-\frac{d}{2}}.
\end{equation}
\end{Lemma}
\begin{proof}
We reduce the problem to consideration of functions on $\T^d$. First, due to the fact that $\mathbb{P}_{\mathrm{phase}}$ is an orthonormal projector in $L^2(\UU(d))$ we get 
\begin{equation}
    \left\|\tilde{\F}^\sigma  \right\|_2 = \left\|\mathbb{P}_{\mathrm{phase}} F^\sigma \right\|_2 \leq \left\| F^\sigma \right\|_2 \ . 
\end{equation}
Using the Weyl integration formula  and definitions of class function $F^\sigma$ (cf. Eq.\eqref{eq:classFUNCTIONdef}) we obtain  $\left\| F^\sigma \right\|_2 = \|f^\sigma_p\|_2$, where the $L^2$ norm of $f^\sigma_p$  is computed using the measure $\dt\mu(\bphi)$ on $\T^d$ (this is a consequence of Eq.\eqref{eq:WaylIntegration}). The claimed result follows now from the inequality (valid for $\sigma\leq1/4$).
\begin{equation}
 \|f^\sigma_p\|_2 \leq 
 \frac{1}{d! (2 \pi)^d}
   2^{\frac{d(d+1)}{2}}\sqrt{d!} \, \sigma^{-\frac{d}{2}},
\end{equation}
which we prove in Lemma \ref{lem:2norm}.
\end{proof}
\blk

%\cgrey
%\begin{proposition}
%\label{prop:uppBOUNDsecondNORMunitary} Let $\tilde{\F}^{\sigma}_{ k}$ be a function on $\U(d)$  defined in Eq.~\eqref{eq:unnormalizedPROJECTIVE} and let $\N^\sigma$ be a constant defined in Eq.~\eqref{eq:AUXfuncDEF}. 
%Then for $k\geq d/\sigma$  and $\sigma\leq 1/4$ we have the following inequality
%\begin{equation}
%   \left\|\tilde{\F}^{\sigma}_k  \right\|_2\leq 2 \frac{1}{d! (2 \pi)^d}
%   2^{\frac{d(d+1)}{2}}\sqrt{d!} \, \sigma^{-\frac12}.
%\end{equation}
%\cgrey
% Consequently,  $L^2$ norm of $\F^\sigma=\tilde{\F}^\sigma/\N^\sigma$ satisfies
% \begin{equation}
%     \left\|\F^\sigma\right\|_2 \leq 1\ .
% \end{equation}
%\end{proposition}

%Combining the results stated in Lemmas \ref{lem:loverBOUNDnormCONST}, \ref{lem:diamondVSop}, \ref{lem:TailUnitary} with Proposition \ref{prop:uppBOUNDsecondNORMunitary} establishes that the function $\F^\sigma$ is normalized, satisfies $\|\F^\sigma\|_2\leq 1$, and for  \sout{$\sigma\leq\frac{\ep-\kappa}{4\sqrt{d}}$}
%$\sigma\leq\frac{\ep}{4\sqrt{d}}$ \cgrey
%we have
%\begin{equation}
%\int_{B(\h{I},\ep)^c \cgrey} \dt \mu (\h{U}) \F^\sigma (\h{U}) \leq 3 \exp\left(-\frac{ \ep^2\blk}{4\sigma^2}\right) \left(\frac{\pi}{2}\right)^{d(d-1)} \ .
%\end{equation}
%\sout{ whenever $\dproj(\h{V},\h{I})\geq\ep$.} Looking at the above conditions we see that they closely resemble properties required from $\F^\sigma_k$ in Theorem \ref{th:FUNC}.  
%\blk

The following key Lemma controls the rate of approximation of $\tilde{\F}^\sigma$ by $\tilde{\F}^\sigma_k$ in $L^2$ norm.

\begin{Lemma}[Approximation of $\tilde{\F}^\sigma$ by $\tilde{\F}^\sigma_k$]\label{lem:L2normBOUND}
 Let  $k\geq d/\sigma %d+\sqrt{d}/\sigma
 $ 
 and let $\sigma\leq1/2$. Let $\tilde{\F}^\sigma_k$ and $\tilde{\F}^\sigma$ be functions defined in Eq.\eqref{eq:unnormalizedPROJECTIVE}. We have the following upper bound on the $L^2$ distance between these functions

 \begin{align}
 \label{eq:2NormBoundPoly1}
    \left\|\tilde{\F}^\sigma - \tilde{\F}^\sigma_k \right\|_2\leq  \frac{1}{(2\pi)^d d!}
    \sqrt{\frac{2^{d(d-1)} \pi^\frac{d}{2} \sqrt{8 \pi} d!}{\Gamma(\frac{d}{2})}}\,\frac{e^{-\frac14 (\frac{k}{\sqrt{d}}-\sqrt{d})^2\sigma^2}}{\sigma}\ .  
 \end{align}
  Further one can estimate this as 
\begin{equation}
\label{eq:2NormBoundPoly}
    \left\|\tilde{\F}^\sigma - \tilde{\F}^\sigma_k \right\|_2 \leq 10\,C_d \frac{e^{-\frac14 (\frac{k}{\sqrt{d}}-\sqrt{d})^2\sigma^2}}{\sigma}\ .
\end{equation}
\end{Lemma}
\begin{proof}[Proof sketch] Similarly as before we reduce the problem to consideration of functions on $\T^d$. 
First, due to the fact that $\mathbb{P}_{\mathrm{phase}}$ is an orthonormal projector in $L^2(\UU(d))$ we get 
\begin{equation}\label{eq:2NormBoundPoly-x}
    \left\|\tilde{\F}^\sigma - \tilde{\F}^\sigma_k \right\|_2 = \left\|\mathbb{P}_{\mathrm{phase}}\left( F^\sigma - F^\sigma_k\right) \right\|_2 \leq \left\| F^\sigma - F^\sigma_k \right\|_2 \ . 
\end{equation}
Using the Weyl integration formula  and definitions of class functions $F^\sigma,F^\sigma_k$ (cf. Eq.\eqref{eq:classFUNCTIONdef}) we obtain  $\left\| F^\sigma - F^\sigma_k \right\|_2 = \|f^\sigma_p-f^\sigma_{p,k}\|_2$, where the $L^2$ distance between $f^\sigma_p$ and $f^\sigma_{p,k}$ is computed using the measure $\dt\mu(\bphi)$ on $\T^d$ (this is a consequence of Eq.\eqref{eq:WaylIntegration}). The claimed result follows now from the  inequalities  (valid for $k\geq d/\sigma$ and $\sigma\leq1/2$)
\begin{equation}
 \|f^\sigma_p-f^\sigma_{p,k}\|_2 \leq 
\frac{1}{(2\pi)^d d!}
    \sqrt{\frac{2^{d(d-1)} \pi^\frac{d}{2} \sqrt{8 \pi} d!}{\Gamma(\frac{d}{2})}}\,\frac{e^{-\frac14 (\frac{k}{\sqrt{d}}-\sqrt{d})^2\sigma^2}}{\sigma} \leq 
 10\,C_d \frac{e^{-\frac14 (\frac{k}{\sqrt{d}}-\sqrt{d})^2\sigma^2}}{\sigma}\, , 
\end{equation}
which we prove in Lemma \ref{lem:normfkf} in Part \ref{app:polyCONVERG} of the Appendix using trigonometric expansions \eqref{eq:PeriodGAUSS} and \eqref{eq:TruncatedGAUSS}. 
\end{proof}

We are now in position to evaluate the $L^2$ of $\tilde\F_k^\sigma$.
\begin{proposition}
\label{prop:2normFksigma}
For $k\geq d/\sigma$ and $\sigma\leq 1/4$ we have 
\begin{equation}
    \left\|\tilde\F_k^\sigma \right\|_2\leq 
    2 \frac{1}{d! (2 \pi)^d}
   2^{\frac{d(d+1)}{2}}\sqrt{d!} \, \sigma^{-\frac{d}{2}}.
\end{equation}
\end{proposition}
\begin{proof}
By triangle inequality 
\begin{align}
    \left\|\tilde\F_k^\sigma\right\|_2
    \leq 
    \left\|\tilde\F^\sigma\right\|_2 + 
    \left\|\tilde\F_k^\sigma- \tilde\F^\sigma\right\|_2 
\end{align}
and application of 
lemma \ref{lem:uppBOUNDsecondNORMunitary1}
and \eqref{eq:2NormBoundPoly1} of Lemma \ref{lem:L2normBOUND},   we get 
\begin{align}
    \left\|\tilde\F_k^\sigma\right\|_2\leq 2 \times 
    \frac{1}{d! (2\pi)^d}
    \max \left\{
    \sqrt{\frac{2^{d(d-1)} \pi^\frac{d}{2} \sqrt{8 \pi}}{\Gamma(\frac{d}{2})}}, 2^\frac{d(d+1)}{2}
    \right\} \sqrt{d!} \,    \sigma^{-\frac{d}{2}}
\end{align}
where we have taken into account that 
the last term in 
\eqref{eq:2NormBoundPoly1} is bounded by $1$, and for $\sigma\leq1$ 
we have $\sigma^{-d/2} \geq \sigma^{-1}$. 
We now have 
\begin{align}
    \max \left\{
    \sqrt{\frac{2^{d(d-1)} \pi^\frac{d}{2} \sqrt{8 \pi}}{\Gamma(\frac{d}{2})}}, 2^\frac{d(d+1)}{2}
    \right\}
    = 2^\frac{d(d+1)}{2} 
    \max\left\{\frac{\pi^{\frac{d}{4}}(8 \pi)^\frac14}{2^d \sqrt{\Gamma(\frac{d}{2})}},1\right\}.
\end{align}
We then have 
\begin{align}
    \frac{\pi^{\frac{d}{4}}(8 \pi)^\frac14}{2^d \sqrt{\Gamma(\frac{d}{2})}}
    \leq 
    \frac{\pi^{\frac{d}{4}}(8 \pi)^\frac14}{2^d \Gamma(\frac{d}{2})}
    \leq 1 
\end{align}
where the last inequality 
is obtained by induction (it is convenient to run it separately for $d$ even and $d$ odd). This ends the proof. 
\end{proof}

The following proposition asserts that for sufficiently large degree $k$
the $L^2$-distance $\|\tilde{\F}^\sigma -\tilde{\F}^\sigma_{k}\|_2$ is comparable with $\N^\sigma$ (cf. Lemma \ref{lem:loverBOUNDnormCONST}) and the upper bound on $ \int_{B({I},r)^c} \dt \mu (U) F^\sigma (U)$ from Lemma \ref{lem:TailUnitary}. 

\begin{proposition}\label{prop:settingK}  Let $\tilde{\F}^\sigma_k$ and $\tilde{\F}^\sigma$ be functions defined in Eq.\eqref{eq:unnormalizedPROJECTIVE}. Moreover let 
 \begin{equation}\label{eq:Kcondition}
     \sigma\leq \min\lbrace1/8, \frac{\pi}{4\sqrt{d}} \rbrace \ ,\  k\geq 5 \frac{d^{\frac{3}{2}}}{\sigma}  \sqrt{\frac{1}{8}\frac{r^2}{d^2\sigma^2} + \ln\frac{1}{\sigma}}\ ,\  r\leq \frac23 \ .
 \end{equation}
Then we have 
\begin{equation}\label{eq:2normComparison}
\|\tilde{\F}^\sigma -\tilde{\F}^\sigma_{k}\|_2 \leq \frac{1}{2}\min\left\{\N^\sigma, \frac{3}{2} C_d\,\sigma^{d(d-1)}\, 
      e^{-\frac{1}{4} \frac{r^2}{\sigma^2}}\right\}\ , 
\end{equation}
where $\N^\sigma$ is a constant defined in Eq.\eqref{eq:AUXfuncDEF} and $\frac{3}{2} C_d\,\sigma^{d(d-1)}\, e^{-\frac{1}{4} \frac{r^2}{\sigma^2}}$ is an upper bound on the integral $ \int_{B(I,r)^c} \dt \mu (U) F^\sigma (U)$ from Lemma \ref{lem:TailUnitary}.
\end{proposition}
The proof of the above result follows from comparison of upper bound \eqref{eq:2NormBoundPoly}  
with the bounds given in Lemmas \ref{lem:loverBOUNDnormCONST} and \ref{lem:TailUnitary}. 
The comparison is provided by 
the (technical) Lemma \ref{lem:settingK-technical}
 proved  in Part \ref{app:polyCONVERG} of the Appendix.
We have now all the necessary ingredients to justify that the function $\F^\sigma_k$ satisfies all the properties required by Theorem \ref{th:FUNC}.

\begin{proof}[Proof of Theorem \ref{th:FUNC}]
First of all Lemma \ref{lem:PolynDegree} ensures that $\F^\sigma_k$ is a polynomial of a suitable degree in $U$ and $\bar{U}$ as claimed in property 3 in Theorem \ref{th:FUNC}.

We now proceed with proofs of the remaining three properties $\F^\sigma_k$.
According to assumptions of the Theorem we will assume  $\sigma\leq\frac{\ep}{4\sqrt{d}}$ and  $k\geq 5 \frac{d^{\frac{3}{2}}}{\sigma}  \sqrt{\frac{1}{8}\frac{\ep^2}{d^2\sigma^2} + \ln\frac{1}{\sigma}}$.  
    
We start by establishing the normalization of $\F^\sigma_k$ (condition 1 in Theorem \ref{th:FUNC}). By definition of $\F^\sigma_k$ (cf. Eq.\eqref{eq:MAINfuncDEF}) this is equivalent to showing  
    \begin{equation}
        \N^\sigma_k=\int_{\U(d)}\dt \mu(\h{U}) \tilde{\F}^\sigma_k(\U)\neq0\ .
    \end{equation}
By simple manipulations we get 
\begin{equation}
\int_{\U(d)}\dt \mu(\h{U}) \tilde{\F}^\sigma_k(\U)= \int_{\U(d)}\dt \mu(\h{U}) \tilde{\F}^\sigma(\U) -    \int_{\U(d)}\dt \mu(\h{U}) (\tilde{\F}^\sigma(\U)-\tilde{\F}^\sigma_k(\U))\geq  \int_{\U(d)}\dt \mu(\h{U}) \tilde{\F}^\sigma(\U) - \|\tilde{\F}^\sigma -\tilde{\F}^\sigma_{k}\|_2\ ,
\end{equation}
where the inequality comes form applying the Cauchy-Schwartz inequality in  $L^2(\U(d))$ to functions $G_1(\U)=\tilde{\F}^\sigma(\U)-\tilde{\F}^\sigma_k(\U)$ and $G_2(\U)=1$. Using \eqref{eq:2normComparison} in the above inequality we obtain 
\begin{equation}\label{eq:NORMkbound}
    \N^\sigma_k \geq \frac12 \N^\sigma>0\ .
\end{equation}

The proof of the second property in Theorem \ref{th:FUNC} (decay of  integral   of modulus of the function 
outside of the balls $B(\h{I},\ep)$) follows the similar logic.  Specifically, using the Cauchy-Schwartz inequality leads, as previously, to
\begin{equation}\label{eq:crucialSTEP}
      \int_{ B(\h{I},\ep)^c} \dt \mu (\h{U})\,  | \tilde{\F}^\sigma_k |  \leq \int_{ B(\h{I},\ep)^c} \dt \mu (\h{U}) \tilde{\F}^\sigma +  \|\tilde{\F}^\sigma -\tilde{\F}^\sigma_{k}\|_2 \ .
\end{equation}
 where we used $|\tilde{\F}^\sigma|=\tilde{\F}^\sigma$ since  
$\tilde{\F}^\sigma$ is positive. Application of the bound \eqref{eq:2normComparison} from Proposition \ref{prop:settingK} and results of Lemmas \ref{lem:diamondVSop} and \ref{lem:TailUnitary} gives
\begin{equation}
      \int_{B(\h{I},\ep)^c} \dt \mu (\h{U}) |\tilde{\F}^\sigma_k |\leq \frac{3}{2} \times \frac{3}{2} C_d\,\sigma^{d(d-1)}\, 
      e^{-\frac{1}{4} \frac{(\ep-\kappa)^2}{\sigma^2}}\ .
\end{equation}
Using the definition of  $\F^\sigma_k$ and employing \eqref{eq:NORMkbound} together with the lower bound for $\N^\sigma$  from Lemma \ref{lem:loverBOUNDnormCONST} we finally obtain the desired result
\begin{equation}
    \int_{ B(\h{I},\ep)^c} \dt \mu (\h{U})  |\F^\sigma_k |= \frac{1}{\N^\sigma_k} \int_{ B(\h{I},\ep)^c} \dt \mu (\h{U}) |\tilde{\F}^\sigma_k| \leq 3\times 3 \exp\left(-\frac{(\epsilon-\kappa)^2}{4\sigma^2}\right) \left(\frac{\pi}{2}\right)^{d(d-1)}\ .
\end{equation}
% We conclude the proof by giving an upper bound on $\|\F^\sigma_{ k}\|_2$. 

We shall now give an upper bound on $\|\F^\sigma_{ k}\|_2$. We have from Eq. \eqref{eq:NORMkbound}
\begin{equation}
    \|\F^\sigma_k\|_2=\frac{1}{\N^\sigma_k}\|\tilde{\F}^\sigma_k\|_2 \leq 2\frac{\|\tilde{\F}^\sigma_k\|_2}{\N^\sigma}. 
\end{equation}
We now use the estimates  
from lemma \ref{lem:loverBOUNDnormCONST}
 and proposition \ref{prop:2normFksigma}
 obtaining for $k\geq d/\sigma$, $\sigma\leq\frac{\pi}{4\sqrt{d}}$ and $\sigma\leq 1/4$ 
 \begin{align}
     \|\F^\sigma_k\|_2\leq 
     2 \times \frac{2 \times
   2^{\frac{d(d+1)}{2}}\sqrt{d!} \, \sigma^{-\frac{d}{2}}}{\frac{1}{2} \prod_{k=1}^d k!\ \sigma^{d(d-1)} \left(\frac{2}{\pi}\right)^{d(d-1)}}\equiv
   8 \times 2^{d^2} B_d\, \sigma^{-d(d-\frac12)},
 \end{align}
 where the terms $1/(d!(2\pi)^d)$
 are already cancelled,
 and 
 \begin{align}
     B_d=
     \frac{\sqrt{d!} 2^\frac{d(d+1)}{2}}{2^{d^2}
     \prod_{k=1}^d k! 
     \left(\frac{2}{\pi}\right)^{d(d-1)}}
 \end{align}
 We prove by induction (one has to actually run induction twice) that $B_d\leq 1$. 
 Finally we note that the condition (so that we can take $\sigma\leq \frac{1}{6\sqrt{d}}$) that we have in assumptions of the theorem implies the above conditions on $\sigma$.

 We conclude the proof by giving an upper bound on $\|\F^\sigma_{ k}\|_1$. 
 Due to Lemma \ref{lem:pos_function_norm_1} it is enough to bound distance in 1-norm between $\F_k^\sigma$ and a positive function our our choice. 
 We choose the positive function to be $\tilde\F_\sigma/\N^\sigma_k$. 
 We then have 
 \begin{align} 
 \label{eq:est_proof_1norm}
 \| \F_k^\sigma - \frac{\tilde\F_\sigma}{\N^\sigma_k}\|_1
 =
 \frac{1}{\N^\sigma_k} \| \F_k^\sigma - \tilde\F_\sigma\|_1 \leq 
 \frac{1}{\N^\sigma_k} \| \F_k^\sigma - \tilde\F_\sigma\|_2.
 \end{align}
Now, under assumptions of Proposition  \ref{prop:settingK} 
 from \eqref{eq:NORMkbound}
and \eqref{eq:LOWERboundCPNST} we get 
\begin{equation}
    \N_k^\sigma \geq 
    \frac{1}{4} C_d\ \sigma^{d(d-1)} \left(\frac{2}{\pi}\right)^{d(d-1)}.
\end{equation}
Inserting this and 
\eqref{eq:2normComparison} into  \eqref{eq:est_proof_1norm} we have  
\begin{equation}
    \| \F_k^\sigma - \frac{\tilde\F_\sigma}{\N^\sigma_k}  \|_1 \leq 
    \frac{\frac{3}{4} C_d\,\sigma^{d(d-1)}\, 
      e^{-\frac{1}{4} \frac{\ep^2}{\sigma^2}}}{\frac{1}{4} C_d\ \sigma^{d(d-1)} \left(\frac{2}{\pi}\right)^{d(d-1)}}
      =3 \left(\frac{2}{\pi}\right)^{d(d-1)}
      e^{-\frac{1}{4} \frac{\ep^2}{\sigma^2}}
\end{equation}
Now, using Lemma \ref{lem:pos_function_norm_1}
we obtain the required bound.

% We proceed analogously as before by relating $\F^\sigma_k$ with  $\tilde{F}^\sigma_k$:
% \begin{equation}
%     \|\F^\sigma_k\|_2=\frac{1}{\N^\sigma_k}\|\tilde{\F}^\sigma_k\|_2 \leq \frac{2}{\N^\sigma}\left(\|\tilde{\F}^\sigma\|_2 +\|\tilde{\F}^\sigma_k -\tilde{\F}^\sigma_k\|_2  \right)\ ,
% \end{equation}
% were the inequality follows from \eqref{eq:NORMkbound} (applied to the denominator) and triangle inequality for $\|\cdot\|_2$ (applied to the numerator). Next, using Proposition \ref{prop:uppBOUNDsecondNORMunitary} ($\|\tilde{\F}^\sigma\|_2\leq \N^\sigma$) and again employing \eqref{eq:2normComparison} ($\|\tilde{\F}^\sigma_k -\tilde{\F}^\sigma_k\|_2\leq \N^\sigma/2$), we finally obtain $\|\F^\sigma_k\|_2\leq 3$, which concludes the proof.
\end{proof}

\bibliography{refsdesign}

\section{Appendix}

In the appendix we provide technical details not included in the main text. Specifically, in Part \ref{app:spectralGAPdacay} we prove the relation between expander norms and restricted gaps on projective unitary group. This allows us to complete the proof of Theorem \ref{thm:Varju}. In Part \ref{app:ClassFUNC} we give the proof of Lemma \ref{lem:PolynDegree}. In the rest of the Appendix we technical results important for the construction of a  polynomial approximation of Dirac delta on $\U(d)$. Part \ref{app:torusEST} contains proofs of certain properties of of periodized Gaussian in $\T^d$ (and its extension to the unitary group $\UU(d)$). The latter Part \ref{app:polyCONVERG} presents estimates for the convergence of the polynomial truncation of this function in suitable norms. Finally, in Part \ref{app:auxTechnical} we gather auxiliary results and facts that are used in ealier sections of the Appendix. 

\subsection{Proof of Theorem \ref{thm:Varju}}\label{app:spectralGAPdacay}
Before we use the results of \cite{Varju2013} we need to introduce a couple of concepts from representation theory (we refer the reader to \cite{HallGroups} for the comprehensive introduction to representation theory of semisimple Lie groups and Lie algebras). Let $G$ be a compact semisimple Lie group and let $\Pi$ be a representation of $G$ in a finite-dimensional Hilbert space $\K$. Let $\pi$ be the associated representation of the Lie algebra of $G$ denoted by $\mathfrak{g}$. Let $\mathfrak{t}$ be a Lie algebra of the maximal torus in $G$. Weights $\alpha\in\mathfrak{t}$  encode joint eigenvalues of elements   $X\in\mathfrak{t}$  in irreducible representations of $G$. In other words for every weight $\alpha$ there exist a representation $\pi$ and a vector $\ket{\psi_\alpha}$ such that for all $X\in\mathfrak{t}$ 
\begin{equation}
    \pi(X)\ket{\psi_\alpha}=\ii \langle \alpha, X \rangle \ket{\psi_\alpha}\ ,
\end{equation}
where $\langle \cdot,\cdot \rangle$ is a non-degenerate inner product in $\mathfrak{g}$ induced by the Killing form. Recall that irreducible representations of compact connected semisimple Lie Groups are finite-dimensional labelled by the so-called highest weights $\lambda$, which are weights that satisfy some additional technical properties. Moreover, for every  finite-dimensional representation $\K$ of $G$ it is possible to chose a basis consisting of weight vectors. 

In \cite{Varju2013} the author considers the decay of spectral gap of $T_{\nu}$ for semisimple compact Lie groups, $G$. In particular he focuses on restriction, $T_\nu|_{\K_r}$, of the operator $T_\nu$ to the space $\K_r$ defined by
\begin{equation}\label{eq:defKr}
    \K_r =\bigoplus_{\lambda:\ 0<\|\lambda\|_G\leq r } \K_\lambda\ ,
\end{equation}
where $\K_\lambda$ denotes the irreducible representation of highest weight $\lambda$ and $\|\lambda\|_G=\sqrt{\langle\lambda,\lambda\rangle}$ is the norm of $\lambda$ induced by the Killing form in $\mathfrak{g}$. In what follows we will be extensivelly using the notation $\K_\lambda \subset_{\scriptscriptstyle G} \K$ to denote the situation in which irreducible representation $\K_\lambda$ of $G$ appears in in the decomposition of $\K$ onto irreducible components.

The spectral gap of $T_\nu|_{\K_r}$ is then defined as 
\begin{equation}
\mathrm{gap}_{r}(G,\nu)=1-\|T_\nu|_{\K_r}\|_{\infty}\ .
\end{equation} 
In this setting the following theorem holds 
\begin{Theorem}[Theorem 6 in \cite{Varju2013}]\label{Varju-original}
For every semisimple compact connected Lie group $G$, there are numbers $c$ and $r_0$ such that the following holds. Let $\nu$ be an arbitrary probability measure on $G$. Then
\begin{equation}\label{varju-slow}
    \mathrm{gap}_{r}(G,\nu)\geq c\,\mathrm{gap}_{r_0}(G,\nu)\log^{-2}(r)\ .
\end{equation}
\end{Theorem}
\noindent  In what follows we apply the above result for $G=\U(d)$. In order to see how Theorem \ref{thm:Varju} follows from  Theorem \ref{Varju-original} let us first note that inequality \eqref{varju-slow} can be written as 
\begin{equation}\label{eq:VARJUbound}
\|T_\nu|_{\K_r}\|_{\infty}\leq 1-\frac{1-\|T_\nu|_{\K_{r_0}}\|_{\infty}}{(1/c)\log^2(r)}\ .
\end{equation} 
We thus aim to find the relation between $r$ and $t$. To this end we first notice that representation $\Pi^{1,1}(\U)\coloneqq U\otimes \bar U$ decomposes as the direct sum of the adjoint representation $\mathrm{Ad}$ and the trivial representation. Furthermore the adjoint representation is the faithful representation of $\U(d)$. Thus every irreducible representation $\mathcal{H}_\lambda$ will appear in the decomposition of $(\Pi^{1,1})^{\otimes t}$ into irreducible components for sufficiently large $t$. The same can be said about the decomposition of $\H_t$ (on which $\U(d)$ acts via its regular representation) into irreducible representations of $\U(d)$ . This follows from considerations given in Section \ref{sec:mixingOP} and the fact that representations $\Pi_A(U)= U^{\ot t}\ot \bar{U}^{\otimes t}$ is equivalent to $\Pi_B(U)= \left(\Pi^{1,1}(U)\right)^{\otimes t}$. Using this observation we conclude that there exist $t_0$ such that for all $\lambda$ satisfying $0<\|\lambda\|_G\leq r_0$  we have $\K_\lambda \subset_{\scriptscriptstyle \U(d)}\H_{t_0}$ (for some suitable $t_0$ depending on $r_0$). Therefore we have 
\begin{equation}\label{eq:NORMr0}
    \|T_\nu|_{\K_{r_0}}\|_\infty \leq  \|T_\nu|_{\H_{t_0}} -T_\mu|_{\H_{t_0}}\|_\infty =\|T_{\nu,t_0} -T_{\mu,t_0}\|_\infty \ ,
\end{equation}
where $T_\mu|_{\H_{t_0}}$ is a projector onto a trivial representation (space of constant functions) in $\H_t$.   The first inequality in the above equation comes from the fact that (c.f. Eq.\eqref{eq:defKr}) $\K_r$ does not contain trivial representations of $\U(d)$. The equality follows form Proposition \ref{prop:gaps} in Section \ref{sec:mixingOP}.

In the second step we observe that all irreducible representations appearing in $(\Pi^{1,1})^{\ot t}$ have highest weights of magnitude $\|\lambda\|_G\leq a t$, where $a$ depens only on $d$. This follows from the fact that weight vectors associated with the representation $(\Pi^{1,1})^{\ot t}$ can be chosen to have tensor product structure i.e.
\begin{equation}
    \ket{\psi_\beta}= \ket{\psi_{\alpha_1}}\otimes \ket{\psi_{\alpha_2}}\ot\ldots\ot \ket{\psi_{\alpha_t}}\ ,
\end{equation}
where $\ket{\psi_{\alpha_2}}\in\C^d\ot\C^d$ is a weight vector in  the representation $\Pi^{1,1}$ . Consequently, all weights $\beta$ occouring in the  representation $(\Pi^{1,1})^{\ot t}$ (or equivalently the regular representation of $\U(d)$ restricted to the function space $\H_t$) are sums of weights associated with $\Pi^{1,1}$: $\beta=\sum_{i=1}^t\alpha_i$. Using triangle inequality we obtain that for all $\lambda$ such that $K_\lambda \subset_{\scriptscriptstyle \U(d)}\H_{t}$
\begin{equation}\label{eq:weightNORM}
   \|\lambda\|_{\U(d)}\leq t \max_\alpha \|\alpha\|_{\U(d)}\ ,
\end{equation}
where the optimization is over all weights $\alpha$ appearing in $\Pi^{1,1}$. It is clear that $\max_\alpha \|\alpha\|_{\U(d)}=a$ depends only on the dimension $d$. 

Inequalities \eqref{eq:NORMr0} and \eqref{eq:weightNORM} allow us to employ Eq.\eqref{eq:VARJUbound} in the context of expander norms. First, from \eqref{eq:weightNORM} it follows that  every nontrivial representation $\K_\lambda\subset_{\scriptscriptstyle \U(d)}\H_t$ satisfies also $\K_\lambda\subset_{\scriptscriptstyle \U(d)}\K_{at}$ and consequently
\begin{equation}
\left\|T_{\nu,t}-T_{\mu,t} \right\|_\infty  \leq  \|T_\nu|_{\K_{a t}}\|_\infty \ .
\end{equation}
Applying to the above first \eqref{eq:VARJUbound} and then \eqref{eq:NORMr0} gives
\begin{equation}
    \left\|T_{\nu,t}-T_{\mu,t} \right\|_\infty \leq 1-\frac{1-\|T_\nu|_{\K_{r_0}}\|_{\infty}}{(1/c)\log^2(at)} \leq 1-\frac{1-\left\|T_{\nu,t_0}-T_{\mu,t_0} \right\|_\infty}{(1/c)\log^2(at)} \ .
\end{equation}
It is now easy to verify that for $t\geq 2$ $\log(at)\leq C\log(t)$ for $C=1+\frac{\log(a)}{\log(2)}$ and therefore by setting $B=C^2/c$ we obtain the claimed result: for $t\geq t_0$ ($t_0$ is defined above Eq.\eqref{eq:NORMr0}) and for any probability measure $\nu$ on $\U(d)$ we have 
\begin{equation}
    \left\|T_{\nu,t}-T_{\mu,t} \right\|_\infty  \leq 1-\frac{1-\left\|T_{\nu,t_0}-T_{\mu,t_0} \right\|_\infty}{B\log^2(t)} \ .
\end{equation}

\subsection{Proof of Lemma \ref{lem:PolynDegree}} \label{app:ClassFUNC}

In order to prove Lemma \ref{lem:PolynDegree} we introduce the following notation. Let $\n=(n_1,\ldots,n_d)\in \mathbb{Z}^d$ be such that $\sum_{i=1}^dn_i=0$ and $\sum_{i=1}^d|n_i|=2k$. Let $S_d$ denote the symmetric group on $d$ symbols. We define action of $\sigma \in S_d$ on $\n\in\mathbb{Z}^d$ by:
\begin{gather}
\sigma(\n)=(n_{\sigma(1)},\ldots,n_{\sigma(d)}).
\end{gather}
Recall that a partition of a set $X$ is a set of non-empty subsets of $X$ such that every element $x\in X$ is in exactly one of these subsets. Let $P_d$ be a partition of $\n$, where we  viewe $\n$ as an $d$-element set. Let $|P_d|$ be the number of subsets in partition $P_d$. By the abuse of notation we define $P_d(\n)$ to be a vector whose first $|P_d|$ coefficients are sums of elements in the corresponding partition subsets and the remaining coefficients are equal to zero. For example, for $d=3$ and partition $P_3$ of $\n$ given by $\{\{n_1,n_2\},n_3\}$ we have $P_3(\n)=(n_1+n_2,n_3,0)$. Finally, recall that
$\mathcal{H}_k=\mathbf{Span}\left\{\tr\left( A_tU^{\otimes t}\otimes \bar{U}^{\otimes t}\right):A\in \End (\mathbb{C}^{d^{2t}}), t\in\left\{0,1,\ldots,k\right\}\right\}$. Our aim is to prove the following Lemma whose corollary is Lemma \ref{lem:PolynDegree}.

\begin{Lemma}\label{lemtrace}
$\mathcal{H}_k$  contains all functions of the form:
\begin{gather}
\sum_{\sigma\in{S_d}}e^{i\sigma(\n)\bphi},
\end{gather}
where $\n$ is such that $\sum_{i=1}^{d}n_i=0$ and $\sum_{i=1}^d|n_i|=2t$, $0\leq t\leq k$. In particular $\mathcal{F}^\sigma_k\in\H_k$.
\end{Lemma}

 In order to prove the above Lemma we first observe that if $C_m$ is the natural matrix representation of the cyclic permutation $(1,\ldots,m)$ on $\mathbb{C}^{d^m}=\mathbb{C}^d\otimes\ldots\otimes \mathbb{C}^d$ then 
 
\begin{gather}
\tr \left(C_m A_1\otimes\ldots\otimes A_m\right)=\tr\left(A_1\cdot\ldots\cdot A_m\right)
\end{gather}
Thus if we choose $A_t\in \End (\mathbb{C}^{d^{2t}})$ to be
\begin{gather}
A_t=C_{n_1}\otimes C_{n_2}\otimes \ldots\otimes C_{n_\alpha}\otimes C_{n_{\alpha+1}}\otimes \ldots\otimes C_{n_d},
\end{gather}
where $\sum_{i=1}^\alpha n_i=t=\sum_{i=\alpha+1}^d n_i$ the resulting function is:

\begin{gather}
f_{\n}(\bphi):=\tr\left( A_tU^{\otimes t}\otimes \bar{U}^{\otimes t}\right)=\tr U^{n_1}\tr U^{n_2}\ldots\tr U^{n_\alpha}\tr \bar U^{n_{\alpha+1}}\ldots\tr \bar U^{n_{d}}=\\
=\prod_{j=1}^\alpha\left(\sum_{k=1}^de^{in_j\phi_k}\right)\prod_{j=\alpha+1}^d\left(\sum_{k=1}^de^{-in_j\phi_k}\right),
\end{gather}
which can be reduced to
\begin{gather}\label{fn}
f_{\n}(\bphi)=\sum_{P_d}\alpha(P_d)\sum_{\sigma\in S_d}e^{i\sigma(P_d(\n))\bphi},\,\n=(n_1,n_2,\ldots,n_\alpha,-n_{\alpha+1},\ldots,-n_d).
\end{gather}
We are now ready to give a proof of Lemma \ref{lemtrace}. 
\begin{proof}
By direct calculations one checks that Lemma \ref{lemtrace} is valid for $k=0$ and $k=1$. We follow by induction, i.e. we assume Lemma \ref{lemtrace} is valid for $k\geq 1$ and our aim is to show that this implies its validity for $k+1$. Let $\n$ be such that $\sum_{i=1}^{d}n_i=0$ and $\sum_{i=1}^d|n_i|=2(k+1)$. Consider the function $f_{\n}\in\mathcal{H}_{k+1}$. The summand corresponding to the full partition $P_d$ is:
\begin{gather}\label{sumand}  
\sum_{\sigma\in S_d}e^{i\sigma(\n)\bphi}
\end{gather}
Thus to show that (\ref{sumand}) belongs to $\mathcal{H}_{k+1}$ it suffices to show that for other partitions $P_d$ the corresponding summands appearing in (\ref{fn}) are either in $\mathcal{H}_{k}$ or can be easily proved to be in $\mathcal{H}_{k+1}$. The latter happens only when $\sum|P_d(\n)_i|=2k+2$ that is $P_d$ respects division of $\n$ into two parts $\{n_1,\ldots,n_{\alpha}\}$ and $\{n_{\alpha+1},\ldots, n_d\}$. Note however that such $P_d(\n)$ has at least one zero entry. We can perform the same reasoning with function $f_{P_d(\n)}$ and select vectors $P_{d^\prime}(P_d(\n))$ satisfying $\sum|P_{d^\prime}(P_d(\n))_i|=2k+2$. They necessary have at least two zero coefficients. Following this procedure we always arrive at the vector  $\n_1=(k+1,0,\ldots,0,-k-1)$ whose corresponding function $f_{\n_1}\in\mathcal{H}_{k+1}$ is 
\begin{gather}\label{sum2}
\sum_{\sigma\in S_d}e^{i\sigma(\n_1)\bphi} + a_1,
\end{gather}
where $a_1\in \mathcal{H}_k$. Thus $\sum_{\sigma\in S_d}e^{i\sigma(\n_1)\bphi}\in\mathcal{H}_k$. Next, reversing the path of the above reasoning we obtain the desired result. 
\end{proof}

\subsection{Estimates for Gaussian functions on a torus}\label{app:torusEST}

In this part of the Appendix we complete proofs of  Lemmas  \ref{lem:loverBOUNDnormCONST} and \ref{lem:TailUnitary} from Section \ref{sec:constr}.
For reader's convenience we  collect here concepts and notations that will be used in the reminder of the Appendix. We will use  the following measures defined on $\T^d$:
\begin{equation}
    \label{eq:phix}
\dt\bphi= \dt \varphi_1 \ldots \dt \varphi_d\ ,\    
    \dt \mu(\bphi)=
    \frac{1}{(2\pi)^d d!}
    \prod_{1\leq i<j\leq d} |e^{i \varphi_i}-e^{i \varphi_j}|^2     \dt \bphi \ .
    \end{equation}
We also introduce the counterparts of these measures on $\R^d$:
\begin{equation}\label{muG-def}
    \dt\x = \dt x_1 \ldots \dt x_d\ ,\  \dt \mu_G(\x)=
    \frac{1}{(2\pi)^d d!}
    \prod_{1\leq i<j\leq d} (x_i-x_j)^2     \dt \x\ .
\end{equation}
Recall the functions on $\T^d$ that were used to define polynomial approximation to the Dirac delta at $\h{I}\in\U(d)$:
\begin{equation}
    \gauss^\sigma (\bphi)= 
    \frac{1}{(\sqrt{2\pi}\sigma)^d} e^{-\frac12 \frac{\bphi^2}{\sigma^2 }}\ ,\  
    \gauss_p^\sigma (\bphi)=  \frac{1}{(\sqrt{2\pi}\sigma)^d}
     \sum_{\k\in\Z^d} e^{-\frac12 \frac{(\bphi+2 \pi \k)^2}{\sigma^2 }}\ .
\end{equation}
Poisson summation formula implies
\begin{equation}\label{eq:Gauss appendix}
    \gauss^{\sigma}_{p}(\bphi)= \frac{1}{(2\pi)^d}\sum_{\n\in \Z^d} e^{-\frac{1}{2} \n^2 \sigma^2} e^{-i\n \bphi} \ . 
\end{equation}
The truncation of this function to trigonometric polynomials of degree at most $k$ is given by
\begin{equation}
    \gauss^{\sigma}_{p,k}(\bphi)= \frac{1}{(2\pi)^d}\sum_{\n\in S_k} e^{-\frac{1}{2} \n^2 \sigma^2} e^{-i\n \bphi}\ ,
\end{equation}
where $S_k=\{\n\ |\ |\n|_1 \leq k \}$.  The above define class functions on $\UU(d)$ which were used in Section \ref{sec:constr}:
\begin{equation}\label{eq:classFUNCTIONdefAPP}
F^\sigma (U) =  f^\sigma_p(\mathrm{Eig}(U))\ , \ F^\sigma_k (U) =  f^\sigma_{p,k}(\mathrm{Eig}(U))\ .
\end{equation}
As we explained in Section \ref{sec:constr}, the Weyl integration formula for class functions in $\UU(D)$
\begin{equation}\label{eq:WaylIntegrationAPP}
\int_{\UU(d)}\dt\mu(U)F(U)=\int_{\T^d} \dt \mu(\bphi) f(\bphi) ,
\end{equation}
enables to compute numerous quantities relevant for functions $F^\sigma$,$F^\sigma_k$ solely in terms of the functions $\gauss^\sigma_p$, $\gauss^{\sigma}_{p,k}$ defined on $\mathbb{T}^d$.

In what follows it will be expedient to bound integrals with respect to the measure $\dt\mu(\bphi)$ by integrals with respect to the measure $\dt\mu_G(\x)$. To this end we establish the following technical result.
\begin{Lemma}\label{lem:IvsJ} 
    Then for any non-negative integralble function $s:\T^d\rightarrow\R$ we have:
\begin{equation}\label{eq:intUPPERgauss} 
    \int_A \dt \mu(\bphi) s(\bphi) 
    \leq 
    \int_A  \dt \mu_G(\x) s(\x)\quad \text{ \rm{for all }} A\subset \T^d
\end{equation}
and 
\begin{equation}\label{eq:intLOWERgauss} 
    \int_A \dt \mu(\bphi) s(\bphi) 
    \geq \left(\frac{2}{\pi}\right)^{d(d-1)}
    \int_A \dt \mu_G(\x) s(\x) \quad \text{ \rm{for all }} A\subset \left[-\frac{\pi}{2},\frac{\pi}{2}\right]^{\times d}.
\end{equation}
\end{Lemma}
\begin{proof}
 The proof follows from comparing densities of measures $\dt \mu(\bphi)$ and $\dt \mu_G(\x)$. We first observe $\dt\mu(\bphi)=\prod_{1\leq i<j\leq d} 4  \sin^2(\frac{\varphi_i-\varphi_j}{2}) \dt\bphi$.  The upper bound in  Eq.\eqref{eq:intUPPERgauss} follows then from the estimate
 \begin{equation}
 \prod_{1\leq i<j\leq d} 4  \sin^2(\frac{\varphi_i-\varphi_j}{2}) \leq \prod_{1\leq i<j\leq d} (\varphi_i-\varphi_j)^2\ ,
 \end{equation}
 where  we used the fact that for all $x\in\R$ $|\sin(x)|\leq|x|$.  On the other hand using the bound 
$\frac{2}{\pi} |x| \leq |\sin(x)|$ 
valid for $x\in[-\pi/2,\pi/2]$ 
we get that
for all  $\varphi_i \in [-\pi/2,\pi/2]$
\begin{equation}
\prod_{1\leq i<j\leq d} 4  \sin^2(\frac{\varphi_i-\varphi_j}{2}) \leq \prod_{1\leq i<j\leq d}  \frac{2}{\pi^2} (\varphi_i-\varphi_j)^2
= \left(\frac{2}{\pi}\right)^{d(d-1)} \prod_{1\leq i<j\leq d} (\varphi_i-\varphi_j)^2\ ,
\end{equation}
which completes the proof of Eq.\eqref{eq:intLOWERgauss}. 
\end{proof}

\subsubsection{Lower bound on normalization constant $\N^\sigma$}
In the sketch of the proof of Lemma \ref{lem:loverBOUNDnormCONST} in Section \ref{sec:constr} we argued that $\N^\sigma\geq\int_{\T^d} \dt\mu(\bphi) \gauss^\sigma(\bphi)$. The following result completes the proof of Eq.\eqref{eq:UnitGroupNormBound}.

\begin{Lemma}\label{lem:denominator}
For $\sigma$  satisfying $\sigma\leq \frac{\pi}{4 \sqrt{d}}$ we have 
\begin{equation}
\int_{\T^d} \dt \mu(\bphi)  \gauss^\sigma(\bphi)
\geq \frac12 C_d\,\sigma^{d(d-1)}
\left(\frac{2}{\pi}\right)^{d(d-1)}\ ,
\end{equation}
where $C_d = \frac{\prod_{k=1}^d k! }{(2\pi)^d d!}$.
\end{Lemma}
\begin{proof}

We first employ Eq.\eqref{eq:intLOWERgauss} from Lemma \ref{lem:IvsJ} to $A=[-\pi/2,\pi/2]^{\times d}$  and function $\gauss^\sigma$ obtaining 
 \begin{equation}
 \label{eq:HvsG}
     \int_{\T^d} \dt \mu(\bphi) \gauss^\sigma(\bphi)
\geq \left(\frac{2}{\pi}\right)^{d(d-1)}
\int_{[-\pi/2,\pi/2]^{\times d}} \dt \mu_G(\x) \gauss^\sigma(\x)\ .
 \end{equation}
 The integral over $[-\pi/2,\pi/2]^{\times d}$ can be further decomposed as \begin{equation}\label{eq:simpleEQ}
    \int_{[-\pi/2,\pi/2]^{\times d}} \dt \mu_G(\x) \gauss^\sigma(\x)= 
\int_{\R^d} \dt \mu_G(\x) \gauss^\sigma(\x)
-
\int_{([-\pi/2,\pi/2]^{\times d})^c} \dt \mu_G(\x) \gauss^\sigma(\x)
\ .
\end{equation}
Both terms in the above expression can be connected to respectivelly: \emph{statistical sum} of GUE ensemble, and the tail probability of the maximal eigenvalue distribution of a GUE random matrix. Concretely, employing estimates from Lemma \ref{lem:gsigmabound} (given in Part \ref{app:auxTechnical}) we get that for $\sigma\leq \pi/(4 \sqrt{d})$
\begin{equation}\label{eq:simpleINEQ}
\int_{\R^d} \dt \mu_G(\x) \gauss^\sigma(\x)=C_d \sigma^{d(d-1)}\ ,\ 
\int_{([-\pi/2,\pi/2]^{\times d})^c}\dt \mu_G(\x) \gauss^\sigma(\x) \leq \frac12 C_d\,\sigma^{d(d-1)}\ ,
\end{equation}
where $C_d=\frac{\prod_{k=1}^d k! }{(2\pi)^d d!}$. Inserting the above expressions into  expressions into Eq.\eqref{eq:simpleINEQ} and Eq.\eqref{eq:simpleEQ} concludes the proof.
\end{proof}

\subsubsection{Upper bound for integral of $F^\sigma$ over the complement of a ball} 
The  following result is an effective restatement of Lemma \ref{lem:TailUnitary} from Section  \ref{sec:constr}. The equivalence of both results follows form the Weyl integration formula which implies $\int_{B(I,r)}\dt \mu(U) F^\sigma (U)=\int_{B_\infty(0,r)}\dt \mu(\bphi) f^\sigma_p (\bphi)$.

\begin{Lemma}[Restatement of Lemma \ref{lem:TailUnitary} from Section  \ref{sec:constr}]\label{lem:numerator} For $\sigma$ and $r$ satisfying $\sigma\leq \frac{r}{4 \sqrt{d}}$, $r\leq 2/3$ we have 
\begin{equation}
\int_{B_\infty(0,r)^c} \dt\mu(\bphi) \gauss^\sigma_p(\bphi)
  \leq  \frac{3}{2}
C_d\,\sigma^{d(d-1)}\,e^{-\frac{1}{4} \frac{\ep^2}{\sigma^2}}\ ,
\end{equation}
where $B_\infty(0,r)=\left\{\bphi\in\T^d\ |\ |\varphi_i|\leq r \right\}$ and  $C_d=\frac{\prod_{k=1}^d k! }{(2\pi)^d d!}$.
\end{Lemma}
\begin{proof}
From the positivity of $\gauss^\sigma$ and the definition of $\gauss^\sigma_p$ it follows that
\begin{equation}
    \int_{B_\infty(0,r)^c} \gauss^\sigma_p(\bphi)
\dt \mu(\bphi) = 
\int_{B_\infty(0,r)^c} \gauss^\sigma(\bphi)
\dt \mu(\bphi) + \| \gauss^\sigma_p - \gauss^\sigma \|_1\ .
\end{equation}
Using Lemma \ref{lem:IvsJ} allows us to write an estimate $\int_{B(0,r)^c} \gauss^\sigma(\bphi)
\dt \mu(\bphi)\leq \int_{([-r,r]^{\times d})^c} \gauss^\sigma(\x)
\dt \mu_G(\x) $, relating the integrals on the torus with integrals on $\R^d$. Next, using Lemma \ref{lem:norm-period-bound} 
(bounding $\| \gauss^\sigma_p - \gauss^\sigma \|_1$)
and Lemma \ref{lem:gsigmabound} (bounding the tail of  Gaussian integral) we obtain 
that for 
$\sigma\leq 1/(6\sqrt{d})$   and 
$\sigma\leq r/(4 \sqrt{d})$  
\begin{equation}
\int_{B_\infty(0,r)^c} \gauss^\sigma_p(\bphi)
\dt \mu(\bphi)
\leq 
\frac12 C_d\,\sigma^{d(d-1)}\, 
      e^{-\frac14 \frac{\ep^2}{\sigma^2}}
      +C_d\, \sigma^{d(d-1)}
  e^{-\frac18 \frac{\pi^2}{\sigma^2
}}\leq \frac32 
C_d\,\sigma^{d(d-1)}\, 
      e^{-\frac14 \frac{\ep^2}{\sigma^2}}\ ,
\end{equation}
where in the last inequality we assumed $r\leq \sqrt{2} \pi$.
Finally, all the assumptions made on $\sigma$ and $r$ in the above reasoning will be satisfied provided $\sigma\leq \frac{r}{4 \sqrt{d}}$, $r\leq 2/3$.
\end{proof}

\begin{Lemma}[$L^1$-norm difference between $\gauss^\sigma$ and $\gauss^\sigma_p$] \label{lem:norm-period-bound}
For  $\sigma$ satisfying $\sigma\leq \frac{1}{6 \sqrt{d}}$ we have
\begin{equation}
    \|\gauss^\sigma - \gauss^\sigma_p\|_1
    \leq 
    C_d\,\sigma^{d(d-1)}\,  e^{-\frac18 \frac{\pi^2}{\sigma^2
}}\ ,
\end{equation}
where $C_d=\frac{\prod_{k=1}^d k!}{(2\pi)^d d!}$.
\end{Lemma}
\begin{proof}As $f^\sigma_p$ is the periodization of $f^\sigma$ we have
\begin{equation}
    \gauss_p^\sigma(\bphi) =
    \gauss^\sigma(\bphi) + 
     \sum_{\k\in\Z^d, \k\not=0} 
     \gauss^\sigma(\bphi+2\pi \k )\ .
\end{equation}
Combining this with the positivity of both $f^\sigma$  we get
\begin{equation}
\label{eq:wrap}
\|\gauss^\sigma - \gauss^\sigma_p\|_1
=
\sum_{\k\in\Z^d, \k\not=0} 
     \int_{\T^d}\dt \mu(\bphi)
     \gauss^\sigma(\bphi+2\pi \k )\ .
\end{equation}
We can write it as 
\begin{equation}
\label{eq:wrap2}
\|\gauss^\sigma - \gauss^\sigma_p\|_1
=
\sum_{m=1}^\infty \sum_{\k\in A_m}
     \int_{\T^d}\dt \mu(\bphi)
     \gauss^\sigma(\bphi+2\pi \k ) ,
\end{equation}
where $A_m=\{\k: -m\leq k_i\leq m, i=1,\ldots,d ; \exists_i |k_i|=m  \}$ (i.e. $A_m$ is the boundary of a regular $d$-cube of in $\Z^d$ centered at the origin and having diameter $2m$ in $L^\infty$ distance). 
The function 
$\gauss^\sigma_\k(\bphi):=\gauss^\sigma(\bphi+2\pi\k)$
is centered around point $-2\pi \k$ and  it is positive. Thus we can apply Lemma 
\ref{lem:IvsJ} to get 
\begin{equation}
    \int_{\T^d}\dt \mu(\bphi) \gauss^\sigma(\bphi+2\pi \k )  \leq 
    \int_{\T^d}\dt \mu_G(\x)\gauss^\sigma(\x+2\pi \k ) \ .
\end{equation}
Now, for $\k\in A_m$, $\x\in\T^d$  we have  $|\x+2\pi\k|_\infty\geq \pi(2m-1)$ (where we used the $L^\infty$ norm: $|\x|_\infty = \max_{i}|x_i|$). Therefore we can bound
\begin{equation}
 \int_{\T^d}\dt \mu(\bphi) \gauss^\sigma(\bphi+2\pi \k ) \leq  
    \int_{B_\infty(0,\pi(2 m -1) )^c}\dt \mu_G(\x) \gauss^\sigma(\x) \,
\end{equation}
where $B_{\infty}(\y,r)=\{\x\in\R^d\ |\ |x_i-y_i|\leq r,i=1,\ldots,d \}$.
Next we apply tail bound from Lemma 
\ref{lem:gsigmabound} which gives
\begin{equation}
    \int_{B(0,2 \pi m-\pi)^c}\dt \mu_G(\x)
    \gauss^\sigma(\x) \leq 
\frac12 C_d
      \,\sigma^{d(d-1)}\, 
      e^{-\frac14 \frac{\pi^2(2m-1)^2}{\sigma^2}}.
\end{equation}
which holds for $    \sigma\leq \frac{\pi(2m-1)}{2 \sqrt{d}}$ so that it is enough to assume that $\sigma \leq \frac{\pi}{2 \sqrt{d}}$.   Going back to Eq.\eqref{eq:wrap} we get 
\begin{equation}
    \|\gauss^\sigma- \gauss^\sigma_p\|_1 
    \leq 
    \frac12 C_d \,\sigma^{d(d-1)}\,  \sum_{m=1}^\infty \sum_{\k\in A_m} 
e^{-\frac14 \frac{\pi^2 (2 m-1)^2}{\sigma^2 }} = 
\frac12 C_d \,\sigma^{d(d-1)}\, 
\sum_{m=1}^\infty |A_m| 
e^{-\frac14 \frac{\pi^2 (2 m-1)^2}{\sigma^2 }}.
\end{equation}
Next we use an estimate for the number of elements in $A_m$: $|A_m|\leq 2d (2m+1)^{d-1}$ which gives \eqref{eq:tail-m1}
\begin{equation}
\label{eq:tail-m1}
    \|\gauss^\sigma- \gauss^\sigma_p\|_1\leq d\, C_d \,\sigma^{d(d-1)}\, 
    \sum_{m=1}^\infty (2m+1)^{d-1} 
e^{-\frac14 \frac{\pi^2 (2 m-1)^2}{\sigma^2 }} \leq d\, C_d \,\sigma^{d(d-1)}\, 3^{d-1} \sum_{m=1}^\infty (2m-1)^{d-1} 
e^{-\frac14 \frac{\pi^2 (2 m-1)^2}{\sigma^2 }}\ ,
\end{equation}
where in the last inequality we used $\left(\frac{2m+1}{2m-1}\right)^{d-1}\leq 3^{d-1}$ for $m\geq1$. 
In order to bound the series 
\begin{equation}
\sum_{m=1}^\infty (2m-1)^{d-1}
e^{-\frac14 \frac{\pi^2 (2 m-1)^2}{\sigma^2 }},    
\end{equation}
we present each summand as a product of two terms 
\begin{eqnarray}
\label{eq:series-gauss-poly}
\sum_{m=1}^\infty (2m-1)^{d-1}
e^{-\frac14 \frac{(2\pi m-\pi)^2}{\sigma^2 }}  \leq 
\max_{m\geq 1}\left(
(2m-1)^{d-1}
e^{-\frac18 \frac{\pi^2 (2m-1)^2}{\sigma^2 }}\right)
\sum_{m=1}^\infty 
e^{-\frac18 \frac{\pi^2 (2m-1)^2}{\sigma^2 }}
\end{eqnarray}
For the first term we have \begin{equation}
\label{eq:maxpoly-m}
    \max_{m\geq 1}\left(
(2m-1)^{d-1}
e^{-\frac18 \frac{\pi^2 (2m-1)^2}{\sigma^2 }}\right)
\leq 
\max_{x\geq 0}\left(
x^{d-1}
e^{-\frac18 \frac{\pi^2 \cred  x^2\blk}{\sigma^2 }}\right) \leq \biggl(\frac{4(d-1) \sigma^2}{ e \pi^2}  \biggr)^\frac{d-1}{2}\ ,
\end{equation}
where we used Proposition  \ref{lem:max-gauss-poly}.
Next we consider the second term of 
\eqref{eq:series-gauss-poly}. We have \begin{eqnarray}
\label{eq:bound-series-gauss-m}
&&\sum_{m=1}^\infty 
e^{-\frac18 \frac{\pi^2 (2m-1)^2}{\sigma^2 }}
\leq 
2 e^{-\frac18 \frac{\pi^2}{\sigma^2
}}+
\sum_{m=3}^\infty 
e^{-\frac18 \frac{\pi^2 (2m-2)^2}{\sigma^2 }} =
2 e^{-\frac18 \frac{\pi^2}{\sigma^2
}}+
\sum_{l=2}^\infty 
e^{-\frac18\frac{\pi^2 l^2}{\sigma^2 }}\leq 
\nonumber \\
&&
\leq
2 e^{-\frac18 \frac{\pi^2}{\sigma^2
}} + 
\int_1^\infty\dt x e^{-\frac18\frac{\pi^2 x^2}{\sigma^2 }} \dt x\leq
2 e^{-\frac18 \frac{\pi^2}{\sigma^2
}} + \frac{\sigma \sqrt{8}}{\sqrt{\pi}}e^{-\frac18 \frac{\pi^2}{\sigma^2
}}\leq 
3 e^{-\frac18 \frac{\pi^2}{\sigma^2
}},
\end{eqnarray}
where to bound the integral we have used the Hoeffding bound of Proposition \ref{prop:Hoeffding} and in  the last step we assumed that $\sigma\leq \sqrt{\pi/8}$.
Inserting \eqref{eq:bound-series-gauss-m} and \eqref{eq:maxpoly-m} into 
\eqref{eq:series-gauss-poly} we get 
\begin{equation}
    \sum_{m=1}^\infty (2m+1)^{d-1}
e^{-\frac14 \frac{\pi^2(2 m-1)^2}{\sigma^2 }}    \leq 
3 \biggl(\frac{4(d-1) \sigma^2}{ e \pi^2}  \biggr)^\frac{d-1}{2}\,
 e^{-\frac18 \frac{\pi^2}{\sigma^2
}}. 
\end{equation}
Coming back to \eqref{eq:tail-m1} we obtain
\begin{equation}
\|\gauss^\sigma- \gauss^\sigma_p\|_1 \leq 
d C_d \,\sigma^{d(d-1)}3^{d-1}\,
3 \biggl(\frac{4(d-1) \sigma^2 }{ e \pi^2}  \biggr)^\frac{d-1}{2}\,
 e^{-\frac18 \frac{\pi^2}{\sigma^2}}=
B_d(\sigma) C_d\,\sigma^{d(d-1)}\,e^{-\frac18 \frac{\pi^2}{\sigma^2}} \ ,
\end{equation}
where $B_d(\sigma)=d3^d\left(\frac{4(d-1)\sigma^2}{e\pi^2}\right)^{\frac{d-1}{2}}$. It is easy to see that $B_d(\sigma)$ is non-increasing function of $\sigma$ for $d\geq2$ and that $B_d(1/(6\sqrt{d}))\leq 1$. Therefore for $\sigma\leq\frac{1}{6\sqrt{d}}$ we have
\begin{equation}
        \|\gauss^\sigma- \gauss^\sigma_p\|_1\leq C_d
\,\sigma^{d(d-1)}\,
        e^{-\frac18 \frac{\pi^2}{\sigma^2}}\ .
\end{equation}
During the above considerations we have assumed
    \begin{equation}
        \sigma \leq \sqrt{\pi/8},\quad 
        \sigma\leq \frac{\pi}{2 \sqrt{d}}\ ,
    \end{equation}
and therefore $\sigma\leq\frac{1}{6\sqrt{d}}$ is the the strongest constraint we had to impose. 
\end{proof}

\subsection{Polynomial truncation of periodized Gaussian}\label{app:polyCONVERG}

This part of the Appendix is devoted to technical results on approximation of $\tilde{\F}^\sigma$ by its polynomial truncation $\tilde{F}^\sigma_k$. As explained in Section \ref{sec:constr} the crucial number of interest, the $L^2$-distance $\|\tilde{\F}^\sigma-\tilde{\F}^\sigma_k\|_2$ can be upper bounded by $L^2$-distance between the corresponding functions on $\T^d$ (computed with respect to the measure $\dt\mu(\bphi)$), $\|\gauss^{\sigma}_{p,k} -  \gauss^{\sigma}_{p}\|_2$. This is is the language that will be used in what follows.

\begin{Lemma}[Needed to prove Lemma \ref{lem:L2normBOUND} from Section \ref{sec:constr}]
\label{lem:normfkf} Assume that 
\begin{equation}
\label{eq:normfkf1}
    k\geq %d+\frac{\sqrt{d}}{\sigma}\ ,
     \frac{d}{\sigma} 
     \ , 
    \quad  \sigma\leq \frac12.
\end{equation}
Let $\gauss^\sigma_p$, $\gauss^\sigma_{p,k}$ be functions on $\T^d$ defined in Eq.\eqref{eq:Gauss appendix}. We have the following estimate on their $L^2$ distance 
 \begin{equation}
\label{eq:polyKbound2}
    \|\gauss^{\sigma}_{p,k} -  \gauss^{\sigma}_{p}\|_2 \leq
   \frac{1}{(2\pi)^d d!}
    \sqrt{\frac{2^{d(d-1)} \pi^\frac{d}{2} \sqrt{8 \pi} d!}{\Gamma(\frac{d}{2})}}\,\frac{e^{-\frac14 (\frac{k}{\sqrt{d}}-\sqrt{d})^2\sigma^2}}{\sigma}\ .  
 \end{equation}
This can be further estimated as 
\begin{equation}\label{eq:polyKbound}
    \|\gauss^{\sigma}_{p,k} -  \gauss^{\sigma}_{p}\|_2 \leq
   10\,C_d \frac{e^{-\frac14 (\frac{k}{\sqrt{d}}-\sqrt{d})^2\sigma^2}}{\sigma}\ ,
\end{equation}
where $C_d=\frac{\prod_{k=1}^d k! }{(2\pi)^d d!}$.
% \end{equation}
\end{Lemma}
\begin{proof}
Using the definition of of the $L^2$ norm on $\T^d$ we get
\begin{equation}
\|\gauss^{\sigma}_{p,k} -  \gauss^{\sigma}_{p}\|_2^2
    =\int_{\T^d}\dt \mu(\bphi)
    |\gauss^{\sigma}_{p,k}(\bphi) -  \gauss^{\sigma}_{p}(\bphi)|^2 
    \leq
    \frac{2^{d(d-1)}}{(2 \pi)^d d!}\int_{\T^d}\dt \bphi 
    |\gauss^{\sigma}_{p,k}(\bphi) -  \gauss^{\sigma}_{p}(\bphi)|^2\ ,
\end{equation}
where we have used definition of $\dt \mu(\bphi)$ (cf. Eq.\eqref{eq:phix}) and the inequality $\prod_{1\leq i<j\leq d}|e^{i \varphi_i} -e^{i \varphi_j}|^2\leq2^{d(d-1)}$. By expanding $\gauss^{\sigma}_{p,k}(\bphi) -  \gauss^{\sigma}_{p}(\bphi)$ in a trigonometric series and using ortogonality of functions $\left\{\exp(\ii\n\bphi)\right\}_{\n\in\Z^d}$ on $\T^d$ (equipped with  measure $\dt\bphi$) we obtain
\begin{equation}\label{eq:fourierEXPbound}
    \|\gauss^{\sigma}_{p,k} -  \gauss^{\sigma}_{p}\|_2^2 \leq\frac{2^{d(d-1)}}{(2 \pi)^d d!}\int_{\T^d}\dt \bphi
    \biggl|  \frac{1}{(2\pi)^d}\sum_{\n:\ |\n|_1>k } e^{-\frac{1}{2} \n^2 \sigma^2} e^{-i\n \bphi}  \biggr|^2 = \frac{1}{(2 \pi)^d}\frac{2^{d(d-1)}}{(2\pi)^{d}d!} 
    \sum_{\n:\ |\n|_1>k }
    e^{-\n^2 \sigma^2}\ ,
\end{equation}
where the summation in second and third expression above is over $\n\in\Z^d$ corresponding to trygonometric polynomials of degree exceeding $k$ (i.e  $\sum_{i=1}^d|n_i|>k$). Using the well-known bound $\n^2\geq|\n|_1^2/d$, valid for $\n\in\Z^d$,  we obtain 
\begin{equation}\label{eq:L1toL2}
    \sum_{\n:\ |\n|_1>k }
    e^{-\n^2 \sigma^2} \leq  \sum_{\n:\ \n^2>\frac{k^2}{d} }
    e^{-\n^2 \sigma^2}\ .
\end{equation}
The second sum in the above expression can be bounded using Lemma \ref{lem:sum-ball} given below. For this result if follows that for  %\sout{$k>d+\frac{\sqrt{d}}{\sigma}$} 
$k\geq d/\sigma$
and $\sigma\leq 1/2$ we have
\begin{equation}
\sum_{\n:\ \n^2>\frac{k^2}{d} }
e^{-\sigma^2 \n^2}\leq
 \frac{\pi^\frac{d}{2}\sqrt{8 \pi}}{\Gamma(\frac{d}{2})\sigma^2}
    \, e^{-\frac12 (\frac{k}{\sqrt{d}}-\sqrt{d})^2\sigma^2}\ .
\end{equation}
Inserting this inequality to \eqref{eq:L1toL2} and using the result in \eqref{eq:fourierEXPbound} yields 
\begin{equation}
    \|\gauss^{\sigma}_{p,k} -  \gauss^{\sigma}_{p}\|_2 \leq 
   C_d A_d \frac{e^{-\frac14 (\frac{k}{\sqrt{d}}-\sqrt{d})^2\sigma^2}}{\sigma}\ ,
\end{equation}
where $C_d=\frac{\prod_{k=1}^d k! }{(2\pi)^d d!}$ and
\begin{equation}
    A_d=
    \frac{1}{\prod_{k=1}^d k!}\sqrt{\frac{2^{d(d-1)} \pi^\frac{d}{2} \sqrt{8 \pi} d!}{\Gamma(\frac{d}{2})}}\ .
\end{equation}
This proves the estimate \eqref{eq:polyKbound2}. In Proposition \ref{prop:AdUPPERbound} we prove that for all $d$ we have $A_d\leq 10$. Making use of this result proves Eq.\eqref{eq:polyKbound}. 
\end{proof}

Before we prove the upper bound for 
the series appearing 
in 
Eq.\eqref{eq:L1toL2},
we shall present bound on $\|f_p^\sigma\|_2$.
\begin{Lemma}
\label{lem:2norm}
For $\sigma \leq 1/4 $ 
we have 
    \begin{equation}
    \|\gauss^{\sigma}_{p}\|_2\leq 
    \frac{1}{d!(2\pi)^{d}}
    \,c(d)\, \sigma^{-\frac{d}{2}},
\end{equation}
where 
\begin{equation}
    c(d)=2^{\frac{d(d+1)}{2}}\,  \sqrt{d!}\,.
\end{equation}
\end{Lemma}
\begin{proof}
Proceeding as in Lemma \ref{lem:normfkf} we get
\begin{equation}
    \|\gauss^{\sigma}_{p}\|_2^2=\frac{1}{d!(2\pi)^d}\int_{\T^d}\dt \bphi \prod_{i>j}|e^{i \varphi_i} -e^{i \varphi_j}|^2
    \frac{1}{(2\pi)^{2d}}
    \biggl|\sum_{\n\in\Z^d}\hat{f}^\sigma(\bphi) e^{-i \n \bphi}\biggr|^2\leq 
    \frac{2^{d(d-1)}}{d!(2\pi)^{2d}} \sum_{\n\in\Z^d}e^{-\n^2 \sigma^2}.
\end{equation}
Subsequently we have 
\begin{eqnarray}
\sum_{\n\in\Z^d}e^{-\n^2 \sigma^2}=\biggl(\,\sum_{l=-\infty}^\infty e^{-l^2 \sigma^2}\biggr)^d.
\end{eqnarray}
We then bound the sum by integral
\begin{equation}
    \sum_{l=-\infty}^\infty e^{-l^2 \sigma^2}
    \leq 1 + \int_{-\infty}^\infty  e^{-x^2 \sigma^2}\, \dt x=1 + \frac{\sqrt{\pi}}{\sigma}\leq \frac{2}{\sigma},
\end{equation}
where the last inequality holds for $\sigma\leq 2-\sqrt{\pi}$ (so we may take $\sigma\leq 1/4$). Thus we obtain
\begin{equation}
    \|\gauss^{\sigma}_{p}\|_2^2 \leq 
    \frac{2^{d(d-1)}}{d!(2\pi)^{2d}}\, \frac{2^d}{\sigma^d}= 
    \left(\frac{1}{d!(2\pi)^{d}}\right)^2\, 
    d!\, 2^{d(d+1)}\,\frac{1}{\sigma^d}
\end{equation}

% Using now estimate of Lemma \ref{lem:factorial} and bounding $2\leq e\leq 4$ we obtain
% \begin{equation}
%     c(d)\leq \cred 2^{\frac{d^2}{2}+1} \blk d^{\frac{d}{2}+\frac14}.
% \end{equation}
\end{proof}

We now prove the result on the upper bound of the series appearing in Eq.\eqref{eq:L1toL2}.
\begin{Lemma}
\label{lem:sum-ball}
For  $r\geq\frac{\sqrt{d}}{\sigma}$ and $\sigma\leq 1/2$ 
we have 
\begin{equation}
\sum_{\n:\ \n^2>r^2} 
e^{-\sigma^2 \n^2}\leq
\frac{\pi^\frac{d}{2}}{\Gamma(\frac{d}{2})}  \frac{2\sqrt{2 \pi}}{\sigma^2}
    \, e^{-\frac12 (r-\sqrt{d})^2\sigma^2}\ ,
\end{equation}
where the summation  is over $\n\in\Z^d$ corresponding satisfying $\n^2\geq r^2$.
\end{Lemma}
\begin{proof}
We write
\begin{equation}
\sum_{\n^2 \geq r^2}
e^{-\sigma^2 \n^2}= \sum_{l=0}^\infty
\sum_{\n \in D_l}
e^{-\sigma^2 \n^2} 
\leq \sum_{l=0}^\infty
|D_l|
e^{-\sigma^2 (r+l\sqrt{d})^2}\ , 
\end{equation}
with $|D_l|$ being the number of elements of the set $D_l$ defined by   $D_l=\{\n\in\Z^d\ |\ (r+l\sqrt{d})^2\leq \n^2\leq (r+ l\sqrt{d}+\sqrt{d})^2\}$.
To evaluate  $|D_l|$  we note that
\begin{equation}
    |D_l|\leq |B_l|\leq 
    \mathrm{vol}(B(r+l\sqrt{d}+\sqrt{d}))\ ,
\end{equation}
where $B_l=\{\n\in\Z^d\| \ \n^2\leq (r+l \sqrt{d})^2\}$, and $B(r)$ denotes the Euclidean ball in$\R^d$ of radius $r$. Indeed, consider all the points $\n$ contained in $B_l$. These are all points $\n$ contained in the Euclidean ball of radius $r+l \sqrt{d}$. We now note, that each such point is in the middle of the unit $d$-dimensional cube
containing only this ball. The diameter of such cube is $\sqrt{d}$, hence if we enlarge the radius of the ball by $\sqrt{d}$, the number of the points in the $(r+l \sqrt{d})$-ball will be 
no smaller than the volume of the 
enlarged ball. 
Since  
\begin{equation}
\mathrm{vol}(B(r))=c_d r^d\ \ \text{for}\ \ c_d=\frac{\pi^\frac{d}{2}}{\Gamma(\frac{d}{2})}
\end{equation}
we  have 
\begin{equation}
|D_l|\leq 
    c_d (r+ l \sqrt{d} +\sqrt{d})^d\ .
\end{equation}
Therefore we obtain the following upper bound
\begin{equation}
\label{eq:sum-n-bound}
\sum_{\n:\ \n^2>r^2}
e^{-\sigma^2 \n^2}\leq c_d 
\sum_{l=0}^\infty
(r+l\sqrt{d}+\sqrt{d})^d e^{-\sigma^2 (r+l\sqrt{d})^2}.
\end{equation}
In what follows we estimate this series by an integral. To this end we need to choose such $r$ that 
the function will be nonincreasing for $x\geq-1$. We find that 
the function $g_r(x)=(r+x\sqrt{d}+\sqrt{d})^d e^{-\frac12 (r+x\sqrt{d})^2\sigma^2}$ 
has three critical points:
\begin{equation}
    x_0=-1-\frac{r}{\sqrt{d}},\quad x_{\pm}=\frac12\left(-1 - \frac{2 r }{\sqrt{d}}\pm\sqrt{\frac{2}{\sigma^2}+1}\right).
\end{equation}
It follows that the function is nonincreasing for $x\geq x_+$, so that we need $r$ such that $x_+\leq -1$. We rewrite the inequality $x_+\leq -1$  as follows:
 \begin{equation}
     r \geq \frac{\sqrt{d}}{2}\left(1+ \sqrt{\frac{2}{\sigma^2}+1}\right)
 \end{equation}
 Note that assuming  $\sigma\leq 1/2$ we have 
 \begin{equation}
     1+ \sqrt{\frac{2}{\sigma^2}+1}\leq \frac{2}{\sigma}
 \end{equation}
 so that it is enough 
 to take 
 \begin{equation}
     r\geq
    %  \frac{\sqrt{d}}{2}+\frac{1}{\sigma}
    \frac{ \sqrt{d}}{\sigma}
    \ .
 \end{equation}
Since for $g_r(-1)=r^d e^{-(r-\sqrt{d})^2\sigma^2}> 0$ and for $x\to\infty$ it goes to zero, 
 we obtain that $g_r(x)>0$ for $x\geq-1$ and we can bound the sum by an integral
\begin{equation}
    \sum_{l=0}^\infty
(r+l\sqrt{d}+\sqrt{d})^d e^{-\sigma^2 (r+l\sqrt{d})^2}\leq \int_{-1}^\infty 
\dt x(r+x\sqrt{d}+\sqrt{d})^d e^{-\sigma^2 (r+x\sqrt{d})^2} 
=\frac{1}{\sqrt{d}}\int_{r-\sqrt{d}}^\infty\dt y(y+\sqrt{d})^d e^{-y^2\sigma^2}. 
\end{equation}
Using positivity of the integrand within the integration limits, we now bound this integral as follows
\begin{eqnarray}
\label{eq:boundh}
     &&\int_{r-\sqrt{d}}^\infty\dt y(y+\sqrt{d})^d e^{-y^2\sigma^2} =
    \int_{r-\sqrt{d}}^\infty\dt y\left((y+\sqrt{d}) e^{-\frac{y^2\sigma^2}{d}}
    \right)^d
    \leq \nonumber \\
    &&
    \label{eq:h}
    \leq\max_{y\geq r-\sqrt{d}}
    \left((y+\sqrt{d})^d e^{-\frac{y^2\sigma^2}{2}}
    \right)
    \int_{r-\sqrt{d}}^\infty\dt y(y+\sqrt{d}) e^{-\frac{y^2\sigma^2}{2}}.
    % =
    % \left(\max_{y\geq r-\sqrt{d}}
    % h(y)\right)^{d-1}
    % \int_{r-\sqrt{d}}^\infty h(y)
    % \dt y
\end{eqnarray}
Let $u(y)=(y+\sqrt{d})^de^{-\frac{y^2\sigma^2}{2}}$. The function $u$ has three critical points: minimum $y_-$ and maximum $y_+$ given by
\begin{equation}
    y_{\pm}=\frac{\sqrt{d}}{2}\left(-1 \pm\sqrt{\frac{4}{\sigma^2}+1}\right), \quad y_0=-\sqrt{d}.
\end{equation}
The maximum $y_+$ is a global maximum, and $y_+\geq0$.  Therefore 
\begin{equation}
    \max_{y\geq r-\sqrt{d}}
    u(y)\leq u(y_+)\leq y_++\sqrt{d}=
    \frac{\sqrt{d}}{2}\left(1 +\sqrt{\frac{4}{\sigma^2}+1}\right)\leq 
    \frac{2\sqrt{d}}{\sigma},
\end{equation}
where the last inequality holds for $\sigma\leq \frac32$.
The second term of the right-hand-side of inequality \eqref{eq:h} we bound  
by Hoeffding-type bound 
from Proposition \ref{prop:Hoeffding}:
\begin{equation}
    \int_{r-\sqrt{d}}^\infty(y+\sqrt{d}) e^{-\frac{y^2\sigma^2}{2}}
    \dt y\leq 
    \frac{\sqrt{2 \pi}}{\sigma} 
    e^{-\frac12 (r-\sqrt{d})^2\sigma^2}
\end{equation}
valid for positive lower limit,
i.e. for $r\geq \sqrt{d}$. Inserting this into  \eqref{eq:boundh}
we get 
\begin{equation}
    \int_{r-\sqrt{d}}^\infty(y+\sqrt{d})^d e^{-y^2\sigma^2} \dt y\leq
    \frac{2\sqrt{2 \pi}\sqrt{d}}{\sigma^2}
    \, e^{-\frac12 (r-\sqrt{d})^2\sigma^2}.
\end{equation}
which gives
\begin{equation}
    \sum_{l=0}^\infty
(r+l\sqrt{d}+\sqrt{d})^d e^{-\sigma^2 (r+l\sqrt{d})^2}
\leq \frac{2\sqrt{2 \pi}}{\sigma^2}
    \, e^{-\frac12 (r-\sqrt{d})^2\sigma^2}.
\end{equation}
Finally, inserting this into 
\eqref{eq:sum-n-bound}
we get 
\begin{equation}
\sum_{\n:\ \n^2>r^2} e^{-\sigma^2 \n^2}\leq
    c_d  \frac{2\sqrt{2 \pi}}{\sigma^2}
    \, e^{-\frac12 (r-\sqrt{d})^2\sigma^2}.
\end{equation}
This ends the proof.
    \end{proof}

% The following result justifies  Proposition \ref{prop:settingK} from Section \ref{sec:constr}. Since $\|\tilde{\F}^\sigma-\tilde{\F}^\sigma_k\|_2\leq \|\gauss^{\sigma}_{p,k} -  \gauss^{\sigma}_{p}\|_2$, setting $k$ for which $\| \gauss^{\sigma}_{p,k} -  \gauss^{\sigma}_{p}\|_2$ is smaller than both lower bound on $\N^\sigma$ (given in Lemma \ref{lem:denominator}) and the upper bound on tails of the periodized Gaussian function (given in Lemma \ref{lem:numerator}) suffices to prove Proposition \ref{prop:settingK}.

% The following result will be used to prove   Proposition \ref{prop:settingK} from Section \ref{sec:constr}. Since $\|\tilde{\F}^\sigma-\tilde{\F}^\sigma_k\|_2\leq \|\gauss^{\sigma}_{p,k} -  \gauss^{\sigma}_{p}\|_2$, setting $k$ for which $\| \gauss^{\sigma}_{p,k} -  \gauss^{\sigma}_{p}\|_2$ is smaller than both lower bound on $\N^\sigma$ (given in Lemma \ref{lem:denominator}) and the upper bound on tails of the periodized Gaussian function (given in Lemma \ref{lem:numerator}) suffices to prove Proposition \ref{prop:settingK}.

We will now state a technical lemma which will allow  to prove Proposition \ref{prop:settingK} from Section \ref{sec:constr}. 
The lemma sets $k$  for which $\| \gauss^{\sigma}_{p,k} -  \gauss^{\sigma}_{p}\|_2$ 
(actually its upper bound given in \eqref{eq:polyKbound})
is smaller than both lower bound on $\N^\sigma$ (given in Lemma \ref{lem:denominator}) and the upper bound on tails of the periodized Gaussian function (given in Lemma \ref{lem:numerator}). 
Since $\|\tilde{\F}^\sigma-\tilde{\F}^\sigma_k\|_2\leq \|\gauss^{\sigma}_{p,k} -  \gauss^{\sigma}_{p}\|_2$, this is enough  to prove Proposition \ref{prop:settingK}.

\begin{Lemma}
\label{lem:settingK-technical}
For $\sigma\leq 1/8$ and 
\begin{equation}
\label{eq:k}
    k\geq 5 \frac{d^{\frac{3}{2}}}{\sigma}  \sqrt{\frac{1}{8}\frac{r^2}{d^2\sigma^2} + \ln\frac{1}{\sigma}}
\end{equation}
we have 
\begin{equation}
    % \|\gauss_{p,k}^\sigma - \gauss^\sigma_p\|_2
    10\,C_d \frac{e^{-\frac14 (\frac{k}{\sqrt{d}}-\sqrt{d})^2\sigma^2}}{\sigma}
    \leq 
    \frac12 
    \min\left\{
    \frac{3}{2} C_d\,\sigma^{d(d-1)}\, 
      e^{-\frac{1}{4} \frac{r^2}{\sigma^2}},
      \frac12 C_d\,\sigma^{d(d-1)}
\left(\frac{2}{\pi}\right)^{d(d-1)}
      \right\}
\end{equation}
where $C_d= \frac{\prod_{k=1}^d k!}{(2\pi)^d d!}$.
\end{Lemma}
\begin{proof}
It is enough to prove a bit stronger estimate:
% \begin{equation}
%     \|\gauss_{p,k}^\sigma - \gauss^\sigma_p\|_2
%     \leq 
%     \frac14 C_d
%       \,\sigma^{d(d-1)}\, 
%       e^{-\frac{1}{4} \frac{r^2}{\sigma^2}} \left( \frac{2}{\pi}\right)^{d(d-1)}.
% \end{equation}
% From Lemma \ref{lem:normfkf}
% we have  for $k\geq d/\sigma$ and 
% $\sigma\leq 1/2$: 
% \begin{equation}
% \label{eq:want-have}
%     \|\gauss^{\sigma}_{p,k} -  \gauss^{\sigma}_{p}\|_2 \leq
%      10\,C_d \frac{e^{-\frac14 (\frac{k}{\sqrt{d}}-\sqrt{d})^2\sigma^2}}{\sigma}
%     \end{equation}
%     We thus need to find $k$ such that 
    \begin{equation}
    \label{eq:stronger}
         10\,C_d \frac{e^{-\frac14 (\frac{k}{\sqrt{d}}-\sqrt{d})^2\sigma^2}}{\sigma}
    \leq
    \frac14 C_d
      \,\sigma^{d(d-1)}\, 
      e^{-\frac{1}{4} \frac{r^2}{\sigma^2}} \left( \frac{2}{\pi}\right)^{d(d-1)}.
    \end{equation}
   We thus need to find how large should be $k$ to ensure the above inequality.
    
    Using  $\pi\leq 4$ and $40 \times 2^{d(d-1)}\leq 2^{3d^2}$ (valid for $d\geq 2$)     we get that the inequality %\eqref{eq:want-have} 
    \eqref{eq:stronger}
    is implied by the following one 
    \begin{equation}
     2^{3d^2}\,
    e^{-\frac14 \frac{(k-d)^2}{d}\sigma^2}
    \leq
    \sigma^{d^2}\, 
      e^{-\frac{1}{4} \frac{r^2}{\sigma^2}}. 
    \end{equation}
    Assuming now $\sigma \leq 1/8$ (so that $2^{-3d^2}\geq \sigma^{d^2}$), we get that %\eqref{eq:want-have} 
    \eqref{eq:stronger}
    is implied by 
    \begin{equation}
    e^{-\frac14 \frac{(k-d)^2}{d}\sigma^2}
    \leq
    \sigma^{2d^2}\, 
      e^{-\frac{1}{4} \frac{r^2}{\sigma^2}}. 
    \end{equation}
    Taking logarithm of both sides, we can rewrite this 
    as follows
    \begin{equation}
    \label{eq:want-have2}
      k \geq 2 \frac{d^{\frac{3}{2}}}{\sigma}  \sqrt{2 \ln\frac{1}{\sigma}+\frac{1}{4}\frac{r^2}{d^2\sigma^2}} +d
    \end{equation}
    For $\sigma\leq 1/8$
    we have 
    \begin{equation}
        2 \frac{d^{\frac{3}{2}}}{\sigma}  \sqrt{2 \ln\frac{1}{\sigma}+\frac{1}{4}\frac{r^2}{d^2\sigma^2}} \geq d,
    \end{equation}
    so to fulfill the inequality
    \eqref{eq:want-have2}
    (and hence 
    %\eqref{eq:want-have} 
    \eqref{eq:stronger}
    )
    it is enough that 
    \begin{equation}
        k \geq 3 \frac{d^{\frac{3}{2}}}{\sigma}  \sqrt{2 \ln\frac{1}{\sigma}+\frac{1}{4}\frac{r^2}{d^2\sigma^2}}
    \end{equation}
    Of course we can take a bit larger but better looking $k$ as in \eqref{eq:k}. 
    In the course of the proof we have assumed that $k\geq d/\sigma$,  $\sigma\leq 1/2$ and $\sigma\leq 1/8$. These constraints are fulfilled if
    $\sigma\leq 1/8 $
    and $k$ satisfies 
    \eqref{eq:k}.
    This ends the proof.
    \end{proof}

\subsection{Auxiliary technical results and facts} \label{app:auxTechnical}

\subsubsection{Estimates of integrals of Gaussian-Vandermonde on $\R^d$}

In this section we shall use the following notation 
%\begin{equation}
%\x=(x_1,\ldots,x_d), \quad
%    \vanx{\x}=\prod_{1\leq i<j\leq d} (x_i-x_j)^2, \quad \dt x_1\ldots \dt x_d = \dt \x.
%\end{equation}
\begin{equation}
    \vanx{\x}=\prod_{1\leq i<j\leq d} (x_i-x_j)^2\ .
\end{equation}
\begin{Lemma}[Upper bounds on Gaussian integrals]
\label{lem:gsigmabound} Let $f^\sigma$ be a standard Gaussian function given by Eq.\eqref{eq:gaussianDistribution} and let $([-r,r]^{\times d})^c$ be the complement of $[-r,r]^{\times d}$ in $\R^d$, and $C_d=\frac{\prod_{k=1}^d k!}{(2\pi)^d d!}$. Then we have the following upper bounds for the integrals 
    \begin{eqnarray}\label{eq:lem-gauss1}
      \int_{([-r,r]^{\times d})^c}\dt\mu_G(\x) \gauss^\sigma(\x)   &\leq &
      \frac12 C_d
      \,\sigma^{d(d-1)}\,
      e^{-\frac12 d \left(\frac{r}{\sigma \sqrt{d}}-2\right)^2}\ ,\ \mathrm{for}\ \sigma\leq\frac{r}{2 \sqrt{d}}\ , 
      \end{eqnarray}
        \begin{eqnarray}\label{eq:lem-gauss2}
      \int_{([-r,r]^{\times d})^c}\dt\mu_G(\x) \gauss^\sigma(\x)   &\leq &
      \frac12 C_d
      \,\sigma^{d(d-1)}\,
      e^{-\frac14 \frac{r^2}{\sigma^2}}\ ,\ \mathrm{for}\  \sigma\leq \frac{r}{4 \sqrt{d}}\ .
      \end{eqnarray}
   Moreover,
      \begin{eqnarray}\label{eq:lem-gauss3}
        \int_{\mathbb{R}^d}\dt \mu_G(\x) \gauss^\sigma(\x)  &= &
      C_d\,\sigma^{d(d-1)}\,.
 \end{eqnarray}
\end{Lemma}
\begin{proof} 
    By the definition of $\mu_G(\x)$ (see Eq.\eqref{muG-def}) we have
    \begin{gather}
    \label{eq:x-GUE}
       \int_{([-r,r]^{\times d})^c} \dt \mu_G(\x) \gauss^\sigma(\x) =
        \frac{1}{(2 \pi)^d d!}
        \frac{1}{(\sqrt{2 \pi}\sigma)^d}\int_{([-r,r]^{\times d})^c} \dt \x
\vanx{\x}\,
e^{-\frac12 \frac{\x^2}{\sigma^2}} 
\end{gather}
Introducing the variable $\y=\x/\sigma$ and making use of the probability measure for Gaussian Unitary Ensemble  reduces \eqref{eq:x-GUE} to
\begin{gather}
\frac{1}{(2 \pi)^d d!}
        \frac{\sigma^{d(d-1)}}{(2 \pi )^\frac{d}{2}}\int_{([-r/\sigma,r/\sigma]^{\times d})^c}
\dt \y \vanx{\y}\,
e^{-\frac12 \y^2} 
=  C_d\,  \sigma^{d(d-1)}\left( \frac{1}{\prod_{k=1}^d k!}
\frac{1}{(2 \pi )^\frac{d}{2}}\int_{([-r/\sigma,r/\sigma]^{\times d})^c} \dt \y
\vanx{\y}\,
e^{-\frac12 \y^2} \right)=
\nonumber \\
= 
C_d\,\sigma^{d(d-1)} \underset{ A\sim \mathrm{GUE}}{\mathrm{Pr}}\left(\|A\|_\infty \geq \frac{r}{\sigma}\right).\label{eq:GUEx}
\end{gather}
    Next we make use of Proposition \ref{lem:gue}
    \begin{equation}\label{GUE212}
        \underset{ A\sim \mathrm{GUE}}{\mathrm{Pr}}\left(\|A\|_\infty \geq \frac{r}{\sigma}\right) 
        =
        \underset{ A\sim \mathrm{GUE}}{\mathrm{Pr}}\left(d^{-1/2}\|A\|_\infty \geq \frac{r}{\sigma \sqrt{d}}\right)
        \leq 
        \frac12 e^{-\frac12 d \left(\frac{r}{\sigma \sqrt{d}}-2\right)^2}.
    \end{equation}
    Combining \eqref{GUE212} with \eqref{eq:GUEx} we obtain \eqref{eq:lem-gauss1}. We next note that for $
        \sigma\leq \frac{r}{4 \sqrt{d}}$
    we have 
    \begin{equation}
        \frac{r}{\sigma \sqrt{d}}-2\geq 
        \frac{r}{2\sigma \sqrt{d}},
    \end{equation}
   therefore 
    \begin{equation}
   \underset{ A\sim \mathrm{GUE}}{\mathrm{Pr}}\left(\|A\|_\infty \geq \frac{r}{\sigma}\right) \leq     \frac12 e^{- \frac{r^2}{4 \sigma^2}}.
    \end{equation}
    Inserting this estimate into 
    (\ref{eq:GUEx}) we obtain estimate \eqref{eq:lem-gauss2}. Proceeding similarly we have 
    \begin{equation}
     \int_{[-r,r]^{\times d}} \dt \mu_G(\x)\gauss^\sigma(\x) =   
     C_d  \sigma^{d(d-1)}\underset{ A\sim \mathrm{GUE}}{\mathrm{Pr}}\left(\|A\|_\infty \leq \frac{r}{\sigma}\right),
    \end{equation}
     which, taking $r\rightarrow\infty$, proves  \eqref{eq:lem-gauss3}.
\end{proof}

\subsubsection{Auxiliary facts}
In this part we provide a number of auxiliary facts that we used in previous sections.

\begin{proposition}[Equivalence of $\dproj$ and $\dproj_\diamond$]\label{prop:equivOFnorms}
Let $\h{U},\h{V}\in\U(d)$. Let $\dproj_\diamond(\h{U},\h{V}) =\|\h{U}-\h{V}\|_\diamond$ be the diamond-norm distance between $\h{U}$ and $\h{V}$. Let $\dproj(\h{U},\h{V})$ be the distance used throughput this work and defined by
\begin{equation}\label{eq:distancesRELetions}
    \dproj\left(\h{U},\h{V}\right)= \min_{\varphi\in[0,2\pi)}\left\|U - \exp(\ii \varphi) V \right\|_\infty\ .
\end{equation}
We then have the following inequalities
\begin{equation}\label{eq:equivDISTANCE2}
 \dproj\left(\h{U},\h{V}\right) \leq  \dproj_\diamond \left(\h{U},\h{V}\right) \leq 2 \dproj\left(\h{U},\h{V}\right)\ .
\end{equation}  
\end{proposition}
\begin{proof}
Without loss of generality we can assume that $\h{V}=\h{I}$. Theorem 26 in \cite{DiamondNormCharact}) states that $\dproj_\diamond\left(\h{U},\h{I}\right)$ equals the chord of the smallest arc on a circle that contains all eigenvalues $\lbrace\exp(\ii\alpha_1),\ldots,\exp(\ii \alpha_d)\rbrace$ of $U$. Note that due to unitary invariance we can assume that $\alpha_1=0$, $\alpha_{i+1} \geq \alpha_i$. It is now easy to see that, due to symmetry of the problem, in this case the optimal angle from  Eq. \eqref{eq:equivDISTANCE2} is given by $\varphi_{opt}=\alpha_d /2$. Inequalities \eqref{eq:distancesRELetions} follow from geometrical relations presented at Figure \ref{fig:Distances}.

\begin{figure}[t]
    \centering
    \includegraphics[width=0.5\textwidth]{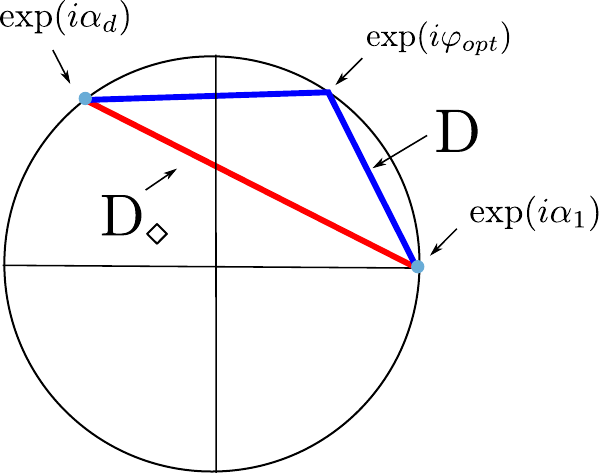}
    \caption{Geometrical presentation of the relation between $\dproj$ and $\dproj_\diamond$. }
    \label{fig:Distances}
\end{figure}

\end{proof}

\begin{proposition}[Tail bounds for spectral norm of GUE matrices \cite{SzarekBook}]
\label{lem:gue}
Let $\underset{ A\sim \mathrm{GUE}}{\mathrm{Pr}}$ be the probability measure on Hermitian matrices given by the Gaussian Unitary Ensemble. For $a>0$ we have 
\begin{equation}
\underset{ A\sim \mathrm{GUE}}{\mathrm{Pr}}\left( 
\|d^{-1/2} A\|_\infty \geq 2+a\right)
\leq \frac12 e^{-\frac12 d a^2}\ .
\end{equation}
\end{proposition}
The constant $C_d$ that we introduced in Lemma \ref{lem:loverBOUNDnormCONST} follows from 
\begin{proposition}[Mehta integral \cite{Mehta-integral}]
\label{prop:Mehta}
The following integral has the analitical form
\begin{equation}
\frac{1}{(2 \pi)^\frac{d}{2}}\int_{\R^d}\dt \x
\vanx{\x}\,
e^{-\frac12 \x^2} = \prod_{k=1}^d k!
\end{equation}
\end{proposition}

Bellow we provide upper bound on the constant $A_d$ introduced in the proof of Lemma \ref{lem:normfkf}.
\begin{proposition}\label{prop:AdUPPERbound}
For all positive integers $d$ we have
\begin{equation}
    A_d =  \frac{1}{\prod_{k=1}^d k!}\sqrt{\frac{2^{d(d-1)} \pi^\frac{d}{2} 2\sqrt{2 \pi} d!}{\Gamma(\frac{d}{2})}} \leq 10\ .
\end{equation}
\end{proposition}
\begin{proof}
We will show that for $d\geq 6$ 
$A_d$ is decreasing sequence. 
Then the proof follows by directly verifying that numerical value of $A_d$ for $d=2,3,4,5,6$ is 
smaller than $10$. 
To prove monotonicity we write:
\begin{equation}
    \frac{A_{d+1}^2}{A_d^2} = 
    \frac{1}{(d!)^2 (d+1)^2}\, 2^{2d}\, 2 \sqrt{2} \pi \frac{\Gamma(d/2)}{\Gamma(d/2+1/2)}
\end{equation}
Now, using that $\Gamma(x)$ is increasing for $x\geq 2$, 
and dropping $(d+1)^2$ in denominator, as well as bounding $d!$ from below by Stirling-type inequality  \cite{HandbookOFMAthematics}: $k!\geq \sqrt{2\pi} \sqrt{k} \left(\frac{k}{e}\right)^k$, we get 
\begin{equation}
    \frac{A_{d+1}^2}{A_d^2}  \leq    \sqrt{2}
    \left(\frac{2e}{d}\right)^{2d}
\end{equation}
The right hand side is decreasing function of $d$ and is less than 1 for $d\geq 6$,
 which proves that $A_d$ is monotonically decreasing for $d\geq 6$.
\end{proof}

Next, we recall a well known bound for a tail of Gaussian intefral.
\begin{proposition}
[Hoeffding-type bound \cite{HandbookOFMAthematics}]
\label{prop:Hoeffding}
For positive $r$ we have 
\begin{equation}
 \int_r^\infty \dt x e^{-b^2 x^2} 
 \leq \frac{\sqrt{\pi}}{b} 
 e^{-b^2 r^2}
\end{equation}
\end{proposition}

We also give a useful fact regarding maximum of function $x^k e^{-p x^2}$
 \begin{proposition}
 \label{lem:max-gauss-poly}
 For a positive $p$ we have 
 \begin{equation}
 \max_{x\geq 0} x^k e^{-p x^2}
     =
     \biggl(\frac{k}{2pe}\biggr)^\frac{k}{2}.    
 \end{equation}
 \end{proposition}

\begin{Lemma}
\label{lem:pos_function_norm_1}
Let $f,g$ be integrable  functions on a measurable space. Let $g$ be nonnegative, and $\int f=1$. The we have
\begin{equation}
    \|f\|_1\leq 1 + 2\|f-g\|_1
\end{equation}
\end{Lemma}
\begin{proof}
Since $\int f=1 $ we have 
\begin{align}
    \int g=1 - \int(f-g)\leq 
    1 + \left|\int f-g\right|
    \leq 1+\|f-g\|_1
\end{align}
Thus since $g\geq 0$ we have \begin{align}
    \|g\|_1\leq 1+\|f-g\|_1.
\end{align}
Now we write
\begin{align}
    \|f\|_1\leq \|f-g\|_1 + \|g\|_1
    \leq 1 + 2 \|f-g\|_1.
\end{align}
\end{proof}

\subsection{Bounds on the volume of balls in $\U(d)$}\label{app:volume}

In this section we present the derivation of constants $c$ and $C$ used in the formulation of Fact \ref{fact:Szarek}. The reasoning follows the argument given in \cite{Szarek98}.  

Let $(V, \|\cdot\|$ be a $d$-dimensional normed vector space and let $\rho(\cdot,\cdot)$ denote the induced metric. Let $K_R$ be a ball of radius $R$ in $V$ and $N(K_R,\rho,\epsilon)$ be a covering number of $K_R$, i.e. the smallest number of balls of radius $\epsilon$ that cover $K_R$. It is well known that \cite{Szarek98}
\begin{gather}\label{cover1}
    \left(\frac{R}{\epsilon}\right)^d\leq N(K_R,\rho,\epsilon) \leq \left(1+\frac{2R}{\epsilon}\right)^d\,.
\end{gather}
For the technical reasons we also introduce $N^\prime(K_R,\rho,\epsilon)$ as the covering number of $K_R$ with additional assumption that $\epsilon$-balls are centered at $K_R$. It is easy to see that 
\begin{gather}\label{Nprime}
    N(K_R,\rho,\epsilon)\leq N^\prime(K_R,\rho,\epsilon)\leq N(K_R,\rho,\frac{\epsilon}{2})\,.
\end{gather}
In fact, this inequality is true for any metric space $(M,\rho)$. We will also need the following lemma:
\begin{Lemma}[\cite{Szarek98}]\label{szarek1}
Let $(M_1,\rho_1)$ and $(M_2,\rho_2)$ be two metric spaces, $K\subset M_1$ and $\Phi: K\rightarrow M_2$ be a map that satisfies
\begin{gather}
    \rho_2(\Phi(x),\Phi(y))\leq L\rho_1(x,y)\,
\end{gather}
where $x,y\in K$ and $L>0$. Then for any $\epsilon>0$ we have 
\begin{gather}
    N^\prime(\Phi(K),\rho_2,L\epsilon)\leq N^\prime(K,\rho_1,\epsilon)
\end{gather}
\end{Lemma}
The vector spaces that we will explore in order to obtain bounds on the volume of balls in $\U(d)$ are Lie algebras $\mathfrak{u}(d)$ and $\mathfrak{su}(d)$ equipped with the operator norm $\|\cdot\|_\infty$. Naturally, $\mathfrak{su}(d)$ is a codimesion one subspace of $\mathfrak{u}(d)$ and the restriction of the exponential map of  $\exp:\mathfrak{u}(d)\rightarrow \UU(d)$ to $\mathfrak{su}(d)$ gives the exponential map $\exp:\mathfrak{su}(d)\rightarrow S\UU(d)$. Following \cite{Szarek98} we also note that exponential map for $\UU(d)$ is a contraction and we have 
\begin{gather}\label{contraction1}
    \|\exp(X)-\exp(Y)\|_\infty\leq \|X-Y\|_\infty\,\,\,\,X,\,Y\in\mathfrak{u}(d).
\end{gather}
\noindent On the other hand we have the following easy to prove Lemma:
\begin{Lemma}\label{expsu}
Let $K=\{X\in \mathfrak{su}(d)|\,\|X\|_\infty\leq 2\pi\}$.  Then $\exp(K)=S\UU(d)$. 
\end{Lemma}
Next, we note that groups $S\UU(d)$ and $\U(d)$ are connected be a natural projection map $\pi:S\UU(d)\rightarrow \U(d)$ that is $d$ to 1. Thus $\pi(\exp(K))=\U(d)$. Moreover, using \eqref{contraction1} and the definition of distance $\dproj(\cdot,\cdot)$ on $\U(d)$ one can easily see that $\pi\circ \exp$ is a contraction, i.e. 
\begin{gather}
\dproj\left(\pi(\exp(X)),\pi(\exp(Y))\right) \leq \|X-Y\|_\infty.
\end{gather}
Hence combining \eqref{cover1} and \eqref{Nprime} with Lemma \ref{szarek1} with $\Phi=\pi\circ \exp$, $L=1$, $K_R=K=\{X\in \mathfrak{su}(d)|\,\|X\|_\infty\leq 2\pi\}$, $\rho_1(X,Y)=\|X-Y\|_\infty$ and $\rho_2=\dproj$ we obtain 
\begin{gather}
    N^\prime\left(\U(d),\dproj,\epsilon\right)\leq N^\prime\left(K,\|\cdot\|_\infty,\epsilon\right)\leq N\left(K,\|\cdot\|_\infty,\frac{\epsilon}{2}\right)\leq \left(1+\frac{8\pi}{\epsilon}\right)^{d^2-1}
\end{gather}
Thus $\vol(\ball{\h{V}}{\ep})\left(1+\frac{8\pi}{\epsilon}\right)^{d^2-1}\geq 1$ and we conclude that $\vol(\ball{\h{V}}{\ep})\geq \left(\frac{\epsilon}{9\pi}\right)^{d^2-1}$ which is a lower bound on the $\ep$-ball volume. In order to find an upper bound we use Lemma 10 of \cite{Szarek98} that ensures that there are $\lambda>0$ and $r<\frac{\pi}{4}$ such that for $X,Y\in \mathfrak{su}(d)$ with $\|X\|_\infty<r$, $\|Y\|_\infty<r$ one has
\begin{gather}\label{separation}
    \dproj(\pi(\exp(X)),\pi(\exp(Y)))\geq \lambda \|X-Y\|_\infty\,.
\end{gather}
A careful study of the proof of Lemma 10 of \cite{Szarek98} leads to relation:
\begin{gather}
   \lambda(r)=\frac{\phi(3r)}{2}-4r\, \\
   \phi(r)=\prod_{k=1}^\infty \left(1-|1-\exp(ir/2^k)|\right),
\end{gather}
which is valid whenever $\lambda(r)>0$. In order to find optimal parameters we let now 
\begin{gather}
    K_r=\{\pi(\exp(X))\in \U(d)|\,\|X\|_\infty<r\}\,.
\end{gather}
By inequality \eqref{separation} we know that $\pi\circ\exp$ has a well defined inverse on $K_r$, denote it by $\Phi$. Using \eqref{cover1}, \eqref{Nprime} and Lemma \ref{szarek1} we get:
\begin{gather}
     \left(\frac{r\lambda}{ \epsilon}\right)^{d^2-1}\leq
    N^\prime\left(\Phi(K_r),\|\cdot\|_\infty,\frac{\epsilon}{\lambda}\right) \leq  N^\prime\left(K_r,\dproj,\epsilon\right)
\end{gather}
Thus $\vol(\ball{\h{V}}{\ep}) \left(\frac{r\lambda}{ \epsilon}\right)^{d^2-1}\leq 1$ and we conclude that $\vol(\ball{\h{V}}{\ep})\leq\left(\frac{\epsilon}{ r\lambda}\right)^{d^2-1}$  which is an upper bound on the $\ep$-ball volume. Our choice of $\lambda$ and $r$ should maximize $r\lambda(r)$. This can be done numerically and the maximal value of $r\lambda=0.011506$. Thus the upper bound on volume is $\vol(\ball{\h{V}}{\ep})\leq\left(87\epsilon\right)^{d^2-1}$.
\end{document}